\documentclass[12pt]{article} 
\usepackage{amssymb}
\usepackage{bm}
\usepackage{mathrsfs} 
\usepackage{natbib}
\usepackage{algorithm}
\usepackage{algorithmic}
\usepackage{hyperref}
\usepackage{amsthm}
\usepackage{mathtools}
\usepackage{epstopdf,xcolor}
\usepackage{bbm}
\usepackage[colorinlistoftodos,prependcaption]{todonotes}
\usepackage{dsfont}
\usepackage{pifont}
\usepackage{subcaption}
\usepackage{tikz, subcaption}
\usetikzlibrary{bayesnet, shapes, arrows, positioning}
\usepackage{fullpage}
\usepackage{enumitem}
\usepackage{graphicx}
\usepackage[page]{appendix}
\usepackage{comment}
\usepackage{multirow}
\usepackage{authblk}
\usepackage{listings}

\newtheorem{theorem}{Theorem}
\newtheorem{lemma}{Lemma}

\newtheorem{proposition}{Proposition}

\DeclarePairedDelimiter\floor{\lfloor}{\rfloor}

\newcommand{\ind}{\perp\!\!\!\!\perp}

\title{\bf \Large Bernstein Polynomial Processes\\ for Continuous Time Change Detection}

\author[1]{Dan Cunha}
\author[2]{Mark Friedl}
\author[1]{Luis Carvalho}
\affil[1]{Department of Mathematics and Statistics, Boston University}
\affil[2]{Department of Earth and Environment, Boston University}
\date{}

\begin{document}
\maketitle

\begin{abstract}
There is a lack of methodological results for continuous time change detection due to the challenges of noninformative prior specification and efficient posterior inference in this setting.  Most methodologies to date assume data are collected according to uniformly spaced time intervals. This assumption incurs bias in the continuous time setting where, \textit{a priori}, two consecutive observations measured closely in time are less likely to change than two consecutive observations that are far apart in time. Models proposed in this setting have required MCMC sampling which is not ideal. To address these issues, we derive the heterogeneous continuous time Markov chain that models change point transition probabilities noninformatively. By construction, change points under this model can be inferred efficiently using the forward backward algorithm and do not require MCMC sampling. We then develop a novel loss function for the continuous time setting, derive its Bayes estimator, and demonstrate its performance on synthetic data. A case study using time series of remotely sensed observations is then carried out on three change detection applications. To reduce falsely detected changes in this setting, we develop a semiparametric mean function that captures interannual variability due to weather in addition to trend and seasonal components.
\end{abstract}

\noindent%
{\it Keywords:}  heterogeneous continuous time Markov chain, expectation maximization, land disturbance modeling
\vfill

\newpage

\section{Introduction}
\label{sec:intro}
Change detection is a widely used modeling and inference procedure for a vast number of applications, many of which use data measured in continuous time or along a continuous axis. However, methodologies to date have largely focused on the discrete time model \citep{truong2020selective,aminikhanghahi2017survey}.  It is not hard to reason, \textit{a priori}, that the probability of change ought to be lower when observations are closer in time than when observations are further apart in time. Our objective is to develop this prior intuition in a principled noninformative way that also yields analytically tractable posterior inference of change points in continuous time models.  We will start by introducing the discrete time setting where most methodologies have been developed.  

Suppose we have conditionally independent observations $[y_i|\Theta_k,z_i=j]$ indexed by equally spaced discrete times $i=0,\dots,n$ with $j = 1,\dots, k$ model states. The likelihood distribution is conditioned on a latent change point process $\{z_i\}_{i=0}^n$ that takes on states $1,\dots,k$ as well as parameters $\Theta_k\in \mathbb{R}^{p\times k}$. When $z_i=j$, the likelihood distribution of $y_i$ is a function of the $j$th parameter vector $\bm{\theta}_j\in\mathbb{R}^{p\times 1}$. What makes $\{z_i\}_{i=0}^n$ a change point process is that $z_0=1$, $z_n=k$, and if $z_i=j$, then $z_{i+1}\in \{j,j+1\}$ with probability 1. We index observations from $i=0$ since $z_0$ always equals $1$, which simplifies notation later on. 

Notice, we could just as well have defined $k$ \textit{segment length} parameters $\{\zeta_j\}_{j=1}^k$ such that $\sum_{j=1}^k \zeta_j = n$ and $\zeta_j=i$ when $z_{i-1}=j$ and $z_i=j+1$. The majority of change detection literature postulates models as a function of segment length parameters $\bm{\zeta}$ or change point locations $\tau_j = \sum_{l=1}^j \zeta_l$ \citep{b399ed80-3ccc-30cf-af0d-87bbdad7ade6, auger1989algorithms, killick2012optimal,fearnhead2006exact,fearnhead2007line,adams2007bayesian}. While the state space and segment length models are equivalent in terms of their likelihood distributions, their corresponding Bayesian inference procedures can be quite different in terms of tractability of the posterior distribution. In the continuous time offline setting, the posterior distribution of the segment length parameters $\bm{\zeta}$ is not analytically tractable under a noninformative $\textbf{Dir}(1_k)$ prior \citep{stephens1994bayesian}. The main contribution of this paper is the development of a heterogeneous continuous time Markov chain $\pi_k(z_t=h|z_s=j) \coloneqq \pi(z_t=h|z_s=j,k)$ for times $0<s<t<1$, that is noninformative in the offline setting and enjoys analytical posterior inferences without approximation nor MCMC sampling.
\subsection{State space models versus partition models}
\cite{chib1998estimation} was the first to show the connection between state variables and segment lengths for the \textit{online model}. In this setting, segment lengths are assumed to follow a geometric distribution, $\pi_k^{(G)}(\zeta_j=i) = p_j^{i-1}(1-p_j)$ \citep{yao1984estimation,barry1993bayesian}. \cite{chib1998estimation} showed these segment lengths can be reparameterized in terms of state variables with transition probabilities $\pi_k^{(G)}(z_{i+1}=j|z_i=j) = p_j$ and $\pi_k^{(G)}(z_{i+1}=j+1|z_i=j) = (1-p_j)$. While these models are distributionally equivalent, the latter is a special case of a hidden Markov model and is equipped with methodological conveniences such as the forward backward algorithm for computing posterior expectations or analytic formulas for simulation \citep{chib1996calculating,fearnhead2006exact}.

\subsection{Offline modeling versus online modeling}
In the current work, we operate in the retrospective (offline) setting and assume a priori all change point sequences are equally likely. Please see \cite{truong2020selective} for a review of offline approaches. For example in discrete time, let $\Omega_{n,k}$ be the sample space of all change point process sequences. If $n=3$ and $k=3$, then $\Omega_{3,3} = \{\{1,1,2,3\},\{1,2,2,3\},\{1,2,3,3\}\}$ and each of these sequences is given equal prior probability in the offline setting. The corresponding prior $\pi_k(\zeta_1)$ is discrete uniform on $0,\dots,n$ and the conditional prior on the $j$th segment length $\pi_k(\zeta_j|\sum_{l=1}^{j-1} \zeta_l = i_0)$ is discrete uniform on the remaining $n-(k-j) - i_0$ positions. The $(k-j)$ term is subtracted to ensure there are enough positions for the remaining segments. The last length $\zeta_k$ is restricted by $\sum_{j=1}^k \zeta_j = n$ \citep{stephens1994bayesian}. However, the corresponding offline model for state variables $\bm{z}$ has been unexplored in both discrete and continuous time.  We develop both approaches in this work. 

\subsection{Continuous time versus discrete time}
In continuous time, the noninformative prior on segment lengths is $\bm{\zeta} \sim \textbf{Dir}(1_k)$, but the posterior distribution of this model is intractable and requires MCMC sampling \citep{stephens1994bayesian} or an approximation.  For this reason, we take a different approach.  First, noninformative priors for the state variables $\bm{z}$ are developed in discrete time and then relaxed to continuous time. We then show our model of the continuous time state variables $\{z_{t_i}\}_{i=0}^n$ is distributionally equivalent to $\bm{\zeta} \sim \textbf{Dir}(1_k)$ using the relationship $1\{z_{t_i}=j\} = 1\{\sum_{l=1}^{j-1} \zeta_l \leq t_i < \sum_{l=1}^j \zeta_l\}$. In doing so, we are able to derive the heterogeneous continuous time Markov chain $\pi_k(z_t=h|z_s=j)$ for $0\leq s<t\leq 1$ and $h\geq j$ for which exact posterior inference procedures are analytically available without MCMC nor approximation.

\subsection{Modeling environmental changes using satellite imagery}
Change detection is an important and challenging problem in remote sensing data.  Applications include detecting land cover change from both natural events (e.g., desertification, fires, etc.) and land use by humans (e.g., urbanization, agriculture, forestry, etc.) \citep{zhu2014continuous,keenan2014net,zhu2020continuous}. Due to high frequency of missing data in available remote sensing data sources and the variability in satellite periodicity, the data are collected in continuous time and as such require continuous time methods for their analysis.  In this work, we provide three case studies to demonstrate that our model generalizes across a range of situations.  The first is a deforestation example in the Rondonia region of the Amazon rainforest, the second is an agricultural land management example in the San Joaquin Valley, California, and the third is a study of vegetative drought detection in a semi-arid region in Texas.

\subsection{Paper structure}
The paper is structured as follows. In Section \ref{sec:NonInfMarg}, we develop our retrospective Bayesian change detection model in discrete time by deriving noninformative marginals $\pi_k(z_i=j)$ and extending them to their corresponding transition probabilities. In Section \ref{sec:relaxCont}, we derive the continuous time marginal distribution $\pi_k(z_t=j)$ and prove these marginals have a distributional equivalence to the noninformative prior $\bm{\zeta}\sim \textbf{Dir}(1_k)$. In Section \ref{sec:conttime}, we derive the heterogeneous continuous time Markov chain $\pi_k(z_{t}=h|z_{s}=j)$ for $0\leq s<t\leq 1$ and $h\geq j$ under the noninformative prior measure.  In Section \ref{sec:meth}, we develop a methodology for inference using expectation maximization, a novel loss function suited for the continuous time change point problem, and derive the Bayes estimator for that loss function. The Bayes estimator can be computed with the posterior moments made available by the forward backward algorithm. We then extend our model to handle outlier observations. In Section \ref{sec:sim}, we provide a simulation study and compare our method to other popular change detection methods. In Section \ref{sec:pheno}, we introduce a semiparametric model that captures interannual variation due to variation in weather and derive constraints on that function to ensure its continuity.  Finally in Section \ref{sec:casestudy}, we provide case studies of change detection examples using remote sensing, including detecting deforestation, crop management, and detecting shrub and grassland drought responses to interannual variation in weather in semi-arid regions.
\section{Noninformative Priors in Discrete Time}
\label{sec:NonInfMarg}
Let $\Omega_{n,k}$ be the sample space of all change point process sequences in discrete time with $n+1$ time points (including time $0$) and $k$ segments. The cardinality of $\Omega_{n,k}$ is $\binom{n}{k-1}$ since there are $n$ ways to choose $k-1$ changes. Now suppose we define a probability measure $\pi$ that places equal probability on all change point sequences in $\Omega_{n,k}$. Using the same counting argument above, the marginal probability $\pi_k(z_i=j)$ can be evaluated by counting the number of change point sequences that occur before $z_i = j$ and the number that occur after. That is, the number of ways to choose $j-1$ changes from $i$ time points, times the number of ways to choose $k-j$ change points from $n-i$ time points,
\begin{proposition}[Marginal noninformative prior in discrete time]
\label{prop:hypergeo}
The marginal noninformative prior on the state space $\{z_i\}_{i=0}^n$ in discrete time is hypergeometric distributed,
\begin{align*}
\pi_k(z_i=j) = \frac{\binom{n-i}{k-j}\binom{i}{j-1}}{\binom{n}{k-1}} 
\end{align*}
\end{proposition}
An example of these discrete time marginals is represented as the dotted lines in Figure \ref{fig:marginalsNdisctrans} (\textit{Left}). 
While this may be interesting, it is not immediately useful since the joint distribution of $\bm{z}$ is not a product of marginals.  But, what is clear from the definition of change point process, is that each state variable only depends on the previous state and thus the joint distribution can be factored as a product of transition probabilities. To compute the transition probabilities, start by noting that $z_{i+1}=1$ implies $z_i=1$ by the definition of change point process. Then since the transition probability is the joint distribution divided by the marginal, we have $\pi_k(z_{i+1}=1|z_{i}=1) = \pi_k(z_{i+1}=1)/\pi_k(z_{i}=1)$. In a similar fashion we can write,
\[\pi_k(z_{i+1}=2) = \big(1-\pi_k(z_{i+1}=1|z_i=1)\big)\pi_k(z_i=1) + \pi_k(z_{i+1}=2|z_i=2)\pi_k(z_i=2)\]
And solve for $\pi_k(z_{i+1}=2|z_i=2)$. Proceeding recursively, and representing all transitions as functions of marginals, we arrive at the noninformative discrete time transition probabilities for the retrospective model,
\begin{proposition}
\label{prop:trans}
The noninformative transition probabilities of the state space variables $\{z_i\}_{i=0}^n$ in discrete time are functions of the noninformative marginals,
\begin{align*}
\pi_k(z_{i} = j | z_{i-1}= j) &= \frac{\sum_{l=1}^j \pi_k(z_i = l) - \sum_{l=1}^{j-1} \pi_k(z_{i-1} = l)}{\pi_k(z_{i-1} = j)}.
\end{align*}
\end{proposition}

\begin{figure}[t]
    \begin{subfigure}{.5\textwidth}
    \includegraphics[width=1.\linewidth]{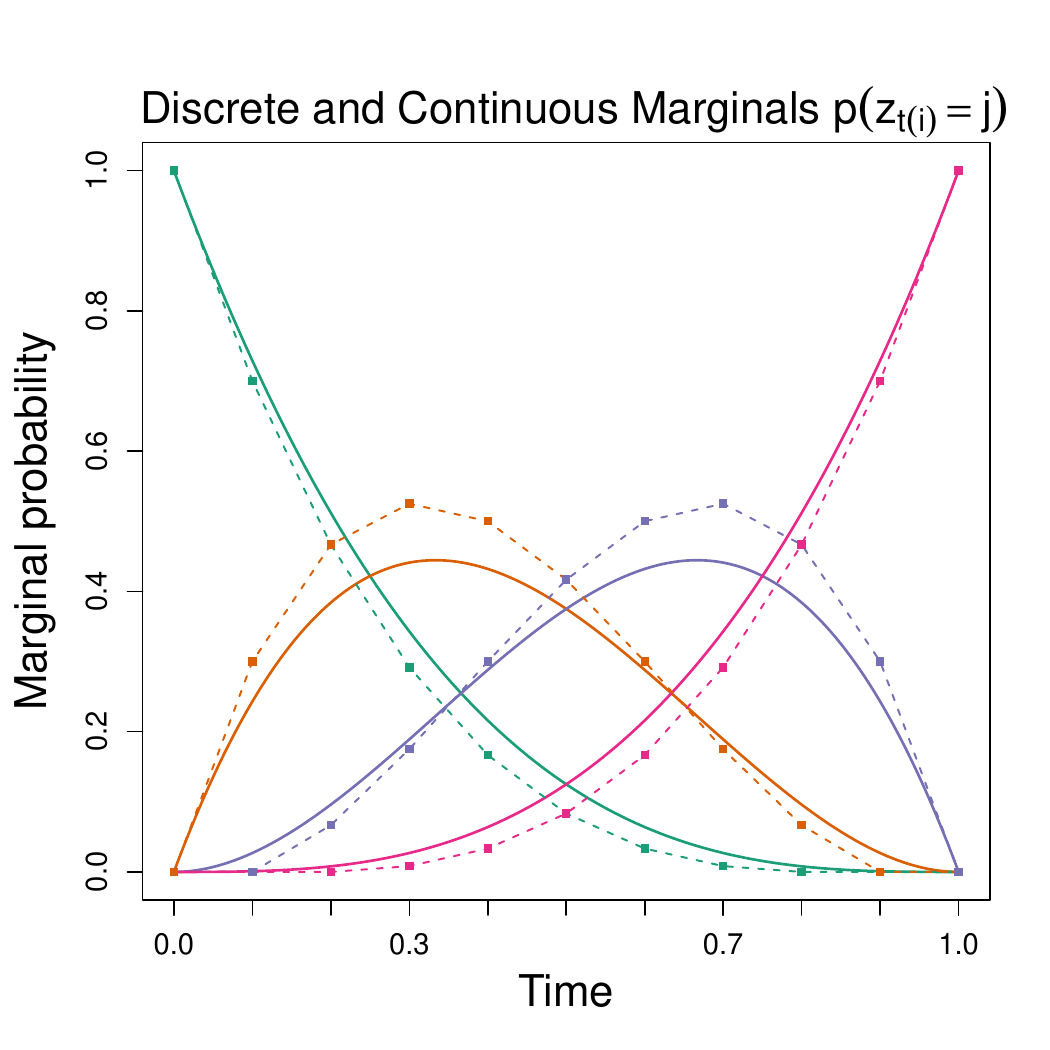}
    \end{subfigure}%
    \begin{subfigure}{.5\textwidth}
        \includegraphics[width=1.\linewidth]{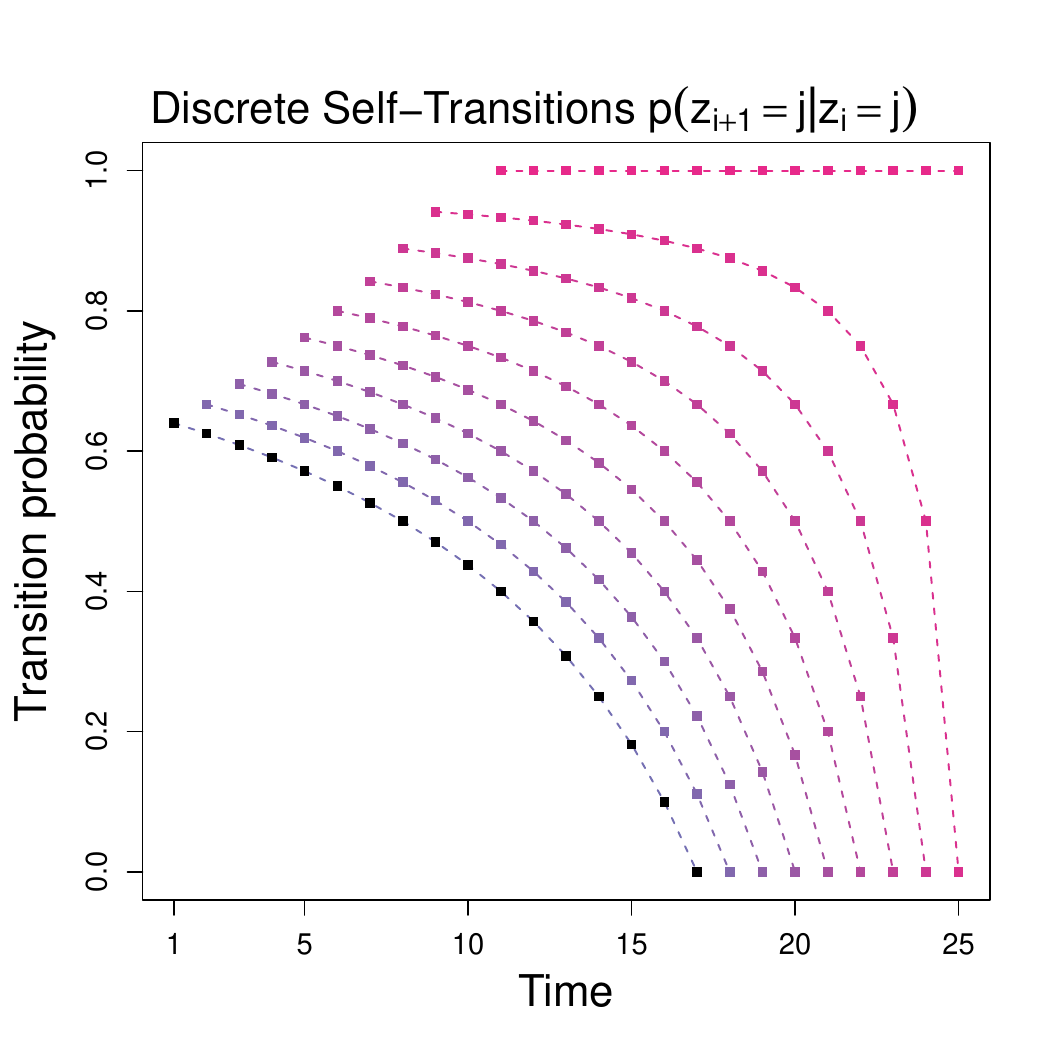}
    \end{subfigure}
\caption{(\textit{Left}) An example of noninformative discrete time hypergeometric state marginals(dotted), $n=10$ and $k=4$ with green as segment 1 and pink as segment 4. Noninformative continuous time state marginal are the corresponding solid lines. (\textit{Right}) Discrete time noninformative transition probabilities for $n=25$ and $k=10$ for illustration. The colors range from $\pi_{10}(z_{i+1}=1|z_i=1)$(purple) to  $\pi_{10}(z_{i+1}=10|z_i=10)$(pink).}
\label{fig:marginalsNdisctrans}
\end{figure}

The proposition shows that in discrete time the noninformative transition probabilities from $z_i$ to $z_{i+1}$ are a direct functional of the hypergeometric distributed marginals. Please see Figure \ref{fig:marginalsNdisctrans} (\textit{Right}) for a demonstration of these noninformative transition probabilities.  Note how each state only has non-zero probability under two conditions.  The first condition is that enough observations have been collected to identify that segment. For example, $z_2$ cannot equal $3$ since there have only been two observations. The second condition is that enough observations remain to exhaust all $k$ segments. For example, $z_{n-1}$ cannot equal $k-2$ since there is only one observation remaining and $z_n=k$ by definition.

\section{Relaxing to Continuous Time}
\label{sec:relaxCont}
To derive the continuous time Markov chain for the state variables $\{z_{t_i}\}_{i=0}^n$, it will be helpful to first derive the continuous time marginals $\pi_k(z_{t_i}=j)$ so that we can later evaluate the transitions via $\pi_k(z_{t_{i+1}}=h|z_{t_i}=j) = \pi_k(z_{t_{i+1}}=h,z_{t_i}=j)/\pi_k(z_{t_i}=j)$. We start with the discrete time, hypergeometric distributed marginals from Proposition \ref{prop:hypergeo} and derive their convergence in distribution as time becomes continuous.  To that end, define time in the interval $t\in[0,1]$.  To relate discrete time to continuous time, let continuous time $t$ be mapped to discrete time through the function $i(t) = \floor{tn}$. 
\begin{theorem}[Noninformative marginal convergence: Hypergeometric to Bernstein]
\label{thm:hgeobern}
Let $t\in[0,1]$. The limit of the marginal prior probability of state $j$ at discrete time $i(t) = \floor{tn}$, as $n\to \infty$, and after normalizing the index $i(t)/n = \floor{tn}/n$, converges to the Bernstein polynomial distribution,
\begin{align*}
\pi_k(z_{t}=j) = \binom{k-1}{j-1} t^{j-1}(1-t)^{k-j}
\end{align*}    
\end{theorem}
Please see Figure \ref{fig:marginalsNdisctrans}(\textit{Left}) for an example of these continuous-time marginals compared with the discrete time marginals from Proposition \ref{prop:hypergeo}.

\subsection{Relating state variables with segment lengths}
\label{subsec:represent}
Now that we have continuous time marginals, the final piece to solving the continuous time transition probabilities is evaluating $\pi_k(z_{t_{i+1}}=h,z_{t_i}=j)$. Recall from our discussion in Section \ref{sec:intro}, there is an equivalence relation between the two parameterizations, namely, $\bm{1}(z_{t_i} = j) = \bm{1}\Big(\sum_{l=1}^{j-1}\zeta_l \leq t_i < \sum_{l=1}^{j}\zeta_l\Big)$. Furthermore, the noninformative prior on segment lengths is $\bm{\zeta}\sim \textbf{Dir}(1_k)$ \citep{stephens1994bayesian}. Having distributional equivalence between $\bm{1}(z_{t_i} = j)$ and $\bm{1}\Big(\sum_{l=1}^{j-1}\zeta_l \leq t_i < \sum_{l=1}^{j}\zeta_l\Big)$ would provide a path for computing the joint probabilities since
\begin{equation}
\label{eq:disteq}
\pi_k(z_{t_{i+1}}=h,z_{t_i}=j) = \pi_k\bigg(\sum_{l=1}^{h-1}\zeta_l \leq t_{i+1} < \sum_{l=1}^{h}\zeta_l, \sum_{l=1}^{j-1}\zeta_l \leq t_{i} < \sum_{l=1}^{j}\zeta_l\bigg)   
\end{equation}
We prove this equivalence in the following theorem,
\begin{theorem}[Distributional equivalence in continuous time]
\label{thm:represent}
Let segment lengths $\bm{\zeta} \sim \mathit{Dir}(\bm{1}_k)$ and let $z_t$ be the random vector defined by the indicators $\bm{1}(z_{t}=j)\vcentcolon=\bm{1}(\sum_{l=1}^{j-1}\zeta_l \leq t < \sum_{l=1}^{j}\zeta_l)$ for $j=1,...,k$. Then the marginals of $z_t$ are Bernstein polynomial distributed, $\pi_k(z_{t}=j) = \binom{k-1}{j-1} (1-t)^{k-j} t^{j-1}$.
\end{theorem}
This theorem connects the duality between state variables $\{z_{t_i}\}_{i=0}^n$ and segment lengths $\{\zeta_j\}_{j=1}^k$ in the offline setting and provides the tools to evaluate the joint probability $\pi_k(z_{t_{i+1}}=h,z_{t_i}=j)$ needed for the continuous time transitions.

\section{Continuous Time Change Point Processes}
\label{sec:conttime}
We are now in a position to derive the continuous time Markov chain $\pi_k(z_{t}=h|z_{s}=j)$ for $0\leq s < t\leq 1$ and $h\geq j$. Since we have the marginals from Theorem \ref{thm:hgeobern}, we only need to evaluate the joint probability $\pi_k(z_{t}=h,z_{s}=j)$. By distributional equivalence in Theorem \ref{thm:represent}, we have the segment lengths in Equation \ref{eq:disteq} are distributed Dirichlet with parameter $1_k$. Putting these tools together to evaluate $\pi_k\big(\sum_{l=1}^{h-1}\zeta_l \leq t_{i+1} < \sum_{l=1}^{h}\zeta_l, \sum_{l=1}^{j-1}\zeta_l \leq t_{i} < \sum_{l=1}^{j}\zeta_l\big)$, we have the following,

\begin{theorem}
\label{thm:conttime}
For times $0 \leq s < t \leq 1$ and states $j = 1, \ldots, k$ and $h = j, \ldots, k$, we have the following transition probabilities $P_{jh}(s, t) \coloneq \pi_k(z_t = h \,|\, z_s = j)$:
\[
P_{jh}(s, t) = \binom{k-j}{h-j}\Bigg(1 - \frac{1 - t}{1-s}\Bigg)^{h-j} \Bigg(\frac{1-t}{1-s}\Bigg)^{k-h}
= b_{h-j, k-j}\Bigg(\frac{t-s}{1-s}\Bigg)
\]
where $b_{\nu, n}(x) = \binom{n}{\nu}x^\nu (1 - x)^{n - \nu}$ is the $\nu$-index $n$-degree Bernstein polynomial. Furthermore, these transition probabilities satisfy the Kolmogorov equations, $P_{jh}(s, t) = \sum_{l = j}^h P_{jl}(s, r)P_{lh}(r, t)$ for $0 \leq s < r < t \leq 1$.
\end{theorem}
The proof is in the appendix. There are a number of interesting corollaries from Theorem~\ref{thm:conttime}. For instance, the continuous time marginals of $z_t$ from Theorem~\ref{thm:hgeobern} are verified by plugging in $p_j(t) \coloneq P_{1,j}(0, t) = \binom{k-1}{j-1}t^{j-1}(1-t)^{k-j}$. 
Also, in particular, self-transitions are given by $P_{jj}(s, t) = [(1-t)/(1-s)]^{k-j}$, which verifies that once state $k$ is reached, the chains stays in state $k$, $P_{kk}(s, t) = 1$. 

Note, one important difference between change point modeling in discrete time versus continuous time is that more than one change can occur between two consecutive observations.  For this reason, our prior specification includes transitions from state $j$ to any of $h = j,\dots,k$. The prior distribution for the vector $\bm{z}$ of continuous time state variables is, 
\begin{equation}\label{eq:BPP}
\pi_k(\bm{z}) = \prod_{i=1}^n \prod_{j=1}^k \prod_{h=j}^k \pi_k(z_{t_i} = h \,|\, z_{t_{i-1}} = j)^{1\{z_{t_i}=h\} 1\{z_{t_{i-1}}=j\}}    
\end{equation}
For the remainder of this work, we refer to this prior as the Bernstein polynomial process or \textbf{BPP} and change detection models that assume this prior \textbf{BPP} models.

\begin{figure}[t]
    \begin{subfigure}{.5\textwidth}
    \includegraphics[width=1.\linewidth]{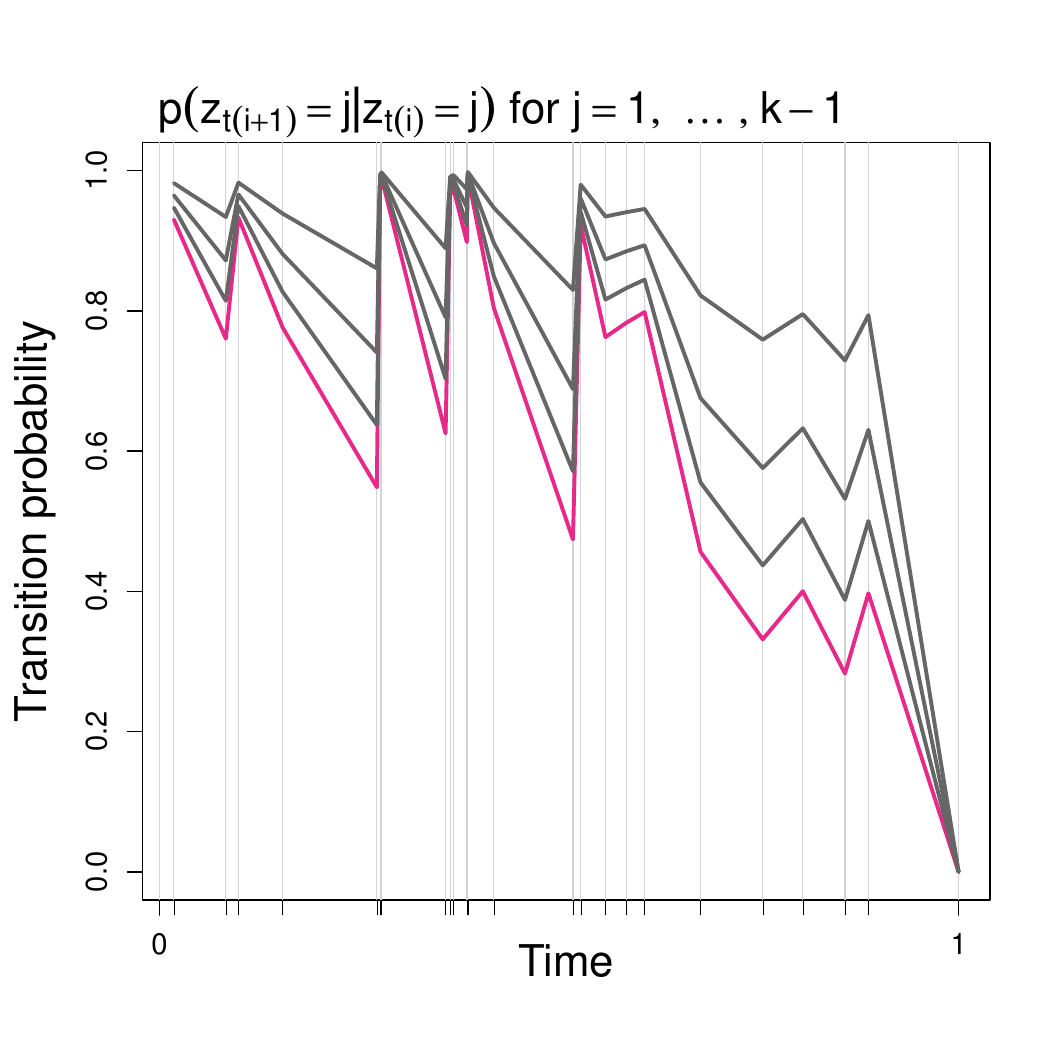}
    \end{subfigure}%
    \begin{subfigure}{.5\textwidth}
        \includegraphics[width=1.\linewidth]{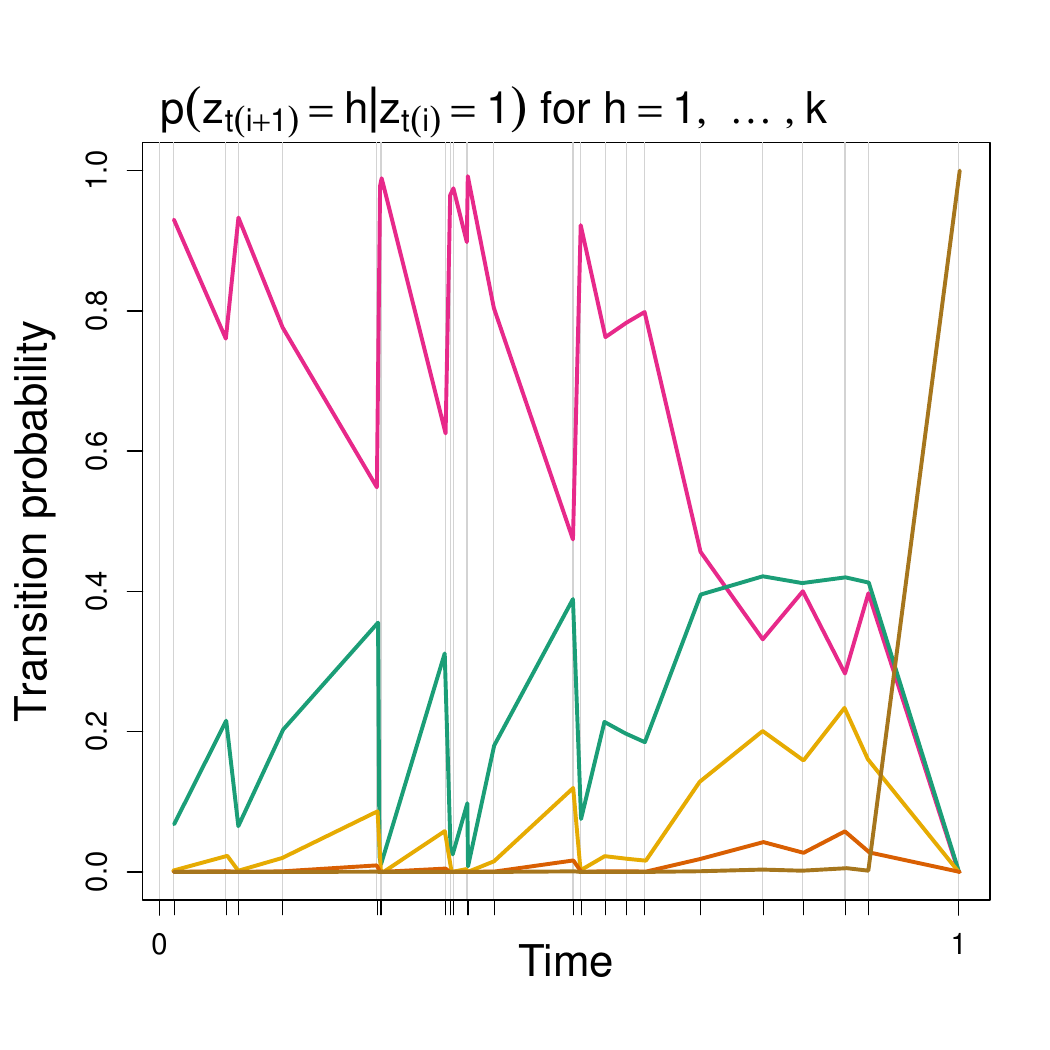}
    \end{subfigure}
\caption{An example of continuous time transitions on 25 uniformly distributed times and 5 segments. (\textit{Left}) Self-transitions $\pi_5(z_{t_{i+1}}=j|z_{t_{i}}=j)$. Pink is $j=1$, with the remaining in gray, following in increasing order of probability. (\textit{Right}) Transitions $\pi_5(z_{t_{i+1}}=h|z_{t_{i}}=1)$ for $h=1,\ldots,k-1$. Pink is $h=1$, followed by green, yellow, orange, and brown.}
\label{fig:conttime}
\end{figure}

Figure~\ref{fig:conttime} (\textit{Left}) captures important behaviors of the self-transitions $\pi_k(z_{t_{i+1}}=j|z_{t_{i}}=j)$ for an example of $25$ time points and $k=5$ segments. When observations are far apart in time, the probabilities of staying in the same state drop.  Whereas, when observations are very close in time, the probabilities jump close to 1. This behavior reflects our noninformative belief that a change is more likely to occur as more time elapses between observations. Using the same example, we can also study transitions $\pi_k(z_{t_{i+1}}=h|z_{t_{i}}=1)$ for $h=1,\ldots,k-1$ in the (\textit{Right}) of Figure \ref{fig:conttime}.  Notice the probability of staying in state $1$ dominates for the first half of time, but as time comes close to an end, the probabilities of transitioning to higher states take over.
\section{Methodology}
\label{sec:meth}
The main benefit of Theorem \ref{thm:conttime} is that the complete data likelihood and priors can now be expressed in terms of state variables $\bm{z}\sim \textbf{BPP}$ drawn from a Bernstein polynomial process, as opposed to segment lengths $\bm{\zeta}\sim \textbf{Dir}(1_k)$ which are hard to perform posterior inference on. With this state variable parameterization, efficient inference using EM is possible since the marginal and pairwise posterior expectations of $[\bm{z}|\bm{y},\Theta_k]$ can be determined using the forward-backward algorithm. Furthermore, posterior samples of the full vector $[\bm{z}|\bm{y},\Theta_k]$ can be simulated exactly within a broader Gibbs approach \citep{chib1996calculating} which does not require a Metropolis-Hastings step. Whereas, the posterior distribution $[\bm{\zeta}|\bm{y},\Theta_k]$ requires a component-wise MCMC on each $\zeta_j$ or a posterior approximation \citep{stephens1994bayesian,chib1998estimation}. We derive this simulation approach in the appendix and for remainder of the paper focus on an EM approach. 

In this section, we will characterize our model end-to-end, providing prior justifications on the change point locations, number of segments, and parameterizations. We will introduce an additional latent variable framework for modeling heavy tailed error distributions and finally propose a novel change detection loss function and derive its Bayes estimator in discrete and continuous time.
\subsection{Model}
For a fixed number of segments $k$, the model can be characterized by its prior distribution on the change point process $\bm{z}$, its likelihood distribution $f$, and its prior on the parameters $\Theta_k$,
\begin{align*}
    \prod_{i=0}^n \prod_{j=1}^k \bigg(f(y_{t_i}|\bm{\theta}_j)p(\bm{\theta}_j)\bigg)^{1\{z_{t_i}=j\}} \prod_{i=0}^{n-1} \prod_{j=1}^{k}\prod_{h=j}^k \pi_k(z_{t_{i+1}}=h|z_{t_{i}}=j)^{1\{z_{t_{i+1}}=h\} 1\{z_{t_i}=j\} }
\end{align*}
We assume the parameters are independent across segments $\bm{\theta}_j \ind \bm{\theta}_l$ for $j\ne l$, as well as conditional independence of the likelihood observations given the model $y_i \ind y_j | \bm{z},\Theta_k$. The choice of likelihood distribution $f$ also does not affect our main results; the EM and simulation procedures for the posterior distribution of $\bm{z}$ are analytically tractable regardless of the likelihood distribution. The maximization steps for $\Theta_k$ in the EM approach and the posterior distribution of $[\Theta_k|\bm{y},\bm{z}]$ in the simulation approach are the only steps for which the analytical tractability depends on the form of the likelihood. In this work, the likelihood distribution $f$ is assumed to be Gaussian.
\subsection{Robustness to outliers}
\label{sub:robust}
There is a tradeoff between outliers and change points. For example, roughly, if there is an outlier very far from its conditional expectation, there may be more evidence for placing two change points directly before and after the outlier, even though it is not a true change.  One way to address this problem is to assume a likelihood distribution with heavy tails. To that end, we introduce auxiliary variance scaling parameters drawn i.i.d. $q_{t_i} \sim \textbf{Ga}(\nu/2,\nu/2)$ such that $[y_{t_i}|\bm{\theta}_j,\sigma_k^2,q_{t_i}] \sim \mathcal{N}(\bm{\theta}_j,\sigma_k^2/q_{t_i})$ and the marginal distribution $[y_{t_i}|\bm{\theta}_j,\sigma_k^2]\sim \textbf{lst}(\bm{\theta}_j,\sigma_k^2,\nu)$ is location-scale t-distributed with $\nu$ degrees of freedom. This approach extends that of \cite{little2019statistical} to the regression and change point settings.

There are two major methodological benefits to introducing these auxiliary parameters.  The first is that within EM or a Gibbs sampling framework, using a Gaussian likelihood, the conditional posterior distribution $[q_{t_i}|y_{t_i},\bm{\theta}_j,\sigma_k^2]$ is gamma distributed making expectations and simulation straightforward. Furthermore, the maximization steps for $\bm{\theta}_j$ and $\sigma_k^2$ remain analytically tractable since the conditional likelihood is Gaussian.  The second major benefit is the marginal likelihood is t-distributed, enabling analytically tractable inference of posterior moments of $\bm{z}$ using the forward-backward algorithm. These benefits are detailed in the subsection \ref{subsec:EM}.
\subsection{Prior on number of segments}
\label{sub:pk}
We would like to discuss two different assumptions that lead to two different priors, respectively, on the number of segments. Typically, researchers choose the prior on the number of segments proportional to the volume of the space of change point sequences associated with that number of segments \citep{chib1998estimation,fearnhead2006exact,peluso2019semiparametric}. For example, in the geometric online setting described in those works, the implied prior probability on the number of segments is binomial distributed. In the offline/retrospective setting, all change point sequences are equally likely \textit{a priori}, and thus, under that reasoning, would lead to a prior on the number of segments that is proportional to their volume of change point sequences. In the following, we challenge this assumption, noting that just because we assume sequences are equally likely \textit{within} each number of segments $k$, this does not behoove us to carry that assumption \textit{across} the number of segments $k$, as is typically assumed. 

\subsubsection{Argument for using noninformative inverse volume}
On the one hand, we could continue to assume all change point sequences are equally likely \textit{across} $k=1\dots,K$ where $K$ is the maximum number of segments, but this would lead to a combinatorially increasing prior probability with respect to the number of segments, which may not be believable. For example, in the discrete time offline setting, this would lead to $\pi(k)\propto \binom{n}{k-1}$, that is, the normalizing constant found in Proposition \ref{prop:hypergeo}. In the continuous time setting, note from Theorem \ref{thm:conttime}, after removing all terms that depend on $j$ or $h$ from the transition probability $\pi_k(z_{t_{i}}=h|z_{t_{i-1}}=j)$, the remaining constant is $\big((1-t_{i})/(1-t_{i-1})\big)^k$, and thus the normalizing constant is the inverse of that value. As such, we would have
\[\pi_0(k)\propto \prod_{i=1}^n \big((1-t_{i})/(1-t_{i-1})\big)^{-k}\]
For $k=1,\dots,K$ under the assumption of equally likely change point sequences \textit{across} $k=1\dots,K$.

On the other hand, we may assume change point sequences are only equally likely \textit{within} each $k=1,\dots,K$ but that \textit{the prior on the number of segments should be noninformative} with respect to its volume of change point sequences. In this case, the prior probability of each $k$ should be inversely proportional to its volume of sequences. In the discrete time setting, this amounts to $\pi(k)\propto \binom{n}{k-1}^{-1}$ and in the continuous time setting 
\[\pi(k)\propto \prod_{i=1}^n \big((1-t_{i})/(1-t_{i-1})\big)^{k}\]
For $k=1,\dots,K$. This prior is more attractive, for example, in remote sensing applications where we expect a small number of changes on the ground. 

\subsubsection{Argument for incorporating parameter space volume}
We also extend this reasoning to the volume of the parameter space associated with each number of segments $k$. For example, suppose we are modeling changes in the mean parameters of a regression with constant variance across segments. If we assume a Gaussian prior for the mean parameters $\bm{\theta}_j|\sigma_k^2 \sim \mathcal{N}(\bm{0},\sigma_k^2 \Phi)$ and an improper prior for the variance, $p(\sigma_k^2)\propto \frac{1}{\sigma_k^2}$ on some reasonable closed interval for $\sigma_k^2$, their joint distribution is,
\[p(\Theta_k,\sigma_k^2) = \prod_{j=1}^k (2\pi\sigma_k^2)^{-\frac{p}{2}} |\Phi^{-1}|^{\frac{1}{2}} \exp{\frac{-1}{\sigma_k^2} \bm{\theta}_j^T\Phi^{-1}\bm{\theta}_j} \frac{1/\sigma_k^2}{C_{\sigma_k^2}}\]
Where $\text{dim}(\bm{\theta}_j) = p\times 1$. Removing all terms that do not depend on $\Theta_k,\sigma_k^2$, we have the volume of the parameter space is $(2\pi)^{\frac{pk}{2}}|\Phi^{-1}|^{-\frac{k}{2}}C_{\sigma_k^2}$. As $(\Theta_k,\sigma_k^2) \ind \bm{z}$ \textit{a priori}, the volume of their joint distribution is the product of their volumes. 

Under the assumption that change point sequences are equally likely across $k = 1,\dots,K$,
\begin{equation}
\label{eq:impliedpk}
\pi_0(k)\propto (2\pi)^{\frac{pk}{2}}|\Phi^{-1}|^{\frac{-k}{2}} \prod_{i=1}^n \big((1-t_{i})/(1-t_{i-1})\big)^{-k}     
\end{equation}
Whereas, under the assumption that the number of segments is noninformative with respect to their corresponding volume, we invert the normalizing constant and obtain,
\begin{equation}
\label{eq:pk}
\pi(k)\propto (2\pi)^{-\frac{pk}{2}}|\Phi^{-1}|^{\frac{k}{2}} \prod_{i=1}^n \big((1-t_{i})/(1-t_{i-1})\big)^{k}     
\end{equation}
We compare these priors from Equations \ref{eq:impliedpk} and \ref{eq:pk} in Figure \ref{fig:prior_num_changes} across $2000$ samples of time from a uniform distribution for an intercept only model. 
Furthermore, we examine the performance of both priors in the case study and find the noninformative prior on number of segments from Equation \ref{eq:pk} has largely better performance.
\begin{figure}[t]
    \centering
    \includegraphics[width=0.5\linewidth]{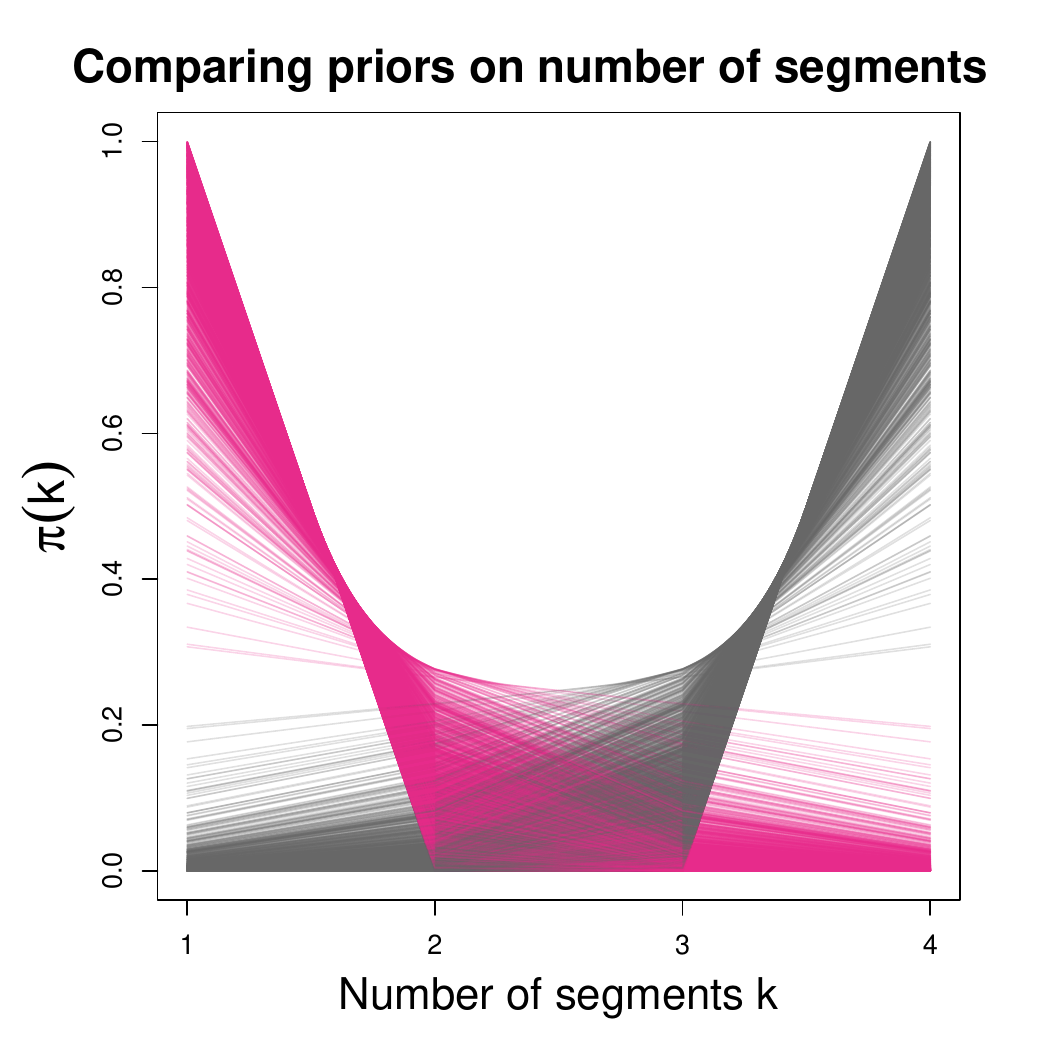}
    \caption{Two priors on number of segments are compared. Twenty time points are simulated from a uniform distribution on $[0,1]$, and $\pi(k)$ is plotted for each. The prior on $k$ assuming equally likely change point sequences across $k=1,\dots,4$ is in gray reaching maximum probability at $k=4$. The prior on $k$ assuming $k$ is noninformative with respect to the volume of its model is in pink, having maximum probability at $k=1$.}
    \label{fig:prior_num_changes}
\end{figure}
\subsection{Expectation Maximization}
\label{subsec:EM}
In many applications of change detection models, particularly in remote sensing data with trillions of time series to analyze, efficient estimation procedures are necessary. The forward backward algorithm within an Expectation Maximization (EM) framework is an efficient inference algorithm for hidden Markov models. The algorithm uses dynamic programming to evaluate the posterior moments of $\bm{z}$ conditioned on maximum \textit{a posteriori} point estimates of the parameters \citep{bishop2006pattern,dempster1977maximum}. Details for the EM algorithm are available in a wide variety of sources \citep{little2019statistical}. As noted earlier, the main contribution of this paper is the continuous time Markov chain from Theorem \ref{thm:conttime}, enabling efficient and exact inference for this model, whereas MCMC or approximate methods were required before.

Define the $Q$ function as the expectation of the log complete data likelihood with respect to the posterior distribution $[\bm{z},\bm{q}|\bm{y},\Theta_k^{(s)},\sigma_k^{2(s)}]$ for the $s$th iteration of the algorithm. Plugging in a Gaussian likelihood for the function $f$ and removing any terms that are not a factor of $\Theta_k$ or $\sigma_k^2$, we arrive at,
\begin{align*}
    Q(\Theta_k|\Theta_k^{(s)}) \overset{(c)}{=} \mathbb{E}_{\bm{q},\bm{z}|\bm{y},X,\Theta_k^{(s)}}\bigg[\sum_{i=0}^n\sum_{j=1}^k 1\{z_{t_{i}}=j\}\bigg( -\log(\sigma_k) - \frac{q_{t_i}}{2\sigma_k^2}(y_{t_i} - x_{t_i}^T\bm{\theta}_j)^2\bigg) + \log\big(p(\Theta_k,\sigma_k^2)\big)\bigg]
\end{align*}
Where we assume a Gaussian prior for the mean parameters $\bm{\theta}_j|\sigma_k^2 \sim \mathcal{N}(\bm{0},\sigma_k^2 \Phi)$ to represent our prior belief that the mean parameters are not far from zero. We assume an improper prior for the variance, $p(\sigma_k^2)\propto \frac{1}{\sigma_k^2}$.

The first step of the EM algorithm is to evaluate the posterior expectations of the relevant terms in the $Q$ function. In this case, we evaluate the posterior expectations of $\big(1\{z_{t_{i}}=j\}q_{t_i}\big)$ and $1\{z_{t_{i}}=j\}$. The first of these can be evaluated as a product of a conditional expectation and a marginal distribution,
\begin{align*}
\mathbb{E}_{z_{t_i},q_i|\bm{y},X,\Theta_k^{(s)}}\bigg[1\{z_{t_{i}}=j\}q_{t_i}\bigg] = \mathbb{E}_{q_{t_i}|z_{t_i},\bm{y},X,\Theta_k^{(s)}}\bigg[q_{t_i}\bigg|1\{z_{t_{i}}=j\}\bigg]\mathbb{E}_{z_{t_i}|,\bm{y},X,\Theta_k^{(s)}}\bigg[1\{z_{t_{i}}=j\}\bigg]
\end{align*}
The conditional expectation can be evaluated using the posterior distribution,
\[
[q_{t_i}|z_{t_i}=j,y_{t_i},X,\bm{\theta}_j^{(s)},\sigma_k^{2(s)}] \overset{(d)}{=} Ga(\frac{\nu+1}{2}, (\frac{\nu}{2} + \frac{(y_{t_i} - x_{t_i}^T\theta_j^{(s)})^2}{2\sigma_k^{2(s)}}))
\]
The expectations of $1\{z_{t_{i}}=j\}$ can be evaluated using the forward-backward algorithm. After these expectations are evaluated, the $Q$ function can be optimized with respect to the parameters $\Theta_k$ and $\sigma_k^2$. The M-steps for the mean parameters can be evaluated analytically since the likelihood is Gaussian, 
\[\bm{\theta}^{(s+1)}_j = \big(X^TW^{(s)}_jX + \Phi^{-1}\big)^{-1}X^T W^{(s)}_j\bm{y}\] 
Where $W^{(s)}_j$ is a diagonal matrix with entries $\mathbb{E}[1\{z_{t_i}=j\}q_{t_i}|y,\Theta_k^{(s)},\sigma_k^{2(s)}]$. The M-step for the variance $\sigma_k^2$ can also be evaluated analytically, 
\[ \sigma_k^{2(s+1)} = \frac{\sum_{i=0}^n \sum_{j=1}^k \mathbb{E} \bigg[1\{z_{t_i}=j\}q_{t_i} \big| \bm{y},\Theta^{(s)},\sigma_k^{2(s)}\bigg] \bigg(y_{t_i} -x_{t_i}^T\bm{\theta}^{(s+1)}_j\bigg)^2 + \sum_{j=1}^k\bm{\theta}^{(s+1)}_j\Phi^{-1}\bm{\theta}^{(s+1)}_j}{n + pk + 2}
\] 
Where $p$ is the dimension of $\bm{\theta}_j$ for all $j$. After the M-step is complete, the E-step is then repeated conditioned on the updated parameters.  The algorithm is repeated until convergence of the $Q$ function. Additional details are in the appendix.

\subsection{Marginal posterior distribution on number of segments}
\label{sub:pk_y}
A major benefit of reparameterizing the change point problem using state variables $\bm{z}\sim \textbf{BPP}$ within a hidden Markov model is the marginal likelihood can be computed using results from the forward backward algorithm. The forward recursions for this model represent the joint probability $a_j(i) = p(y_{t_0},\dots,y_{t_i},z_{t_i}=j|\Theta_k,\sigma_k^2)$ and take the form,
\begin{align*}
    a_j(0) &= \begin{cases}
        f(y_{t_0}|\bm{\theta}_1,\sigma_k^2) \text{ if } j=1\\
        0 \text{ else }
    \end{cases}\\
    a_j(i+1) &=  f(y_{t_{i+1}}|\bm{\theta}_j,\sigma_k^2) \sum_{l=1}^j a_l(i) \pi_k(z_{t_{i+1}}=j|z_{t_i}=l) 
\end{align*}
Where $f$ is the location-scale t-distributed likelihood described in subsection \ref{sub:robust}. Note, however, marginal likelihoods can be obtained using this approach for general conditional likelihood distributions.  The marginal likelihood is given by $f(\bm{y}|\Theta_k,\sigma_k^2) = \sum_{j=1}^k a_j(n)$. 

Since the integral of $f(\bm{y}|\Theta_k,\sigma_k^2)p(\Theta_k,\sigma_k^2)$ over $\Theta_k$ and $\sigma_k^2$ is intractable, we use a Laplace approximation keeping only the terms associated with the Bayesian information criterion for computational purposes \citep{schwarz1978estimating,tierney1986accurate,konishi2008information,killick2012optimal}. As such, the log marginal posterior distribution on the number of segments is approximated up to a normalizing constant,
\begin{equation}\label{eq:pk_y}
\log p(k|\bm{y}) \overset{(c)}{\approx} \log f(\bm{y}|\hat{\Theta}_k,\hat{\sigma_k}^2) - \frac{p_k}{2}\log(n) + \log p(k)  
\end{equation}
Where $p_k = \text{dim}(\Theta_k) + 1$ and the prior $p(k)$ is from Equation \ref{eq:pk} established in subsection \ref{sub:pk}.

\subsection{Loss function for change point locations}
\label{subsec:loss}
In this section, we introduce a loss function on change point locations and derive a Bayes estimator for that loss function in both discrete and continuous time. Define $\tau_j = \sum_{l=1}^j \zeta_j$ for $j=1,\dots,k-1$ as change point location parameters. To avoid identifiability issues, we restrict these change point locations to be at observed times $t_i\in [0,1]$ for continuous time or $t_i\in \{0,\dots,n\}$ for discrete time. For a specified number of changes $k-1$, a natural loss function for comparing two change point configurations is
the absolute loss between the $\tau_j$ locations, $L(\bm{\tau},\bm{\tau}^*)\doteq\sum_{j=1}^{k-1} |\tau_j - \tau_j^*|$. See \cite{truong2020selective} for a review of other loss functions. This absolute loss in turn induces a weighted Hamming loss between change point state sequences $\bm{z}$ and $\bm{z}^*$,
\begin{align*}
L(\bm{\tau},\bm{\tau}^*) &\doteq \sum_{j=1}^{k-1} |\tau_j - \tau_j^*| =
\sum_{i=1}^n \sum_{j=1}^{k-1}
1\big\{\min\{\tau_j, \tau_j^*\} < t_i \leq \max\{\tau_j, \tau_j^*\}\big\}(t_i - t_{i-1}) \\
&= \sum_{i=1}^n |z_{t_i} - z_{t_i}^*|(t_i - t_{i-1}) = H(\bm{z}, \bm{z}^*).
\end{align*}
The second equality holds since the difference $|\tau_j-\tau_j^*|$ is the sum of time increments within that window. The third equality holds since the indicator function of an observed time $t_i\in [\min\{\tau_j, \tau_j^*\},\max\{\tau_j, \tau_j^*\}\big\}]$ can occur for multiple segments $j$ and thus the interval $(t_i - t_{i-1})$ should be summed $|z_{t_i} - z_{t_i}^*|$ times. This final step yields a doubly weighted Hamming distance.  Note in discrete time each of the intervals $(t_i - t_{i-1}) = 1$ and so the distance reduces to summing over 
$|z_{t_i} - z_{t_{i-1}}|$.

Note, the weighted Hamming loss does not depend on the number of change points in the configurations and is thus more general. We choose to find the Bayes estimator for $H(\bm{z}, \bm{z}^*)$ so we can infer change point locations and number of change points simultaneously.
\begin{theorem}
\label{thm:bayesEst}
    The Bayes estimator for the weighted Hamming loss $H(\bm{z}, \bm{z}^*)$ between change point process realizations $\{z_{t_i}\}_{i=0}^n$ is for each $z_{t_i}$
    \[
        \hat{z}_{t_i} = \underset{j}{\text{min}}\bigg( \sum_{k = 1}^K\sum_{l=1}^j \pi(z_{t_i}=l|k,\bm{y})\pi(k|\bm{y}) \geq 0.5 \bigg)
    \]
That is, the Bayes estimator for $\bm{z}$ is the component-wise medians.  Furthermore, the Bayes estimator $\{\hat{z}_{t_i}\}_{i=1}^n$ is a change point process.  This result holds in both the discrete and continuous time settings.
\end{theorem}

Using our EM inference procedure, we estimate this Bayes estimator as follows. The $\pi(k|\bm{y})$ term is estimated using Equation \ref{eq:pk_y} and the $\pi(z_{t_i}=l|k,\bm{y})$ terms are estimated using the marginal expectations $\mathbb{E}[z_{t_i}=l|k,\Theta^{(s)},\bm{y}]$ computed in the last $s$th iteration of the forward backward algorithm.

\section{Simulation Study}
\label{sec:sim}
\begin{figure}[t]
    \centering
    \includegraphics[width=1.\linewidth]{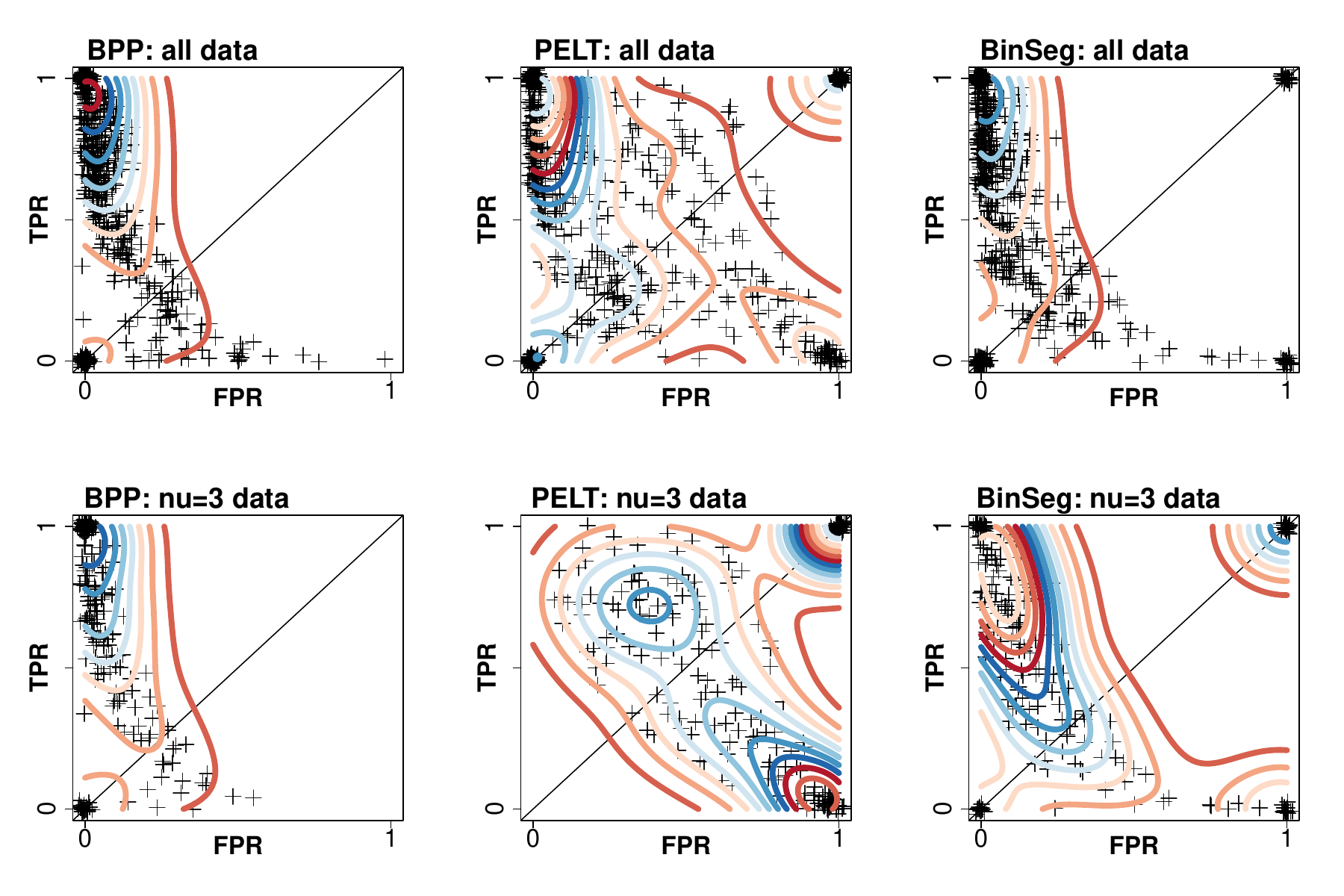}
    \caption{True positive rate against false positive rate on a synthetic data set with varying sizes of change, time distribution, outlier magnitude, and number of segments. The first row compares \textbf{BPP},\textbf{PELT},and \textbf{BinSeg} on all of the data from the factorial study and the second row compares them on a subset when $\nu=3$.}
    \label{fig:main_study}
\end{figure}
\begin{figure}[t]
    \centering
    \includegraphics[width=1.\linewidth]{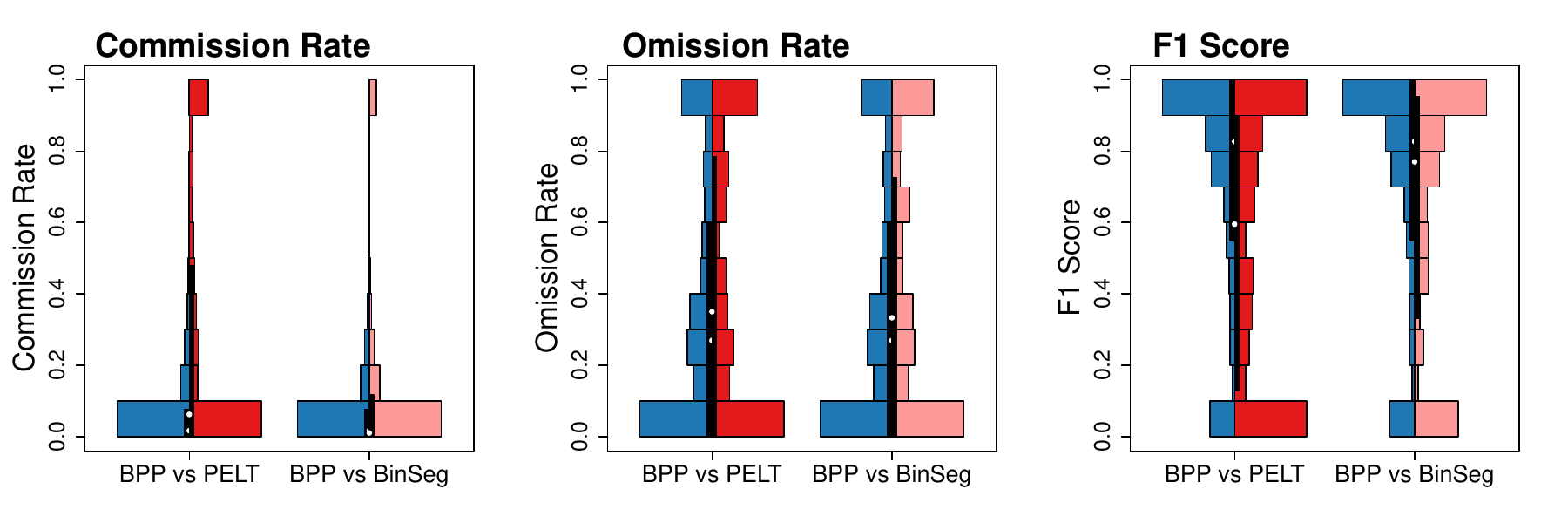}
    \caption{Breakdown of commission rate, omission rate, and F1 score for the \textbf{BPP}, \textbf{PELT}, and \textbf{BinSeg models}.}
    \label{fig:comomf1}
\end{figure}
Our simulation study is aimed at characterizing the performance of different change point models across a variety of conditions within a factorial setup. Each combination in the factorial study has $100$ replicates.  All simulated datasets are intercept only models with constant variance across segments and scaled t-distributed error as described below. Here are the settings that make up the factorial study,
\begin{itemize}
    \item \textbf{Time and change point distributions}: $1.)$ Uniformly spaced discrete time with uniformly distributed change points. This is the time and change point distribution assumed in \cite{killick2012optimal}. $2.)$ $\textbf{Beta}(0.5,0.5)$ distributed time with \textbf{BPP} distributed change points, $3.)$ $\textbf{Beta}(2,2)$ distributed time with \textbf{BPP} distributed change points.
    \item \textbf{Error variance}: $\sigma^2 \in (0.1,0.2,0.3)$
    \item \textbf{Robustness parameter}: $\nu \in (3,10,100)$
    \item \textbf{Size of change in intercept}: $(0.1,0.3,0.5,0.7,0.9,1.1)$
    \item \textbf{Number of segments}: $k = 1,\dots,4$
\end{itemize}
There are six models being compared in this study. All models explore up to $K=6$ segments except for the \textbf{PELT} model which does not support a maximum number of a segments argument.
\begin{itemize}
    \item \textbf{BPP}: This is our main model with transition probabilities according to Theorem \ref{thm:conttime}, location-scale t-distributed likelihood with $\nu=3$ according to subsection \ref{sub:robust} and noninformative prior on number of segments from Equation \ref{eq:pk}.
    \item \textbf{PELT} and \textbf{BinSeg}: These are two popular models used for change detection \citep{killick2012optimal,scottnknott}.  Both models are available in the \textit{changepoint} R package \citep{killick2014changepoint}. These models assume observations are from uniformly spaced discrete time intervals. We used the default setting for the \textit{cpt.mean} function which assumes a Normal cost function and a Modified BIC penalty \citep{zhang2007modified}. Readers are encouraged to learn more about these models in the references above.
    \item \textbf{BPP nonrobust}: This model is the same as \textbf{BPP} except it assumes a Gaussian likelihood without robustness to outliers.
    \item \textbf{Noninformative discrete time model}: This is our noninformative discrete time model from Propositions \ref{prop:hypergeo} and \ref{prop:trans} with location-scale t-distributed likelihood from subsection \ref{sub:robust} and noninformative prior on number of segments from Equation \ref{eq:pk}.
    \item \textbf{BPPE}: This is the \textbf{BPP} model but with prior on number of segments that assumes equally likely sequences across number of segments $k$ from Equation \ref{eq:impliedpk}.
\end{itemize}
The results for the last three models are in the appendix. In total, there are $3^3 \cdot 6 \cdot 4 \cdot 100 = 64800$ total datasets that are modeled. For each of the $648$ different factorial settings, the $100$ replicates were used to calculate the omission and commission rates for each model.  In order to capture settings similar to our case study, each data set is simulated with $n=500$ observations over a $20$ year period. Detected changes are considered true if they are within a $3$-month window of the true change which reduces to a $0.0225$ window after time is mapped to $[0,1]$ \citep{zhu2014continuous,zhu2020continuous,cohen2017similar}. If two changes are detected within the window of a true change, the closer one is considered a true change and the other is considered a false positive \citep{killick2012optimal}. Those 648 results were then plotted as points with a kernel density estimator with bandwidth .5 for visualization aid.  

The performance of the \textbf{BPP} is notable in Figure \ref{fig:main_study} and Figure \ref{fig:comomf1}. It appears that the combination of continuous time state space modeling, in addition to a noninformative prior on the number of segments and a robust likelihood lead to better performance than the other models. The results for all $6$ models are broken down further by time distribution, error variance, robustness, and number of segments in the appendix.  Those results confirm that each of the additional models, ($\textbf{BPP}$ without robustness, without a continuous time prior, or without a noninformative prior on the number of segments), each perform worse than the full $\textbf{BPP}$ model. Notably, Figures \ref{fig:os_ss_nseg} and \ref{fig:ms_ss_nseg} respectively show the discrete time noninformative model has poorer performance on the $k=1$ datasets, whereas, the \textbf{BPP} model does very well in the $k=1$ setting. As the only difference between those two models is the assumption of discrete versus continuous time, this demonstrates the need for a continuous-time change detection model to avoid additional false positives in the continuous time setting.  

Note in Figure \ref{fig:ms_ss_robust} in the appendix, the main performance disparities between \textbf{PELT}, \textbf{BinSeg}, and \textbf{BPP} appear in the datasets with $\nu=3$, that is, the datasets with the largest heavy-tailed error distribution, and otherwise their performance is comparable. In Figure \ref{fig:os_ss_robust}, the \textbf{BPP nonrobust} model also does well on the $\nu=3$ subset of datasets compared to \textbf{PELT} and \textbf{BinSeg}, demonstrating that, even though it assumes a Gaussian likelihood, its \textbf{BPP} continuous time prior achieves robustness to outliers by noninformatively down weighting the probability of change in cases when an outlier occurs shortly after the previous observation.

While our \textbf{BPP} model and methodology offer orders of magnitude of computational improvement compared to MCMC methods, \textbf{PELT} and \textbf{BinSeg} have much lower computational cost. In our simulated study, the \textbf{BPP} model runs in $3\mathrm{e}{-1}$ seconds per dataset, whereas the \textbf{PELT} and \textbf{BinSeg} models run in $7\mathrm{e}{-4}$ and $9\mathrm{e}{-4}$ seconds per dataset, respectively.

This computational disparity lies in the difference of the model being assumed. \textbf{PELT} and \textbf{BinSeg} do not allow the penalty (the log prior in our setting) to depend on the number or location of change points \citep{killick2012optimal}. Whereas, the \textbf{BPP} model assumes a prior on both the number and location of change points due to its latent Markov chain. Future research may explore if pruning can be used for inference in the \textbf{BPP} model to enjoy similar computational benefits enjoyed by \textbf{PELT}. Otherwise, a part of this computational gap may be closable by reimplementing our code in C, which we save for future work.

Finally, we derive a full Bayesian approach for the \textbf{BPP} model using exact simulation for the conditional posterior of the continuous time state variables following \cite{chib1996calculating} and test its performance on the synthetic data in the appendix as well. Code for running our proposed models can be found \url{https://github.com/daniel-s-cunha/BPP/}.

\section{Phenological Modeling with Multiple Change Points}
\label{sec:pheno}
Phenology is the study of the timing of biological activity over the course of a year, particularly in relation to climate. Phenological modeling of vegetation is often carried out using imagery data from Earth observation satellites. Spectral reflectances of Earth's surface are commonly combined to create vegetation indices, the most widely used of which is called the Normalized Difference Vegetative Index (NDVI), which are then used to monitor seasonal changes in vegetation. Specifically, the NDVI exploits the fact that healthy leaves are highly reflective in near infrared wavelengths and highly absorptive of light in the red wavelengths.  By taking the normalized difference of these two measurements, the NDVI provides an excellent surrogate measure for the amount of green leaf area on the ground: $\text{NDVI} = \frac{\text{NIR} - \text{Red}}{\text{NIR} + \text{Red}}$. Our goal is to model and detect changes within this vegetation index over time, so that land cover changes due to environmental factors are not mistakenly modeled as phenological signal.

\subsection{Harmonic regression}
\label{subsec:interannual}
Let $\bf{y}$ be an observed time series of NDVI from a single pixel of satellite imagery. Let time be standardized such that ${t_i}\in[0,1]$ by subtracting the minimum time and dividing by the total time interval. Let $T$ be the time interval of the study with units in days. For an observation $y_{t_i}$, define the harmonic regression model as,
\begin{equation}\label{eq:GausLik}
    y_{t_i}|\bm{\theta},\sigma^2,q_{t_i} \stackrel{\text{ind}}{\sim} N\Bigg(\alpha+\beta t_i + \sum_{h=1}^H \gamma_{h} \sin(h\omega t_i) + \delta_h \cos(h\omega t_i),\, \frac{\sigma^2}{q_{t_i}}\Bigg)
\end{equation}
where $\bm{\theta} = (\alpha,\beta,\{\gamma_h,\delta_h\}_{h=1}^H)^T$ are the mean model parameters, $\omega=2\pi T/365$ is the harmonic frequency, $q_{t_i}\sim \textbf{Ga}(\nu/2,\nu/2)$ is the robustness latent variable from Subsection \ref{sub:robust}, and $H$ is the number of harmonics in the model. Define a design matrix $X$ such that $\bm{x}_{t_i}^T\bm{\theta} = \alpha+\beta t_i + \sum_{h=1}^H \gamma_{h} \sin(h\omega t_i) + \delta_h \cos(h\omega t_i)$. This model is popular in the remote sensing community because the decomposition of phenological dynamics into an intercept, slope, and harmonics enables researchers to make inferences about seasonality as well as long term trends \citep{zhu2014continuous}.

\subsection{Interannually varying harmonics}
One limitation of the above model is it assumes the mean function follows the same seasonality pattern each year. Change point detection algorithms that use the above model may exhibit higher false positive rates, since seasonal anomalies are not captured by the model and thus may be falsely detected as change points. To address this limitation, we introduce harmonic contrasts to the model, giving it flexibility to capture interannual variation. Let $l(t)$ be the year at time $t$ and consider a harmonic contrast for each year as follows,
\begin{equation*}
    y_{t_i}|\bm{\theta},\bm{\phi},\sigma^2,q_{t_i} \stackrel{\text{ind}}{\sim} N\Bigg(\bm{x}_{t_i}^T\bm{\theta} + \sum_{h=1}^H \gamma_{h,l(i)} \sin(h\omega t_i) + \delta_{h,l(i)} \cos(h\omega t_i),\, \frac{\sigma^2}{q_{t_i}}\Bigg),
\end{equation*}
where $\bm{\phi} = (\{\gamma_{h,l}\}_{h=1,l=1}^{H,J},\{\delta_{h,l}\}_{h=1,l=1}^{H,J})^T$ is the contrast parameter vector and $\bm{x}_{t_i}^T\bm{\theta}$ is the mean function without contrasts. Let $W$ be the design matrix for the harmonic contrasts designed such that $\bm{w}^{T}_{t_i}\bm{\phi} = \sum_{h=1}^H \gamma_{h,l(i)} \sin(h\omega t_i) + \delta_{h,l(i)} \cos(h\omega t_i)$. Our model can then be summarized in terms of the mean components and contrast components as,
\begin{equation}\label{eq:interanLik}
    y_{t_i}|\bm{\theta},\bm{\phi},\sigma^2,q_{t_i} \stackrel{\text{ind}}{\sim} N\Bigg(\bm{x}_{t_i}^T\bm{\theta} + \bm{w}^{T}_{t_i}\bm{\phi},\, \frac{\sigma^2}{q_{t_i}}\Bigg),
\end{equation}

\subsection{Continuity constraints on the mean function and its derivative}
Since we are interested in detecting changes in phenological signal, it is important that the mean $\bm{x}_{t}^T\bm{\theta} + \bm{w}^{T}_{t}\bm{\phi}$ and its first derivative are continuous for all $t\in[0,1]$ so there are no discontinuities that can be mistaken for changes in the intercept or slope. Thus, we introduce the following constraints,
\begin{proposition}
\label{prop:IAVH}
Placing continuity constraints, with respect to time, on the mean function and its first derivative in~\eqref{eq:interanLik} yields the following linear constraints on the contrast harmonic parameters for each $l$th year,
\[
\gamma_{H,l} = -\sum_{h=1}^{H-1} \gamma_{h,l} \quad \text{and} \quad
\delta_{H,l} = -\sum_{h=1}^{H-1} \delta_{h,l}.
\]
\end{proposition}

These continuity constraints also have implications for how we design the prior for the contrast parameters. Specifically, if the vector $\bm{\phi}$ of all contrasts including the $H$th harmonic parameters has the following distribution $\bm{\phi} \sim \mathcal{N}(\bm{0},\Phi^{(0)})$, then we need to adjust the prior by conditioning on the continuity constraints, 
\[ \bigg(\bm{\phi} \bigg| \big\{\gamma_{H,l} = -\sum_{h=1}^{H-1} \gamma_{h,l}\big\}_{l=1}^L,
\big\{\delta_{H,l} = -\sum_{h=1}^{H-1} \delta_{h,l}\big\}_{l=1}^L \bigg) \sim \mathcal{N}(\bm{0},\Phi^{(1)})\]
Where the new covariance matrix $\Phi^{(1)}$ is derived in the appendix.
\section{Case Study}
\label{sec:casestudy}
Our case study aims to demonstrate robust continuous-time change point detection for three remote sensing examples.  We chose canonical examples of how change detection is or can be used in this broad field of research using data collected from Earth observation satellites. Data from the Landsat satellites are used in all three studies \citep{friedl2022medium}.  These imagery have a spatial resolution of 30 meters and a repeat frequency of 8 to 16 days, not accounting for missing data from clouds. To provide independent reference data allowing us to identify changes on the ground, we use temporally-sparse high-spatial resolution imagery in Google Earth, and high-quality continuous precipitation data such as the standardized precipitation-evapotranspiration index and drought data \citep{yu2012google,begueria2014standardized,owens2007unl}.  Note that while these data sources are highly informative, they are not sufficiently comprehensive to determine all changes. However, they are sufficient for the purpose of demonstrating the robustness of our method. The first study applies our method to the challenge of detecting deforestation in the Rondonia region of the Amazon rainforest.  The second study  applies our method to the problem of detecting changes in land management in an agricultural field in the San Joaquin Valley of California, and the third study applies our method to detecting responses of semi-arid vegetation in Texas to drought and year-to-year variation in precipitation.

\subsection{Case study model}
The phenological model from Equation \ref{eq:interanLik} is used with $H = 2$ harmonics, $K=6$ maximum number of segments, and a parameter prior covariance that accounts for continuity constraints in the mean function and its first derivative as detailed in the appendix. We assume a robustness parameter of $\nu=3$.

Changes are searched for in the mean parameters $\bm{\theta} = \{\alpha,\beta,\{\gamma_h,\delta_h\}_{h=1}^H\}$ but not in the interannual harmonics $\bm{\phi} = (\{\gamma_{h,l}\}_{h=1,l=1}^{H,J},\{\delta_{h,l}\}_{h=1,l=1}^{H,J})$ nor the error variance. For example, a change in the intercept $\alpha$ could represent deforestation or other changes where a land cover is removed or added. A change in trend $\beta$ can represent a growth pattern, a decline, or a stabilization.  Changes in the mean harmonics $\{\gamma_h,\delta_h\}_{h=1}^H$ can capture events such as crop changes or land cover changes in general.

While the interannual harmonics $\bm{\phi} = (\{\gamma_{h,l}\}_{h=1,l=1}^{H,J},\{\delta_{h,l}\}_{h=1,l=1}^{H,J})$ are necessary for capturing inter-seasonal variation in the phenological signature, we do not search for changes in these parameters since they are temporally local parameters introduced for each year. We assume a single variance $\sigma^2$ across all segments to reflect our belief that measurement error is independent of phenological signal.

\subsection{Deforestation in Rondonia}
\begin{figure}[t]
\includegraphics[width=1.\linewidth]{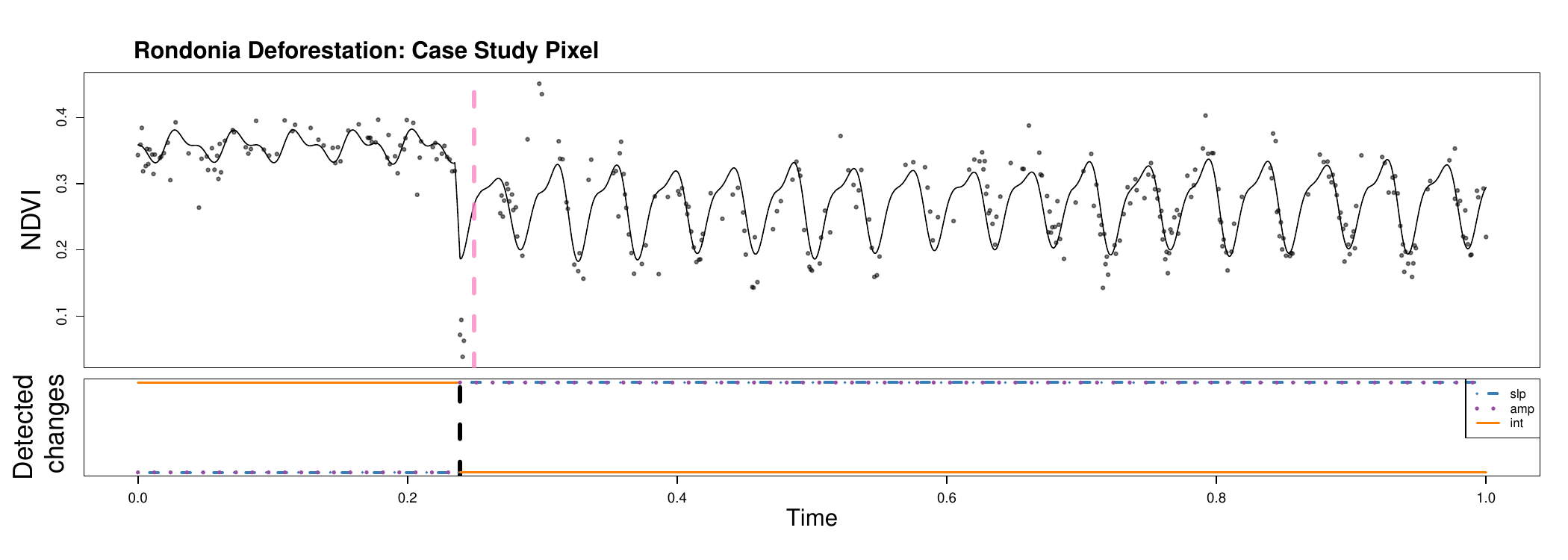}  
\caption{NDVI time series for a single pixel in the Amazon rainforest of Rondonia. A deforestation event (\color{pink} pink dotted line\color{black}) occurs in 2005 as confirmed using high resoultion imagery in Google Earth Pro as well as MapBiomas Brazil \citep{rs12172735}. The model detected change is shown in the lower panel.}
\label{fig:rondonia}
\end{figure}
Monitoring and limiting deforestation is of paramount importance towards slowing anthropogenically driven climate change. This can be difficult to do, especially in remote regions such as Amazonia, which are difficult to navigate and observe on the ground. The Rondonia region of the Amazon rainforest has experienced some of the highest rates of deforestation on the planet over the last 50 years \citep{pedlowski1997patterns,pedlowski2005conservation,butt2011evidence}. Here we present results from a single pixel located at -11.89 latitude and -63.59 longitude for a study period extending from 2000 to the end of 2022.

From a data perspective, detecting forest cover changes in tropical rainforests can be difficult because these regions have persistent cloud cover throughout much of the year.  This leads to high frequencies of missing data and non-uniform spacing of cloud free observations.  Hence, discrete time change detection models will be biased if they do not account for the missingness properly. We used an externally generated forest change data source (the MapBiomas Brazil project \citep{rs12172735}) to confirm the location and timing of deforestation in this case study. The model results are shown in Figure \ref{fig:rondonia}. 

\subsection{Crop rotation in the San Joaquin Valley}
Identifying and monitoring land management in agricultural regions from remote sensing data is an important task for a wide array of applications such as harvest and food supply projections \citep{li2024automated,Boryan01082011}. We chose an agriculture plot in the San Joaquin Valley of California with latitude 35.03 and longitude -118.91. Monitoring agricultural land management can be a difficult task because of crop rotations and other management decisions made by growers. For this reason, researchers often apply classifications at annual time steps that are independent of other years in the time series.  This strategy loses the benefits of longer time series that are available from remote sensing in many locations. Change detection methods can help solve this problem as they can be used to determine when phenological changes happen (i.e., that are diagnostic of specific crops or management practices), and thus classification can be done on each change segment of data as opposed to each year.

Using Google Earth Imagery and CropScape \citep{li2024automated} we annotated three changes in land management that were clearly visible in high-resolution imagery.  The high resolution imagery showed stable crops from 2000 until August 2006, at which point there is a fallow period with no crops. In August of 2012, the field is re-planted, and in October 2016 the geometric patterns related to crop type  visibly change in the high resolution imagery. The model detects each of these three changes as well as an additional change at the beginning of the time series during a period when we do not have high resolution imagery available for confirmation. The interpreted changes from the high-resolution imagery agree well with the changes detected automatically in the low-temporal and medium spatial resolution Landsat imagery (Figure \ref{fig:cropRotation}.).

\begin{figure}[t]
\includegraphics[width=1.\linewidth]{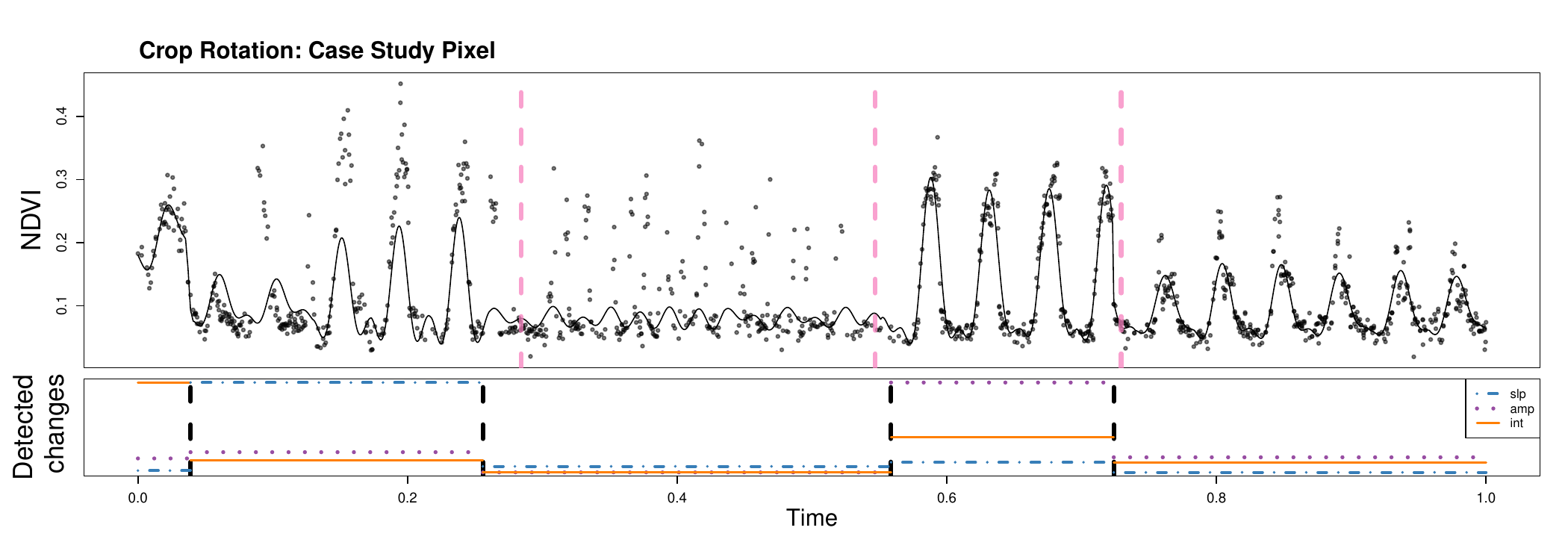} 
\caption{NDVI time series for a pixel located in an agriculture field in California. Three crop rotation events are annotated by the (\color{pink} pink dotted line\color{black}). The model detects an additional change towards the beginning of the time series when high resolution imagery is not available.}
\label{fig:cropRotation}
\end{figure}

\subsection{Semi-arid vegetation responses to drought}
\begin{figure}[t]
\includegraphics[width=1.\linewidth]{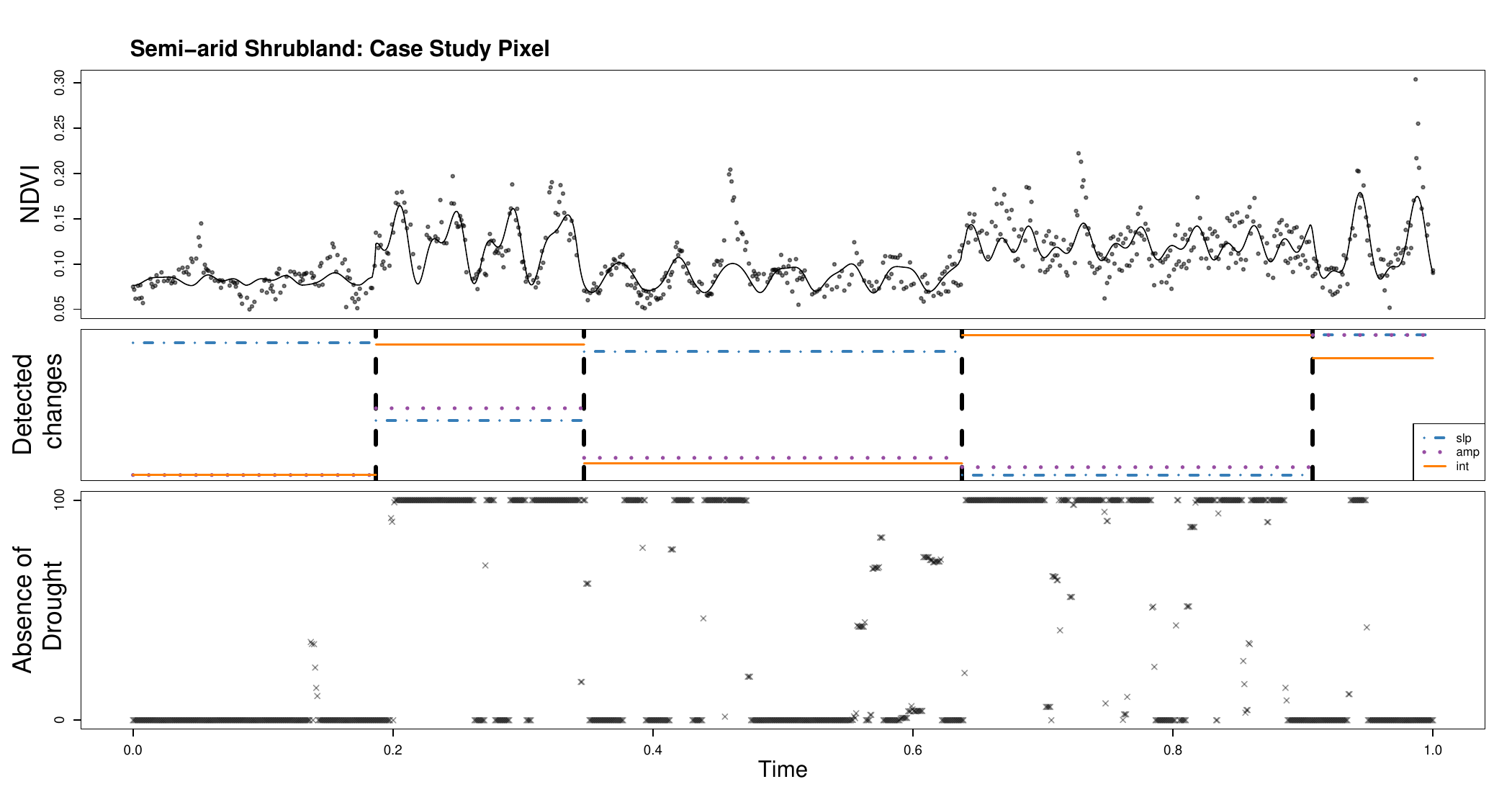} 
\caption{NDVI time series for a pixel located in a semi-arid region of Texas. Drought data from the National Drought Mitigation Center are plotted in the third panel. The model detects four major changes in the NDVI that appear to be related to major precipitation and drought events captured in the drought data.}
\label{fig:permianBasin}
\end{figure}
Climate change is affecting precipitation regimes in many parts of the world, leading to, for example, faster drought onsets \citep{mukherjee2018climate,Shenoy}.  Detecting changes in phenological signatures due to drought is thus an important and open question that can have implications for climate modeling and other tasks such as land management.  Semi-arid regions are particularly affected by drought and understanding the resilience of vegetation to stress from climate change in these areas is important.  Our study area for this case study is located in a semi-arid region of Texas at latitude 31.87 and longitude -103.64. We used high resolution imagery from Google Earth Pro to find a location with stable shrub and grass land cover (i.e., no land use)  in order to isolate the effects of drought. 

Drought index data are provided by the National Drought Mitigation Center, University of Nebraska-Lincoln \citep{owens2007unl}. We use these data to compare drought events to changes detected in NDVI at the same location.  Specifically, they provide a \textit{No Drought} measurement which scales from $0$ (i.e. full drought) to $100$ (i.e. no drought). We denote this measurement as \textit{Absence of Drought} in the third panel of Figure \ref{fig:permianBasin}. Our model of NDVI detects four changes in the phenological signature that co-occur with  significant precipitation anomalies and drought events as captured by the drought monitor. The results are in Figure \ref{fig:permianBasin}. Note that when we run our model without interannually varying harmonics, the model does not detect any of the drought/precipitation events. Those additional results can be found in the appendix.

\section{Discussion}
\label{sec:discuss}
In this work, we offer an end to end solution for continuous time change detection.  The change detection problem is first reframed in terms of a state space model where efficient and exact inference on the state variables is possible using the forward backward algorithm.  We then derived noninformative priors on change point processes and their corresponding transition probabilities in both discrete and continuous time and showed the continuous time priors have equivalent moments to $\textbf{Dir}(1_k)$. The continuous time transition probabilities are particularly notable, forming a class of Bernstein polynomial processes that adjust to the spacing of time measurements for the data at hand. These priors confirm our intuition that two consecutive observations that are closer in time are less likely to change than two consecutive observations far apart in time.

The prior on the number of segments is also tackled in this work.  We provide a discourse on measuring model space volumes in order to construct noninformative prior mass on the number of segments. Our reasoning is confirmed in synthetic studies where the \textbf{BPP} model competes with current state of the art methods and out performs them in the heavy tailed error distribution cases. This performance benefit is also owed to our development of a robust likelihood that can be inferred efficiently within the forward backward algorithm framework used to infer change points. 

Our case study addresses three canonical examples of change detection commonly used in remote sensing literature.  We developed a new semiparametric model that capture interannual variability due to weather while also maintaining interpretable parameters such as intercept and slope for which we'd like to infer changes. This new likelihood model out performed its commonly used harmonic regression predecessor as demonstrated in the appendix.

Future work may consider extending this continuous time model to the spatial domain.  There are also interesting parallels between our findings and theory regarding random partitions that did not fall within the scope of this work. Finally, it would be interesting to see how the continuous time transition probabilities can be parameterized to accommodate additional prior information or for use within an empirical Bayes approach.

\bibliography{Refs.bib}

\begin{thebibliography}{}

\bibitem[Adams and MacKay, 2007]{adams2007bayesian}
Adams, R.~P. and MacKay, D.~J. (2007).
\newblock Bayesian online changepoint detection.
\newblock {\em arXiv preprint arXiv:0710.3742}.

\bibitem[Aminikhanghahi and Cook, 2017]{aminikhanghahi2017survey}
Aminikhanghahi, S. and Cook, D.~J. (2017).
\newblock A survey of methods for time series change point detection.
\newblock {\em Knowledge and information systems}, 51(2):339--367.

\bibitem[Auger and Lawrence, 1989]{auger1989algorithms}
Auger, I.~E. and Lawrence, C.~E. (1989).
\newblock Algorithms for the optimal identification of segment neighborhoods.
\newblock {\em Bulletin of mathematical biology}, 51(1):39--54.

\bibitem[Barry and Hartigan, 1993]{barry1993bayesian}
Barry, D. and Hartigan, J.~A. (1993).
\newblock A bayesian analysis for change point problems.
\newblock {\em Journal of the American Statistical Association}, 88(421):309--319.

\bibitem[Beguer{\'\i}a et~al., 2014]{begueria2014standardized}
Beguer{\'\i}a, S., Vicente-Serrano, S.~M., Reig, F., and Latorre, B. (2014).
\newblock Standardized precipitation evapotranspiration index (spei) revisited: parameter fitting, evapotranspiration models, tools, datasets and drought monitoring.
\newblock {\em International journal of climatology}, 34(10):3001--3023.

\bibitem[Bishop, 2006]{bishop2006pattern}
Bishop, C.~M. (2006).
\newblock {\em Pattern recognition and machine learning}, volume~4.
\newblock Springer.

\bibitem[Boryan et~al., 2011]{Boryan01082011}
Boryan, C., Yang, Z., Mueller, R., and and, M.~C. (2011).
\newblock Monitoring us agriculture: the us department of agriculture, national agricultural statistics service, cropland data layer program.
\newblock {\em Geocarto International}, 26(5):341--358.

\bibitem[Butt et~al., 2011]{butt2011evidence}
Butt, N., De~Oliveira, P.~A., and Costa, M.~H. (2011).
\newblock Evidence that deforestation affects the onset of the rainy season in rondonia, brazil.
\newblock {\em Journal of Geophysical Research: Atmospheres}, 116(D11).

\bibitem[Chib, 1996]{chib1996calculating}
Chib, S. (1996).
\newblock Calculating posterior distributions and modal estimates in markov mixture models.
\newblock {\em Journal of Econometrics}, 75(1):79--97.

\bibitem[Chib, 1998]{chib1998estimation}
Chib, S. (1998).
\newblock Estimation and comparison of multiple change-point models.
\newblock {\em Journal of econometrics}, 86(2):221--241.

\bibitem[Cohen et~al., 2017]{cohen2017similar}
Cohen, W.~B., Healey, S.~P., Yang, Z., Stehman, S.~V., Brewer, C.~K., Brooks, E.~B., Gorelick, N., Huang, C., Hughes, M.~J., Kennedy, R.~E., et~al. (2017).
\newblock How similar are forest disturbance maps derived from different landsat time series algorithms?
\newblock {\em Forests}, 8(4):98.

\bibitem[Dempster et~al., 1977]{dempster1977maximum}
Dempster, A.~P., Laird, N.~M., and Rubin, D.~B. (1977).
\newblock Maximum likelihood from incomplete data via the em algorithm.
\newblock {\em Journal of the royal statistical society: series B (methodological)}, 39(1):1--22.

\bibitem[Fearnhead, 2006]{fearnhead2006exact}
Fearnhead, P. (2006).
\newblock Exact and efficient bayesian inference for multiple changepoint problems.
\newblock {\em Statistics and computing}, 16:203--213.

\bibitem[Fearnhead and Liu, 2007]{fearnhead2007line}
Fearnhead, P. and Liu, Z. (2007).
\newblock On-line inference for multiple changepoint problems.
\newblock {\em Journal of the Royal Statistical Society Series B: Statistical Methodology}, 69(4):589--605.

\bibitem[Friedl et~al., 2022]{friedl2022medium}
Friedl, M.~A., Woodcock, C.~E., Olofsson, P., Zhu, Z., Loveland, T., Stanimirova, R., Arevalo, P., Bullock, E., Hu, K.-T., Zhang, Y., et~al. (2022).
\newblock Medium spatial resolution mapping of global land cover and land cover change across multiple decades from landsat.
\newblock {\em Frontiers in Remote Sensing}, 3:894571.

\bibitem[Keenan et~al., 2014]{keenan2014net}
Keenan, T.~F., Gray, J., Friedl, M.~A., Toomey, M., Bohrer, G., Hollinger, D.~Y., Munger, J.~W., O’Keefe, J., Schmid, H.~P., Wing, I.~S., et~al. (2014).
\newblock Net carbon uptake has increased through warming-induced changes in temperate forest phenology.
\newblock {\em Nature Climate Change}, 4(7):598--604.

\bibitem[Killick et~al., 2012]{killick2012optimal}
Killick et~al. (2012).
\newblock Optimal detection of changepoints with a linear computational cost.
\newblock {\em Journal of the American Statistical Association}, 107(500):1590--1598.

\bibitem[Killick and Eckley, 2014]{killick2014changepoint}
Killick, R. and Eckley, I.~A. (2014).
\newblock changepoint: An r package for changepoint analysis.
\newblock {\em Journal of statistical software}, 58:1--19.

\bibitem[Konishi and Kitagawa, 2008]{konishi2008information}
Konishi, S. and Kitagawa, G. (2008).
\newblock {\em Information criteria and statistical modeling}.
\newblock Springer Science \& Business Media.

\bibitem[Li et~al., 2024]{li2024automated}
Li, H., Di, L., Zhang, C., Lin, L., Guo, L., Yu, E.~G., and Yang, Z. (2024).
\newblock Automated in-season crop-type data layer mapping without ground truth for the conterminous united states based on multisource satellite imagery.
\newblock {\em IEEE Transactions on Geoscience and Remote Sensing}, 62:1--14.

\bibitem[Little and Rubin, 2019]{little2019statistical}
Little, R.~J. and Rubin, D.~B. (2019).
\newblock {\em Statistical analysis with missing data}, volume 793.
\newblock John Wiley \& Sons.

\bibitem[Mukherjee et~al., 2018]{mukherjee2018climate}
Mukherjee, S., Mishra, A., and Trenberth, K.~E. (2018).
\newblock Climate change and drought: a perspective on drought indices.
\newblock {\em Current climate change reports}, 4:145--163.

\bibitem[Owens, 2007]{owens2007unl}
Owens, J.~C. (2007).
\newblock Unl center helps develop national drought initiative.

\bibitem[Pedlowski et~al., 1997]{pedlowski1997patterns}
Pedlowski, M.~A., Dale, V.~H., Matricardi, E.~A., and da~Silva~Filho, E.~P. (1997).
\newblock Patterns and impacts of deforestation in rond{\^o}nia, brazil.
\newblock {\em Landscape and Urban Planning}, 38(3-4):149--157.

\bibitem[Pedlowski et~al., 2005]{pedlowski2005conservation}
Pedlowski, M.~A., Matricardi, E.~A., Skole, D., Cameron, S., Chomentowski, W., Fernandes, C., and Lisboa, A. (2005).
\newblock Conservation units: a new deforestation frontier in the amazonian state of rond{\^o}nia, brazil.
\newblock {\em Environmental Conservation}, 32(2):149--155.

\bibitem[Peluso et~al., 2019]{peluso2019semiparametric}
Peluso, S., Chib, S., and Mira, A. (2019).
\newblock Semiparametric multivariate and multiple change-point modeling.

\bibitem[Schwarz, 1978]{schwarz1978estimating}
Schwarz, G. (1978).
\newblock Estimating the dimension of a model.
\newblock {\em The annals of statistics}, pages 461--464.

\bibitem[Scott and Knott, 1974a]{b399ed80-3ccc-30cf-af0d-87bbdad7ade6}
Scott, A.~J. and Knott, M. (1974a).
\newblock A cluster analysis method for grouping means in the analysis of variance.
\newblock {\em Biometrics}, 30(3):507--512.

\bibitem[Scott and Knott, 1974b]{scottnknott}
Scott, A.~J. and Knott, M. (1974b).
\newblock A cluster analysis method for grouping means in the analysis of variance.
\newblock {\em Biometrics}, 30(3):507--512.

\bibitem[Shenoy et~al., 2022]{Shenoy}
Shenoy, S., Gorinevsky, D., Trenberth, K.~E., and Chu, S. (2022).
\newblock Trends of extreme us weather events in the changing climate.
\newblock {\em Proceedings of the National Academy of Sciences}, 119(47):e2207536119.

\bibitem[Souza et~al., 2020]{rs12172735}
Souza, C.~M., Z.~Shimbo, J., Rosa, M.~R., Parente, L.~L., A.~Alencar, A., Rudorff, B. F.~T., Hasenack, H., Matsumoto, M., G.~Ferreira, L., Souza-Filho, P. W.~M., de~Oliveira, S.~W., Rocha, W.~F., Fonseca, A.~V., Marques, C.~B., Diniz, C.~G., Costa, D., Monteiro, D., Rosa, E.~R., Vélez-Martin, E., Weber, E.~J., Lenti, F. E.~B., Paternost, F.~F., Pareyn, F. G.~C., Siqueira, J.~V., Viera, J.~L., Neto, L. C.~F., Saraiva, M.~M., Sales, M.~H., Salgado, M. P.~G., Vasconcelos, R., Galano, S., Mesquita, V.~V., and Azevedo, T. (2020).
\newblock Reconstructing three decades of land use and land cover changes in brazilian biomes with landsat archive and earth engine.
\newblock {\em Remote Sensing}, 12(17).

\bibitem[Stephens, 1994]{stephens1994bayesian}
Stephens, D. (1994).
\newblock Bayesian retrospective multiple-changepoint identification.
\newblock {\em Journal of the Royal Statistical Society: Series C (Applied Statistics)}, 43(1):159--178.

\bibitem[Tierney and Kadane, 1986]{tierney1986accurate}
Tierney, L. and Kadane, J.~B. (1986).
\newblock Accurate approximations for posterior moments and marginal densities.
\newblock {\em Journal of the american statistical association}, 81(393):82--86.

\bibitem[Truong et~al., 2020]{truong2020selective}
Truong, C., Oudre, L., and Vayatis, N. (2020).
\newblock Selective review of offline change point detection methods.
\newblock {\em Signal Processing}, 167:107299.

\bibitem[Yao, 1984]{yao1984estimation}
Yao, Y.-C. (1984).
\newblock Estimation of a noisy discrete-time step function: Bayes and empirical bayes approaches.
\newblock {\em The Annals of Statistics}, pages 1434--1447.

\bibitem[Yu and Gong, 2012]{yu2012google}
Yu, L. and Gong, P. (2012).
\newblock Google earth as a virtual globe tool for earth science applications at the global scale: progress and perspectives.
\newblock {\em International Journal of Remote Sensing}, 33(12):3966--3986.

\bibitem[Zhang and Siegmund, 2007]{zhang2007modified}
Zhang, N.~R. and Siegmund, D.~O. (2007).
\newblock A modified bayes information criterion with applications to the analysis of comparative genomic hybridization data.
\newblock {\em Biometrics}, 63(1):22--32.

\bibitem[Zhu and Woodcock, 2014]{zhu2014continuous}
Zhu, Z. and Woodcock, C.~E. (2014).
\newblock Continuous change detection and classification of land cover using all available landsat data.
\newblock {\em Remote sensing of Environment}, 144:152--171.

\bibitem[Zhu et~al., 2020]{zhu2020continuous}
Zhu, Z., Zhang, J., Yang, Z., Aljaddani, A.~H., Cohen, W.~B., Qiu, S., and Zhou, C. (2020).
\newblock Continuous monitoring of land disturbance based on landsat time series.
\newblock {\em Remote Sensing of Environment}, 238:111116.

\end{thebibliography}
\bibliographystyle{apalike}
\section{Appendix A: Theoretical Results}
\subsection{Proofs}
\begin{proof}[Proof of Proposition \ref{prop:hypergeo}]
We remove the prior designation of $p$ for notational simplicity. Recall the binomial coefficient property $\binom{m}{l} = \binom{m-1}{l} + \binom{m-1}{l-1}$.  Note there are four possibilities for position $i$ being in state $j$.  
\begin{align*}
p(z_i=j) = \sum_{l=j}^{j+1}\sum_{m=j-1}^{j} p(z_i=j,z_{i-1}=m,z_{i+1}=l)
\end{align*}    
Similarly to how the denominator was counted, note that when $z_i=j$, there are $\binom{i-1}{m-1}$ ways to choose the initial $m-1$ changes in the first $i$ time points, and $\binom{n-i-1}{k-l}$ ways to choose the final $k-l$ changes in the final $n-i$ time points. Following with distributive property of multiplication and the binomial coefficient property,
\begin{align*}
p(z_i=j) &= \sum_{l=j}^{j+1}\sum_{m=j-1}^{j} \frac{\binom{i-1}{m-1}\binom{n-i-1}{k-l}}{\binom{n}{k-1}}\\
&=\frac{\binom{i-1}{j-2}\binom{n-i-1}{k-j}}{\binom{n}{k-1}} + \frac{\binom{i-1}{j-1}\binom{n-i-1}{k-j}}{\binom{n}{k-1}} + \frac{\binom{i-1}{j-2}\binom{n-i-1}{k-j-1}}{\binom{n}{k-1}} + \frac{\binom{i-1}{j-1}\binom{n-i-1}{k-j-1}}{\binom{n}{k-1}}\\
&=\frac{\binom{i}{j-1}\binom{n-i-1}{k-j}}{\binom{n}{k-1}} + \frac{\binom{i}{j-1}\binom{n-i-1}{k-j-1}}{\binom{n}{k-1}}\\
&=\frac{\binom{i}{j-1}\binom{n-i}{k-j}}{\binom{n}{k-1}}
\end{align*}    
\end{proof}
\begin{proof}[Proof of Proposition \ref{prop:trans}]
First note that in order to be in segment $1$ at time $i$, then the data process must have been in segment $1$ at time $i-1$,
\begin{align*}
    p(z_i = 1) = p(z_i = 1 | z_{i-1}=1)p(z_{i-1} = 1)
\end{align*}
Which implies,
\begin{align*}
    p(z_{i} = 1 | z_{i-1}=1) = \frac{p(z_{i} = 1)}{p(z_{i-1} = 1)}
\end{align*}
Moving on to the next segment, 
\begin{align*}
    p(z_{i} = 2) &= p(z_{i} = 2 | z_{i-1}=2)p(z_{i-1} = 2) + p(z_{i} = 2 | z_{i-1}=1)p(z_{i-1} = 1)\\
    &= p(z_{i} = 2 | z_{i-1}=2)p(z_{i-1} = 2) + (1- p(z_{i} = 1 | z_{i-1}=1))p(z_{i-1} = 1)\\
    &= p(z_{i} = 2 | z_{i-1}=2)p(z_{i-1} = 2) + (1- \frac{p(z_{i} = 1)}{p(z_{i-1} = 1)})p(z_{i-1} = 1)\\
    &= p(z_{i} = 2 | z_{i-1}=2)p(z_{i-1} = 2) + (p(z_{i-1}= 1) - p(z_{i} = 1))
\end{align*}
Which implies,
\begin{align*}
    p(z_{i} = 2 | z_{i-1}=2) &= \frac{\sum_{l=1}^2 p(z_{i} = l) - \sum_{l=1}^1 p(z_{i-1} = l)}{p(z_{i-1} = 2)}
\end{align*}
And in general we find recursively for all $1<j<k$,
\begin{align*}
    p(z_{i} = j | z_{i-1}= j) &= \frac{\sum_{l=1}^j p(z_{i} = l) - \sum_{l=1}^{j-1} p(z_{i-1} = l)}{p(z_{i-1} = j)}
\end{align*}
With $p(z_{i} = k | z_{i-1}= k) = 1$ by assumption.
\end{proof}

\begin{proof}[Proof of Theorem \ref{thm:hgeobern}]
Let $t\in[0,1]$ and define corresponding discrete time $i=\floor{tn}$. We prove the statement in terms of $i$ noting that $\lim_{n\to\infty} i/n = \lim_{n\to\infty} \floor{tn}/n = t$. Define the continuous time marginal $p(z_{t}=j) \coloneq \lim_{n\to\infty} p(z_{i}=j)$. Begin by evaluating the binomial coefficients,

\begin{align*}
p(z_{i}=j) &= \frac{\binom{n-i}{k-j} \binom{i}{j-1}}{\binom{n}{k-1}}\\
&= \frac{\frac{(n-i)!}{(k-j)!(n-i-k+j)!}\frac{i!}{(j-1)!(i-j+1)!}}{\frac{n!}{(k-1)!(n-k+1)!}}&&\text{(expand)}\\
&= \binom{k-1}{j-1}\frac{(n-i)!}{(n-i-k+j)!}\cdot \frac{i!}{(i-j+1)!}\cdot \frac{(n-k+1)!}{n!}&&\text{(rearrange)}\\
&= \binom{k-1}{j-1}\frac{n^{k-1}}{n^{k-j}n^{j-1}}\cdot \frac{(n-i)!}{(n-i-k+j)!}\cdot \frac{i!}{(i-j+1)!}\cdot \frac{(n-k+1)!}{n!}&&\text{(multiply by 1)}\\
&= \binom{k-1}{j-1}\frac{(n-i)!}{n^{k-j}(n-i-k+j)!}\cdot \frac{i!}{n^{j-1}(i-j+1)!}\cdot \frac{n^{k-1}(n-k+1)!}{n!}&&\text{(rearrange)}
\end{align*}

We will inspect the limit of each of the three fractions separately as $n\to\infty$,

\begin{align*}
\lim_{n\to \infty} \frac{1}{n^{k-j}}\frac{(n-i)!}{(n-i-k+j)!} &= \lim_{\substack{i\to \infty\\n\to \infty}} \frac{(n-i)\cdot \dots \cdot (n-i-k+j+1)}{n^{k-j}}\\
&= \lim_{n\to \infty} \frac{(n-i)}{n}\cdot \dots \cdot \frac{(n-i-k+j+1)}{n}&&\text{(there are $k-j$ terms)}\\
&= \lim_{n\to \infty} (1-t)\cdot \dots \cdot (1-t-\frac{k}{n}+\frac{j}{n}+\frac{1}{n})&&\text{(every term is $\frac{n-i+\text{ const.}}{n}$)}\\
&= (1-t)^{k-j}
\end{align*}

\begin{align*}
\lim_{n\to \infty} \frac{1}{n^{j-1}}\frac{i!}{(i-j+1)!} &= \lim_{n\to \infty} \frac{i\cdot \dots \cdot (i-j+2)}{n^{j-1}}\\
&= \lim_{n\to \infty} \frac{i}{n}\cdot \dots \cdot \frac{(i-j+2)}{n}&&\text{(there are $j-1$ terms)}\\
&= \lim_{n\to \infty} t\cdot \dots \cdot (t-\frac{j}{n}+\frac{2}{n})&&\text{(every term is $\frac{i+\text{ const.}}{n}$)}\\
&= t^{j-1}
\end{align*}

\begin{align*}
\lim_{n\to \infty} n^{k-1}\frac{(n-k+1)!}{n!} &= \lim_{n\to \infty} \frac{n^{k-1}}{n\cdot \dots \cdot (n-k+2)}\\
&= \lim_{n\to \infty} \frac{n}{n}\cdot \dots \cdot \frac{n}{n-k+2}&&\text{(there are $k-1$ terms)}\\
&=1
\end{align*}

Finally, multiply these three limits together, since the limit of products is the product of limits when each limit is convergent,
\begin{align*}
p(z_{t}=j) &= \binom{k-1}{j-1} (1-t)^{k-j} t^{j-1}
\end{align*}
\end{proof}
\begin{proposition}
\label{prop:represent}
    For state variables $\{z_{t_i}\}_{i=0}^n$ and segment lengths $\{\zeta_j\}_{j=1}^k$ the following equivalence representation holds:
    \begin{align*}
        \bm{1}(z_{t_i} = j) = \bm{1}\Bigg(\sum_{l=1}^{j-1}\zeta_l \leq t < \sum_{l=1}^{j}\zeta_l\Bigg). 
    \end{align*}
\end{proposition}

\begin{proof}[Proof of Proposition \ref{prop:represent}]
    Note the definition of $\zeta_j$ is the length of time between the first occurrence of state $j$ and state $j+1$.  Then $\sum_{l=1}^j \zeta_l$ is equal to the time of the first occurrence of state $j+1$.  As such, if $t_i$ is between the first time of state $j$ and the first time of state $j+1$, then by the definition of change point process $z_{t_i}=j$.
\end{proof}
\noindent We will need the following lemma to prove Theorem \ref{thm:represent}.
\begin{lemma}
\label{lemma:represent}
Let $t\in[0,1]$ and suppose $\{\zeta_j\}_{j=1}^k$ are the continuous time segment lengths that sum to $1$, then, 
    \begin{align*}
        \mathbb{P}(\sum_{l=1}^{j-1}\zeta_l \leq t < \sum_{l=1}^{j} \zeta_l) = \mathbb{P}(\sum_{l=1}^{j-1}\zeta_l \leq t) - \mathbb{P}(\sum_{l=1}^{j} \zeta_l \leq t )
    \end{align*}
\end{lemma}
\begin{proof}[Proof of Lemma \ref{lemma:represent}]
    \begin{align*}
    \mathbb{P}(\sum_{l=1}^{j-1}\zeta_l \leq t < \sum_{l=1}^{j} \zeta_l) &= \mathbb{P}((\sum_{l=1}^{j-1} \zeta_l \leq t) \cap (t < \sum_{l=1}^{j} \zeta_l)) \\
    &= \mathbb{P}(t < \sum_{l=1}^{j} \zeta_l |\sum_{l=1}^{j-1} \zeta_l \leq t)\mathbb{P}(\sum_{l=1}^{j-1} \zeta_l \leq t) \\
    &= (1 - \mathbb{P}(\sum_{l=1}^{j} \zeta_l \leq t | \sum_{l=1}^{j-1} \zeta_l \leq t))\mathbb{P}(\sum_{l=1}^{j-1} \zeta_l \leq t) \\
    &= (1 - \frac{\mathbb{P}(\sum_{l=1}^{j} \zeta_l \leq t \cap \sum_{l=1}^{j-1} \zeta_l \leq t)}{\mathbb{P}(\sum_{l=1}^{j-1}\zeta_l \leq t)} )\mathbb{P}(\sum_{l=1}^{j-1} \zeta_l \leq t) \\
    &= (1 - \frac{\mathbb{P}(\sum_{l=1}^{j} \zeta_l \leq t )}{\mathbb{P}(\sum_{l=1}^{j-1}\zeta_l \leq t)} )\mathbb{P}(\sum_{l=1}^{j-1} \zeta_l \leq t) &&(\sum_{l=1}^{j}\zeta_l\leq t) \to (\sum_{l=1}^{j-1}\zeta_l\leq t)\\
    &= \mathbb{P}(\sum_{l=1}^{j-1}\zeta_l \leq t) - \mathbb{P}(\sum_{l=1}^{j} \zeta_l \leq t )
\end{align*}
\end{proof}

\begin{proof}[Proof of Theorem \ref{thm:represent}]
Define the Bernstein polynomial $b_j(t)\coloneqq \binom{k-1}{j-1} (1-t)^{k-j} t^{j-1}$ and the distribution of the Dirichlet indicator $d_{j}(t) \coloneqq p(\sum_{l=1}^{j-1}\zeta_l \leq t < \sum_{l=1}^{j}\zeta_l)$. We wish to prove $d_j(t) = b_j(t)$. The proof proceeds as follows.  The first step is to evaluate $d_j(t)$ using the aggregation property of the Dirichlet distribution.  The next step is to argue that since neither function has an additive constant, it is sufficient to prove $\frac{d b_j(t)}{dt} = \frac{d d_j(t)}{dt}$. Or note $b_j(0)=d_j(0)$. Finally, we establish equality of the two derivatives. 

We start with $d_{j}(t)$. We have by Lemma \ref{lemma:represent}, $d_{j}(t) = (p(\sum_{l=1}^{j-1}\zeta_l \leq t) - p(\sum_{l=1}^{j}\zeta_l \leq t))$
These two cumulative distributions can be evaluated using the aggregation property of the Dirichlet. As such, we have $(\sum_{l=1}^{j-1}\zeta_l,\sum_{l=j}^{k}\zeta_l)' \sim \mathit{Dir}(j-1,k-j+1)$ and $(\sum_{l=1}^{j}\zeta_l,\sum_{l=j+1}^{k}\zeta_l)' \sim \mathit{Dir}(j,k-j)$. Thus, the cumulative distribution for the $(j-1)$th case is,
\begin{align*}
    p(\sum_{l=1}^{j-1}\zeta_l \leq t) = \frac{1}{B(j-1,k-j+1)}\int_0^t u^{(j-1)-1} (1 - u)^{(k-j+1)-1} du
\end{align*}
And similarly for the $j$th case.  The derivative of the Bernstein polynomial follows from the product and chain rules,
\begin{align}
\frac{db_j(t)}{dt} = \binom{k-1}{j-1}[(j-1)t^{j-2}(1-t)^{k-j} - (k-j)t^{j-1}(1-t)^{k-j-1}]
\end{align}
And the derivative of the Dirichlet probability interval follows from the fundamental theorem of calculus, 
\begin{align*}
    \frac{d p(\sum_{l=1}^{j-1}\zeta_l \leq t)}{dt} = \frac{1}{B(j-1,k-j+1)}t^{(j-1)-1} (1 - t)^{(k-j+1)-1}
\end{align*}
Which holds similarly for the $(j)$th case. As such the derivative of (1) is,
\begin{align}
    \frac{d (p(\sum_{l=1}^{j-1}\zeta_l \leq t) - p(\sum_{l=1}^{j}\zeta_l \leq t))}{dt} = \frac{1}{B(j-1,k-j+1)}t^{(j-1)-1} (1 - t)^{(k-j+1)-1} \nonumber \\ 
    - \frac{1}{B(j,k-j)}t^{j-1} (1 - t)^{(k-j)-1}
\end{align}
And now we can now show (2)=(3) as desired.
\begin{align*}
\frac{dd_j(t)}{dt} &= \frac{1}{B(j-1,k-j+1)} t^{(j-2)} (1 - t)^{(k-j)} \\
&\ \ \ \ \ \ \ \ - \frac{1}{B(j,k-j)}t^{j-1} (1 - t)^{(k-j)-1} \\
&= \frac{\Gamma(k)}{\Gamma(j-1)\Gamma(k-j+1)} t^{(j-2)} (1 - t)^{(k-j)} \\ 
&\ \ \ \ \ \ \ \ - \frac{\Gamma(k)}{\Gamma(j)\Gamma(k-j)}t^{(j-1)} (1 - t)^{(k-j-1)} \\
&= \frac{(k-1)!}{(j-2)!(k-j)!} t^{(j-2)} (1 - t)^{(k-j)} \\ 
&\ \ \ \ \ \ \ \ - \frac{(k-1)!}{(j-1)!(k-j-1)!}t^{(j-1)} (1 - t)^{(k-j-1)}\\
&= \binom{k-1}{j-1} (j-1) t^{(j-2)} (1 - t)^{(k-j)} \\ 
&\ \ \ \ \ \ \ \ - \binom{k-1}{j-1} (k-j) t^{(j-1)} (1 - t)^{(k-j-1)}\\
&= \frac{db_j(t)}{dt}
\end{align*}    
\end{proof}

\begin{proof}[Proof of Theorem \ref{thm:conttime}]
Using Theorem \ref{thm:represent}, note that the probability $p(z_{t} = h \,|\, z_{s} = j)$ is equivalent to the probability $p(\sum_{l=1}^{h-1}\zeta_l\leq t < \sum_{l=1}^h \zeta_l | \sum_{l=1}^{j-1}\zeta_l\leq s < \sum_{l=1}^j \zeta_l)$. As such we evaluate the joint probability of these events, and then divide by the marginal.\\
\textbf{Case 1: $h=j$}\\
    \begin{align*}
        &p(\sum_{l=1}^{j-1}\zeta_l\le s, t<\sum_{l=1}^{j}\zeta_l)\\ 
        &=\int_0^{s} \int_0^{s-\zeta_1} \dots \int_0^{s - \sum_{l=1}^{j-2}\zeta_l}\int_{t-\sum_{l=1}^{j-1}\zeta_l}^{1 - \sum_{l=1}^{j-1}\zeta_l}\int_0^{1-\sum_{l=1}^j \zeta_l}\dots\int_0^{1 - \sum_{l=1}^{k-2}\zeta_l} B^{-1}(1_k) \partial\zeta_{k-1}\dots\partial\zeta_1\\
        &=\int_0^{s} \int_0^{s-\zeta_1} \dots \int_0^{s - \sum_{l=1}^{j-2}\zeta_l}\int_{t-\sum_{l=1}^{j-1}\zeta_l}^{1 - \sum_{l=1}^{j-1}\zeta_l}\int_0^{1-\sum_{l=1}^j \zeta_l}\dots\int_0^{1 - \sum_{l=1}^{k-3}\zeta_l} B^{-1}(1_k)(1-\sum_{l=1}^{k-2}\zeta_l) \partial\zeta_{k-2}\dots\partial\zeta_1\\
        &=\int_0^{s} \int_0^{s-\zeta_1} \dots \int_0^{s - \sum_{l=1}^{j-2}\zeta_l}\int_{t-\sum_{l=1}^{j-1}\zeta_l}^{1 - \sum_{l=1}^{j-1}\zeta_l}\int_0^{1-\sum_{l=1}^j \zeta_l}\dots\int_0^{1 - \sum_{l=1}^{k-4}\zeta_l} B^{-1}(1_k)\frac{(1-\sum_{l=1}^{k-3}\zeta_l)^2}{2!} \partial\zeta_{k-3}\dots\partial\zeta_1\\
        &=\int_0^{s} \int_0^{s-\zeta_1} \dots \int_0^{s - \sum_{l=1}^{j-2}\zeta_l}\int_{t-\sum_{l=1}^{j-1}\zeta_l}^{1 - \sum_{l=1}^{j-1}\zeta_l}\int_0^{1-\sum_{l=1}^j \zeta_l}\dots\int_0^{1 - \sum_{l=1}^{k-5}\zeta_l} B^{-1}(1_k)\frac{(1-\sum_{l=1}^{k-4}\zeta_l)^3}{3!} \partial\zeta_{k-4}\dots\partial\zeta_1\\
        &=\int_0^{s} \int_0^{s-\zeta_1} \dots \int_0^{s - \sum_{l=1}^{j-2}\zeta_l}\int_{t-\sum_{l=1}^{j-1}\zeta_l}^{1 - \sum_{l=1}^{j-1}\zeta_l} B^{-1}(1_k)\frac{(1-\sum_{l=1}^{j}\zeta_l)^{k-j-1}}{(k-j-1)!} \partial\zeta_{j}\dots\partial\zeta_1\\
        &=\int_0^{s} \int_0^{s-\zeta_1} \dots \int_0^{s - \sum_{l=1}^{j-2}\zeta_l}B^{-1}(1_k)\frac{(1-t)^{k-j}}{(k-j)!} \partial\zeta_{j-1}\dots\partial\zeta_1\\
        &=\int_0^{s} \int_0^{s-\zeta_1} \dots \int_0^{s - \sum_{l=1}^{j-3}\zeta_l} (s - \sum_{l=1}^{j-2}\zeta_l) B^{-1}(1_k)\frac{(1-t)^{k-j}}{(k-j)!} \partial\zeta_{j-2}\dots\partial\zeta_1\\
        &=\int_0^{s} \int_0^{s-\zeta_1} \dots \int_0^{s - \sum_{l=1}^{j-4}\zeta_l} \frac{(s - \sum_{l=1}^{j-3}\zeta_l)^2}{2!} B^{-1}(1_k)\frac{(1-t)^{k-j}}{(k-j)!} \partial\zeta_{j-3}\dots\partial\zeta_1\\
        &=\int_0^{s} \int_0^{s-\zeta_1} \frac{(s - \sum_{l=1}^{2}\zeta_l)^{j-3}}{(j-3)!} B^{-1}(1_k)\frac{(1-t)^{k-j}}{(k-j)!} \partial\zeta_{2}\partial\zeta_1\\
        &=\int_0^{s} \frac{(s - \zeta_1)^{j-2}}{(j-2)!} B^{-1}(1_k)\frac{(1-t)^{k-j}}{(k-j)!} \partial\zeta_1\\
        &=B^{-1}(1_k)\frac{s^{j-1}}{(j-1)!}\frac{(1-t)^{k-j}}{(k-j)!}\\
        &=\binom{k-1}{j-1} s^{j-1}(1-t)^{k-j}
    \end{align*}
Then using this as the numerator of $p(z_{t}=j |z_{s}=j)$, and using the continuous time marginal of $z_{s}$ from Theorem \ref{thm:hgeobern},
\begin{align*}
    p(z_{t}=j |z_{s}=j) &=\frac{\binom{k-1}{j-1} s^{(j-1)}(1-t)^{(k-j)}}{\binom{k-1}{j-1} s^{(j-1)}(1-s)^{(k-j)}}\\
    &=\left(\frac{1-t}{1-s}\right)^{(k-j)}\\
    &=\binom{k-j}{j-j}\left(1-\frac{1-t}{1-s}\right)^{(j-j)}\left(\frac{1-t}{1-s}\right)^{(k-j)}
\end{align*}
\textbf{Case 2: $h=j+1$}\\
Now we derive the transition from $j$ to $j+1$.  Start with the integrand in the joint numerator,
\begin{align*}
        &p(\sum_{l=1}^{j-1}\zeta_l \le s <\sum_{l=1}^{j}\zeta_l \le t <\sum_{l=1}^{j+1}\zeta_l)\\ 
        &=\int_0^{s} \int_0^{s-\zeta_1} \dots \int_0^{s - \sum_{l=1}^{j-2}\zeta_l}\int_{s-\sum_{l=1}^{j-1}\zeta_l}^{t - \sum_{l=1}^{j-1}\zeta_l}\int_{t-\sum_{l=1}^j \zeta_l}^{1-\sum_{l=1}^j \zeta_l}\int_0^{1 - \sum_{l=1}^{j+1}\zeta_l}\dots\int_0^{1-\sum_{l=1}^{k-2}\zeta_l} B^{-1}(1_k) \partial\zeta_{k-1}\dots\partial\zeta_1\\
        &=\int_0^{s} \int_0^{s-\zeta_1} \dots \int_0^{s - \sum_{l=1}^{j-2}\zeta_l}\int_{s-\sum_{l=1}^{j-1}\zeta_l}^{t - \sum_{l=1}^{j-1}\zeta_l}\int_{t-\sum_{l=1}^j \zeta_l}^{1-\sum_{l=1}^j \zeta_l} \frac{(1-\sum_{l=1}^{j+1}\zeta_l)^{k-j-2}}{(k-j-2)!}B^{-1}(1_k) \partial\zeta_{j+1}\dots\partial\zeta_1\\
        &=\int_0^{s} \int_0^{s-\zeta_1} \dots \int_0^{s - \sum_{l=1}^{j-2}\zeta_l}\int_{s-\sum_{l=1}^{j-1}\zeta_l}^{t - \sum_{l=1}^{j-1}\zeta_l} \frac{(1-t)^{k-j-1}}{(k-j-1)!}B^{-1}(1_k) \partial\zeta_{j}\dots\partial\zeta_1\\
        &=\int_0^{s} \int_0^{s-\zeta_1} \dots \int_0^{s - \sum_{l=1}^{j-2}\zeta_l} (t-s)\frac{(1-t)^{k-j-1}}{(k-j-1)!}B^{-1}(1_k) \partial\zeta_{j-1}\dots\partial\zeta_1\\
        &=\frac{s^{j-1}}{(j-1)!}(t-s)\frac{(1-t)^{k-j-1}}{(k-j-1)!}B^{-1}(1_k)\\
        &=\binom{k-1}{j-1}s^{j-1}(t-s)(1-t)^{k-j-1}(k-j)
\end{align*}
Then the transition probability is given by dividing the marginal probability of $z_s=j$,
\begin{align*}
    p(z_{t}=j+1|z_{s}=j)&= \frac{\binom{k-1}{j-1}s^{j-1}(t-s)(1-t)^{k-j-1}(k-j)}{\binom{k-1}{j-1}s^{j-1}(1-s)^{k-j}}\\
    &=(k-j)\frac{t-s}{1-s}\left(\frac{1-t}{1-s}\right)^{k-j-1}\\
    &=\binom{k-j}{j+1-j}\Bigg(1-\frac{1-t}{1-s}\Bigg)^{(j+1-j)}\left(\frac{1-t}{1-s}\right)^{k-j-1}
\end{align*}
\textbf{Case 3: $h>j+1$}
\begin{align*}
    &p(\sum_{l=1}^{j-1}\zeta_l\leq s < \sum_{l=1}^j \zeta_l \leq \sum_{l=1}^{h-1}\zeta_l\leq t < \sum_{l=1}^h \zeta_l)\\
    &=\int_0^s\int_0^{s-\zeta_1}\dots\int_0^{s - \sum_{l=1}^{j-2}\zeta_l}\int_{s-\sum_{l=1}^{j-1}\zeta_l}^{t - \sum_{l=1}^{j-1}\zeta_l}\int_0^{t - \sum_{l=1}^j \zeta_l}\dots\int_0^{t - \sum_{l=1}^{h-2}\zeta_l}\int_{t-\sum_{l=1}^{h-1}\zeta_l}^{1-\sum_{l=1}^{h-1}\zeta_l}\int_{0}^{1-\sum_{l=1}^{h}\zeta_l}\dots \int_0^{1-\sum_{l=1}^{k-2}\zeta_l}\\&\quad B^{-1}(1_k)\partial\zeta_{k-1}\dots\partial\zeta_{h+1}\partial\zeta_h\partial\zeta_{h-1}\dots\partial\zeta_{j+1}\partial\zeta_j\partial\zeta_{j-1}\dots\partial\zeta_2\partial\zeta_1\\
    &=\int_0^s\int_0^{s-\zeta_1}\dots\int_0^{s - \sum_{l=1}^{j-2}\zeta_l}\int_{s-\sum_{l=1}^{j-1}\zeta_l}^{t - \sum_{l=1}^{j-1}\zeta_l}\int_0^{t - \sum_{l=1}^j \zeta_l}\dots\int_0^{t - \sum_{l=1}^{h-2}\zeta_l}\int_{t-\sum_{l=1}^{h-1}\zeta_l}^{1-\sum_{l=1}^{h-1}\zeta_l}\\ &\quad \frac{(1-\sum_{l=1}^h\zeta_l)^{k-h-1}}{(k-h-1)!} B^{-1}(1_k)\partial\zeta_h\partial\zeta_{h-1}\dots\partial\zeta_{j+1}\partial\zeta_j\partial\zeta_{j-1}\dots\partial\zeta_2\partial\zeta_1\\
    &=\int_0^s\int_0^{s-\zeta_1}\dots\int_0^{s - \sum_{l=1}^{j-2}\zeta_l}\int_{s-\sum_{l=1}^{j-1}\zeta_l}^{t - \sum_{l=1}^{j-1}\zeta_l}\int_0^{t - \sum_{l=1}^j \zeta_l}\dots\int_0^{t - \sum_{l=1}^{h-2}\zeta_l}\\ &\quad \frac{(1-t)^{k-h}}{(k-h)!} B^{-1}(1_k)\partial\zeta_{h-1}\dots\partial\zeta_{j+1}\partial\zeta_j\partial\zeta_{j-1}\dots\partial\zeta_2\partial\zeta_1\\
    &=\int_0^s\int_0^{s-\zeta_1}\dots\int_0^{s - \sum_{l=1}^{j-2}\zeta_l}\int_{s-\sum_{l=1}^{j-1}\zeta_l}^{t - \sum_{l=1}^{j-1}\zeta_l} \frac{(t - \sum_{l=1}^j \zeta_l)^{h-j-1}}{(h-j-1)!}\frac{(1-t)^{k-h}}{(k-h)!} B^{-1}(1_k)\partial\zeta_j\partial\zeta_{j-1}\dots\partial\zeta_2\partial\zeta_1\\
    &=\int_0^s\int_0^{s-\zeta_1}\dots\int_0^{s - \sum_{l=1}^{j-2}\zeta_l} \frac{(t - s)^{h-j}}{(h-j)!}\frac{(1-t)^{k-h}}{(k-h)!} B^{-1}(1_k)\partial\zeta_{j-1}\dots\partial\zeta_2\partial\zeta_1\\
    &=\frac{s^{j-1}}{(j-1)!} \frac{(t - s)^{h-j}}{(h-j)!}\frac{(1-t)^{k-h}}{(k-h)!} B^{-1}(1_k)
\end{align*}
Now divide by the marginal $z_s=j$ to get the transition probability,
\begin{align*}
    p(z_t=h|z_s=j) &=\frac{\frac{s^{j-1}}{(j-1)!} \frac{(t - s)^{h-j}}{(h-j)!}\frac{(1-t)^{k-h}}{(k-h)!}(k-1)!}{\frac{(k-1)!}{(j-1)!(k-j)!}s^{j-1}(1-s)^{k-j}}\\
    &=\frac{\frac{(t - s)^{h-j}}{(h-j)!}\frac{(1-t)^{k-h}}{(k-h)!}}{\frac{1}{(k-j)!}(1-s)^{k-j}}\\
    &=\binom{k-j}{h-j}\frac{(t - s)^{h-j}(1-t)^{k-h}}{(1-s)^{k-j}}\\
    &=\binom{k-j}{h-j}\Bigg(\frac{t - s}{1-s}\Bigg)^{h-j}\Bigg(\frac{1-t}{1-s}\Bigg)^{k-h}\\
    &=\binom{k-j}{h-j}\Bigg(1-\frac{1-t}{1-s}\Bigg)^{h-j}\Bigg(\frac{1-t}{1-s}\Bigg)^{k-h}
\end{align*}
\end{proof}
\begin{proof}[Proof of Theorem \ref{thm:conttime} Kolmogorov Equations]
\begin{align*}
    \sum_{l = j}^h P_{jl}(s, r)P_{lh}(r, t)
    &= \sum_{l = j}^h \binom{k-j}{l-j}\Bigg(1-\frac{1-r}{1-s}\Bigg)^{l-j}\Bigg(\frac{1-r}{1-s}\Bigg)^{k-l} \binom{k-l}{h-l}\Bigg(1-\frac{1-t}{1-r}\Bigg)^{h-l}\Bigg(\frac{1-t}{1-r}\Bigg)^{k-h}\\
    &= \Bigg(\frac{1-t}{1-s}\Bigg)^{k-h}\sum_{l = j}^h \binom{k-j}{l-j}\binom{k-l}{h-l}\Bigg(\frac{r-s}{1-s}\Bigg)^{l-j}\Bigg(\frac{1-r}{1-s}\Bigg)^{h-l} \Bigg(\frac{t-r}{1-r}\Bigg)^{h-l}\\
    &= \Bigg(\frac{1-t}{1-s}\Bigg)^{k-h}\sum_{l = j}^h \binom{k-j}{l-j}\binom{k-l}{h-l}\Bigg(\frac{r-s}{1-s}\Bigg)^{l-j}\Bigg(\frac{t-r}{1-s}\Bigg)^{h-l}\\
    &= \Bigg(\frac{1-t}{1-s}\Bigg)^{k-h}\Bigg(\frac{1}{1-s}\Bigg)^{h-j}\sum_{l = j}^h \binom{k-j}{l-j}\binom{k-l}{h-l}(r-s)^{l-j}(t-r)^{h-l}\\
    &= \Bigg(\frac{1-t}{1-s}\Bigg)^{k-h}\Bigg(\frac{1}{1-s}\Bigg)^{h-j}\sum_{l = j}^h \frac{(k-j)!}{(l-j)!(k-l)!} \frac{(k-l)!}{(h-l)!(k-h)!}(r-s)^{l-j}(t-r)^{h-l}\\
    &= \binom{k-j}{h-j}\Bigg(\frac{1-t}{1-s}\Bigg)^{k-h}\Bigg(\frac{1}{1-s}\Bigg)^{h-j}\sum_{l = j}^h \binom{h-j}{l-j}(r-s)^{l-j}(t-r)^{h-l}\\
    &= \binom{k-j}{h-j}\Bigg(\frac{1-t}{1-s}\Bigg)^{k-h}\Bigg(\frac{t-s}{1-s}\Bigg)^{h-j}\\
    &=P_{jh}(s,t)
\end{align*}
\end{proof}
\begin{proof}[Proof of Theorem \ref{thm:bayesEst}]
This is a constrained optimization problem since we need to find the configuration $\bm{z}$ that minimizes the expected loss subject to being a change point process.  We first derive the Bayes estimator in the unconstrained space(which contains the constrained space). We then show, despite that we found the estimator in the bigger unconstrained space, the estimator yields a change point process almost surely, satisfying the constraint. 

\textit{Estimator for the unconstrained space}\\
Since the loss is a sum over $i=1,\dots,n$, and since we are operating in the unconstrained space, the problem reduces to finding the Bayes estimator separately for each $z_{t_i}$,
\[
    \underset{z_{t_i}}{\arg\min}\ \mathbb{E}_{\bm{z}^*|\bm{y}}\bigg[|z_{t_i} - z_{t_i}^*|(t_i - t_{i-1})\bigg] = \underset{z_{t_i}}{\arg\min}\ \mathbb{E}_{\bm{z}^*|\bm{y}}\bigg[|z_{t_i} - z_{t_i}^*|\bigg]
\]
Since $(t_i - t_{i-1})$ is a constant.  It is well known the Bayes estimator for absolute loss is the median. Thus, the Bayes estimator for the unconstrained problem is 
\[
    \hat{z}_{t_i} = \underset{j}{\text{min}}\bigg( \sum_{k = 1}^K\sum_{l=1}^j p(z_{t_i}=l|k,y)p(k|y) \geq 0.5 \bigg)
\]

\textit{Show this estimator is a change point process: discrete time case}\\
Now we show this estimator is a change point process with probability 1.  The proof strategy is to show for arbitrary median $\hat{z}_{t_i}$, that the median $ \hat{z}_{t_{i+1}} \in \{\hat{z}_{t_i},\hat{z}_{t_i}+1\}$ with probability 1 under the posterior measure of interest. To that end, let $\Omega$ be the set of all change point process sample points $\omega$ with positive support under the prior on $\bm{z}$. These configurations represent a superset of the configurations with positive support under the posterior measure.  Let $\hat{z}_{t_i} = \text{median}(z_{t_i})$ under the posterior measure be arbitrary. By definition of median,
\begin{align*}
p(z_{t_i}(\omega)\leq \hat{z}_{t_i}|\bm{y})\geq0.5\\
\text{  and  }\\
p(z_{t_i}(\omega)\geq \hat{z}_{t_i}|\bm{y})\geq0.5
\end{align*}
Since $\omega$ is a change point process, 
\[\omega \in \Omega\big(z_{t_i}=\hat{z}_{t_i}\big) \implies \omega\in \Omega\big(z_{t_{i+1}}\in\{\hat{z}_{t_i},\hat{z}_{t_i}+1\}\big)\]
with probability 1, where we define $\Omega(A)$ as the subset of the sample space $\Omega$ where the condition $A$ is true. The above then also implies, 
\[\omega \in \Omega\big(z_{t_i}\leq \hat{z}_{t_i}\big) \implies \omega\in \Omega\big(z_{t_{i+1}}\leq \hat{z}_{t_i} \text{ or } z_{t_{i+1}}\leq \hat{z}_{t_i}+1\big)\] with probability 1.
Plugging these implications back into the probability inequalities that define the median, and using the fact that $A\subset B$ implies $p(A)\leq p(B)$,
\begin{align*}
p(z_{t_{i+1}}(\omega)\leq \hat{z}_{t_i}\text{ or }z_{t_{i+1}}(\omega)\leq \hat{z}_{t_i}+1|\bm{y})\geq0.5\\
\text{ and }\\
p(z_{t_{i+1}}(\omega)\geq \hat{z}_{t_i}\text{ or }z_{t_{i+1}}(\omega)\geq \hat{z}_{t_i}+1|\bm{y})\geq0.5
\end{align*}
The "or"-events in these probabilities can be reduced to mutual exclusivity by removing their intersection as follows,
\begin{align*}
p(z_{t_{i+1}}(\omega)\leq \hat{z}_{t_i}|\bm{y})+p(z_{t_{i+1}}(\omega)= \hat{z}_{t_i}+1|\bm{y})\geq0.5\\
\text{ and }\\
p(z_{t_{i+1}}(\omega)= \hat{z}_{t_i}|\bm{y})+p(z_{t_{i+1}}(\omega)\geq \hat{z}_{t_i}+1|\bm{y})\geq0.5
\end{align*}
We are now in a position to determine the median of $z_{t_{i+1}}$. The first case is when $p(z_{t_{i+1}}(\omega)\leq \hat{z}_{t_i}|\bm{y})>0.5$, which implies $p(z_{t_{i+1}}(\omega)\geq \hat{z}_{t_i}+1|\bm{y})<0.5$ since the two probabilities sum to 1. In this case, we have $\hat{z}_{t_{i+1}} = \hat{z}_{t_{i}}$ since, 
\begin{align*}
p(z_{t_{i+1}}(\omega)\leq \hat{z}_{t_i}|\bm{y})\geq0.5\\
\text{ and }\\
p(z_{t_{i+1}}(\omega)= \hat{z}_{t_i}|\bm{y})+p(z_{t_{i+1}}(\omega)\geq \hat{z}_{t_i}+1|\bm{y})\geq0.5
\end{align*}

The second case is when $p(z_{t_{i+1}}(\omega)\leq \hat{z}_{t_i}|\bm{y})<0.5$, which implies $p(z_{t_{i+1}}(\omega)\geq \hat{z}_{t_i}+1|\bm{y})>0.5$ since the two probabilities sum to 1. In this case we have $\hat{z}_{t_{i+1}} = \hat{z}_{t_{i}}+1$ since, 
\begin{align*}
p(z_{t_{i+1}}(\omega)\leq \hat{z}_{t_i}|\bm{y})+p(z_{t_{i+1}}(\omega)= \hat{z}_{t_i}+1|\bm{y})\geq0.5\\
\text{ and }\\
p(z_{t_{i+1}}(\omega)\geq \hat{z}_{t_i}+1|\bm{y})\geq0.5
\end{align*}
The last case, when $p(z_{t_{i+1}}(\omega)\leq \hat{z}_{t_i}|\bm{y})=0.5$, follows similarly. Thus, in all cases, the median $\hat{z}_{t_{i+1}}$ is either  $\hat{z}_{t_{i}}$ or $\hat{z}_{t_{i}}+1$ with probability 1, and the resulting estimator is the Bayes estimator for the weighted Hamming loss in the constrained space of change point processes in discrete time. 

\textit{Show this estimator is a change point process: continuous time case}\\
In continuous time, more than one change point can occur between two consecutive observations, so the proof changes slightly. Suppose the median at time $t_i$ is $\hat{z}_{t_i}$. Let $\omega$ be a continuous time change point process such that $\omega\in\Omega(z_{t_i}=\hat{z}_{t_i})$. This implies $\omega\in\Omega(z_{t_{i+1}}\in\{\hat{z}_{t_i},\dots,K\})$. Furthermore, extending these statements with inequalities, we have, $\omega\in\Omega(z_{t_i}\leq\hat{z}_{t_i})$ implies $\omega\in\Omega(z_{t_{i+1}}\leq \hat{z}_{t_i}\text{ or }z_{t_{i+1}}\in\{\hat{z}_{t_i}+1,\dots,K\})$ and similarly $\omega\in\Omega(z_{t_i}\geq\hat{z}_{t_i})$ implies $\omega\in\Omega(z_{t_{i+1}}\in\{\hat{z}_{t_i},\dots,K\})$. Using the fact that $A\subset B$ implies $p(A)\leq p(B)$, applying these results to the definition of median,
\begin{align*}
p(z_{t_{i+1}}(\omega)\leq \hat{z}_{t_i}|\bm{y})+ \sum_{j = \hat{z}_{t_i}+1}^K p(z_{t_{i+1}}(\omega)= j|\bm{y})\geq0.5\\
\text{ and }\\
p(z_{t_{i+1}}(\omega)= \hat{z}_{t_i}|\bm{y}) + \sum_{j = \hat{z}_{t_i}+1}^K p(z_{t_{i+1}}(\omega)= j|\bm{y})\geq0.5
\end{align*}
The first of those inequalities is trivial since the probabilities sum to one, but is also constructive for the proof. 
Now proceed ruling out each possibility as we did in the discrete case. If $p(z_{t_{i+1}}(\omega)\leq \hat{z}_{t_i}|\bm{y})>0.5$ then the second summation in the second equation is less than 0.5, and the median is $\hat{z}_{t_{i+1}} = \hat{z}_{t_i}$. 

Proceeding iteratively, now suppose $p(z_{t_{i+1}}(\omega)\leq \hat{z}_{t_i}|\bm{y})<0.5$. Then, by the first equation, $\sum_{j = \hat{z}_{t_i}+1}^K p(z_{t_{i+1}}(\omega)= j|\bm{y})\geq0.5$. But since $p(z_{t_{i+1}}(\omega)\leq \hat{z}_{t_i}|\bm{y})<0.5$ then by the sum of probability to 1,  $\sum_{j = \hat{z}_{t_i}+1}^K p(z_{t_{i+1}}(\omega)= j|\bm{y})\geq0.5$ and the median $\hat{z}_{t_{i+1}} = \hat{z}_{t_i}+1$. 

This process iterates and we conclude that $\hat{z}_{t_{i+1}} \in \{\hat{z}_{t_i},\dots,K\}$ with probability 1.
\end{proof}

\begin{proof}[Proof of Proposition \ref{prop:IAVH}]
WLOG, let $t^*=365/T$ be the starting time of the second year. Enforcing continuity of the mean function requires setting its left limit equal to its right limit at time $t^*$. Limiting from the left, the harmonic contrast coefficients are zero since there is no contrast in the first year.  We also have that the $\sin$ terms are zero and $\cos$ terms are 1 at $t^*$, thus,
\begin{align*}
    &= \lim_{t\to t^{*-}} \mu(t) + \sum_{h=1}^H \gamma_{h,j(t)} \sin(h\omega t) + \delta_{h,j(t)} \cos(h\omega t) && \text{(left hand limit)}\\
    &= \lim_{t\to t^{*-}} \mu(t)&& \text{(contrasts are 0 in first year)}\\
    &= \alpha + \beta t^* + \sum_{h=1}^H \delta_{h} && \text{($\sin$ terms 0, $\cos$ terms 1 at $t^*$)}
\end{align*}

From the right, the contrast coefficients for the second year may not be 0. We have all $\sin$ terms are zero at the limit and $\cos$ terms are 1,
\begin{align*}
    &= \lim_{t\to t^{*+}} \mu(t) + \sum_{h=1}^H \gamma_{h,j(t)} \sin(h\omega t) + \delta_{h,j(t)} \cos(h\omega t) \\
    &= \alpha + \beta t^* + \sum_{h=1}^H \delta_{h} + \sum_{h=1}^H \delta_{h,1}
\end{align*}

Setting these two limits equal we arrive at the first result, 
\begin{align}\label{cl1e1}
    \sum_{h=1}^H \delta_{h,1} = 0 
\end{align}
Now let $t^*$ be the starting time of the second year. Using similar arguments above, we arrive at the equality,
\begin{align*}
    \sum_{h=1}^H \delta_{h,1} &= \sum_{h=1}^H \delta_{h,2}\\
    0 &= \sum_{h=1}^H \delta_{h,2} &&\text{(from Equation \ref{cl1e1})}
\end{align*}
Thus, the continuity constraint at the starting time of each $j$th year of the mean function leads to the constraints,
\begin{align*}
    \delta_{H,j} = -\sum_{h=1}^{H-1} \delta_{h,j}
\end{align*}
Using a similar argument for enforcing continuity of the derivative of the mean function, we have,
\begin{align*}
    \gamma_{H,j} = -\sum_{h=1}^{H-1} \gamma_{h,j}
\end{align*}
\end{proof}

\section{Appendix B: Noninformative Segment Lengths in Discrete and Continuous Time}

We wish to derive the noninformative distribution of $\{\zeta_j\}_{j=1}^k$ in the discrete time case. 
Whereas, in the hypergeometric distribution, the number of samples until success is considered fixed and the number of successes at that time is considered random,  we wish to relate this distribution to one where the segment lengths are random-- that is, the number of samples until a specified number of successes is random. With this in mind, define the Inverse Hypergeometric Distribution as the distribution on the number of samples $i$ until the $j$th success, with population size $n$ and $k$ total successes in the population.  We first derive the $(n,J,1)$-Inverse Hypergeometric Distribution, that is, the distribution of the length until the first success.
\begin{proposition}[$(n,J,1)$-Inverse Hypergeometric Distribution]
\label{prop:IHG1}
Suppose there is an urn with population $n$ and $J$ total successes. The distribution of the length until first success is,
\begin{align*}
p(\zeta_1=i) = \frac{J}{n-(i-1)}\cdot \frac{\binom{n-J}{i-1}}{\binom{n}{i-1}}
\end{align*}
\end{proposition}
\begin{proof}[Proof of Proposition \ref{prop:IHG1}]
    Choose the first $i-1$ draws out of the possible $n-J$ failures. The denominator of those first $i-1$ draws is all the ways to choose $i-1$ draws from population $n$. Then, conditioned on the first $i-1$ failures, the probability that the $i$th draw is a success is $J/(n-(i-1))$.
\end{proof}

Using this distribution, derive the general case by conditioning on one success at a time,

\begin{theorem}[$(n,J,j)$-Inverse Hypergeometric Distribution]
\label{thm:IHG}
Suppose there is an urn with population $n$ and total successes $J$. Define $\zeta_0 = 0$. The distribution of $\{\zeta_l\}_{l=1}^{j}$, the first $j$ consecutive lengths-until-success, is the product, 
\begin{align*}
    p(\{\zeta_l=i^{(l)}\}_{l=1}^j)
    &=\prod_{l=1}^j \mathit{IHG}((n-\sum_{m=0}^{l-1}\zeta_m),(J-l+1),1)
\end{align*}
We say $\{\zeta_l=i^{(l)}\}_{l=1}^j$ is $\mathit{IHG}(n,J,j)$ distributed. In the special case of $J=k-1$ and $j = k-1$, we have,
\begin{align*}
    p(\{\zeta_l=i^{(l)}\}_{l=1}^{k-1})&=\prod_{l=1}^{k-1} \mathit{IHG}((n-\sum_{m=0}^{l-1}\zeta_m),(k-l),1)\\ 
    &=\frac{(k-1)!}{n(n-1)\dots (n-(k-2))}\\
    &=\frac{1}{\binom{n}{k-1}}
\end{align*}
\end{theorem}
\begin{proof}[Proof of Theorem \ref{thm:IHG}]
From Proposition \ref{prop:IHG1}, $p(\zeta_{1}=i^{(1)})$ is $(n,J,1)$-Inverse Hypergeometric Distributed. Note that conditioned on $\zeta_1=i^{(1)}$, the population is now $n-i^{(1)}$ and the remaining total number of successes is $J-1$, thus, \begin{align*}
    p(\zeta_2=i^{(2)}|\zeta_1=i^{(1)}) = \mathit{IHG}(n-i^{(1)},J-1,1)
\end{align*}
Note that in the general case, for $p(\zeta_l=i^{(l)}|\{\zeta_m=i^{(m)}\}_{m=1}^{l-1})$, a similar argument holds.  The population is reduced to $n - \sum_{m=0}^{l-1}\zeta_l$ and the number of successes is reduced to $J-l+1$.  Thus, using the law of conditional probability, the result is a product of inverse hypergeometric distributions as written in the statement. Now to prove the simplification of this statement, consider the cases of $k=2$ and $k=3$ as follows. 
Let $\{\zeta_j\}_{j=1}^{k-1}$ be $\mathit{IHG}(n,k-1,k-1)$ distributed. Suppose $k = 2$, then,
\begin{align*}
p(\zeta_1 = i_1) &= \frac{1}{n-i_1+1} \cdot  \frac{\binom{n-1}{i_1-1}}{\binom{n}{i_1-1}}\\
    &= \frac{1}{n-i_1+1} \cdot  \frac{(n-1)!(n-i_1+1)!}{(n-i_1)!n!}\\
    &= \frac{1!}{n}\\ 
\end{align*}
When $k=3$, observe the following telescopic cancellation,
\begin{align*}
    p(\zeta_1 = i_1, \zeta_2 = i_2) &= p(\zeta_1 = i_1)p(\zeta_2 = i_2|\zeta_1 = i_1)\\
    &= (\frac{2}{n-i_1+1} \frac{\binom{n-2}{i_1-1}}{\binom{n}{i_1-1}}) \cdot  (\frac{1}{n-i_1-i_2+1} \frac{\binom{n-i_1-1}{i_2-1}}{\binom{n-i_1}{i_2-1}})\\
    &= (\frac{2}{n-i_1+1}\frac{(n-2)!(n-i_1+1)!}{(n-i_1-1)!n!})\cdot (\frac{1}{n-i_1-i_2+1} \frac{(n-i_1-1)!(n-i_1-i_2+1)!}{(n-i_1-i_2)!(n-i_1)!})\\
    &= \frac{2(n-i_1)}{n(n-1)}\cdot \frac{1}{(n-i_1)}\\
    &= \frac{2!}{n(n-1)}\\
\end{align*}
For general $k$, using the same telescoping cancellation approach, notice,
\begin{align*}
    p(\{\zeta_j = i_j\}_{j=1}^{k-1}) &= \frac{k-1}{n-i_1+1}\frac{\binom{n-(k-1)}{i_1 - 1}}{\binom{n}{i_1-1}} \cdot  \frac{k-2}{n-i_1-i_2+1}\frac{\binom{n-i_1-(k-2)}{i_2 - 1}}{\binom{n-i_1}{i_2-1}}\cdot p(\{\zeta_j\}_{j=3}^{k-1}|\{\zeta_j\}_{j=1}^2)\\
    &=\frac{k-1}{n-i_1+1}\frac{(n-(k-1))!(n-i_1+1)!}{n!(n-k-i_1+2)!}\\
    &\ \ \ \ \cdot \frac{k-2}{n-i_1-i_2+1}\frac{(n-i_1-k+2)!(n-i_1-i_2+1)!}{(n-i_1)!(n-i_1-k-i_2+3)!}\cdot p(\{\zeta_j\}_{j=3}^{k-1}|\{\zeta_j\}_{j=1}^2)\\
    &=\frac{(k-1)(k-2)}{n(n-1)\dots (n-(k-2))}\cdot \frac{(n-i_1-i_2)!}{(n-i_1-k-i_2+3)!}\cdot p(\{\zeta_j\}_{j=3}^{k-1}|\{\zeta_j\}_{j=1}^2)\\
\end{align*}
There are three points to make here.  The first is that the numerator is recursively forming $(k-1)!$.  The second is that the first denominator is already equal to $(n(n-1)\dots (n-(k-2)))^{-1}$.  Finally, the term in the middle can be rewritten in terms of $j$, in order to understand how it changes during recursion,
\begin{align*}
    \frac{(n-i_1-i_2)!}{(n-i_1-k-i_2+3)!} = \frac{(n-\sum_{l=1}^{j}i_l)!}{(n-(\sum_{l=1}^j i_l) -k+(j+1))!}
\end{align*}
Using this equation, after recursing through $k-2$ conditional probabilities, we arrive at,
\begin{align*}
    &=\frac{(k-1)!}{n(n-1)\dots (n-(k-2))}\cdot \frac{(n-\sum_{l=1}^{k-2}i_l)!}{(n-(\sum_{l=1}^{k-2}i_l)-k+((k-2)+1))!}\cdot p(\zeta_{k-1}|\{\zeta_j\}_{j=1}^{k-2})\\
    &=\frac{(k-1)!}{n(n-1)\dots (n-(k-2))}\cdot \frac{(n-\sum_{l=1}^{k-2}i_l)!}{(n-(\sum_{l=1}^{k-2}i_l)-1)!}\cdot p(\zeta_{k-1}|\{\zeta_j\}_{j=1}^{k-2})\\
    &=\frac{(k-1)!}{n(n-1)\dots (n-(k-2))}\cdot \frac{(n-\sum_{l=1}^{k-2}i_l)!}{(n-(\sum_{l=1}^{k-2}i_l)-1)!}\\
    &\hspace{5em} \cdot \frac{1}{n-\sum_{l=1}^{k-1}i_l +1}\frac{(n-\sum_{l=1}^{k-2}i_l-1)!(n-\sum_{l=1}^{k-1}i_l+1)!}{(n-\sum_{l=1}^{k-2}i_l)!(n-\sum_{l=1}^{k-1}i_l)!}\\
    &=\frac{(k-1)!}{n(n-1)\dots (n-(k-2))}
\end{align*}
\end{proof}

The last part of this theorem confirms the Inverse-Hypergeometric distribution is the noninformative prior on discrete segment lengths, as it measures each change point process sample point with equal probability $\frac{1}{\binom{n}{k-1}}$.

From the other direction, we ought to expect that the Inverse Hypergeometric distribution in discrete time converges in distribution to the Dirichlet as well. This is indeed the case,

\begin{theorem}[Noninformative segment Length Convergence: Inverse Hypergeometric to Dirichlet]
\label{thm:IHGDir}
Let $\{\zeta^*_j\}_{j=1}^{k-1} \in (0,1)$ be arbitrary having $\sum_{j=1}^{k-1} \zeta^*_j < 1$. Define the corresponding discrete case as $\zeta_j = \floor{(n-j+1) \zeta^*_j}$ for all $j$. Then the inverse hypergeometric segment lengths converge in distribution to the noninformative continuous time distribution on segment lengths, Dirichlet $\bm{1}_k$ as $n\to\infty$,
\begin{align*}
    F_{IHG}(\{\zeta_j\}_{j=1}^{k-1}) \to F_{Dir}(\{\zeta_j^*\}_{j=1}^{k-1};\bm{1}_k)
\end{align*}
Where $F$ denotes the distribution function and the IHG distribution is parameterized as in the noninformative case with population $n$, $k-1$ total successes, and samples until $k-1$ successes.
\end{theorem}  

\begin{proof}[Proof of Theorem \ref{thm:IHGDir}]
    Let $\{\zeta^*_j\}_{j=1}^k \in (0,1)$ such that $\sum_{j=1}^k \zeta^*_j = 1$ be otherwise arbitrary. Define the discretization $\zeta_j = \floor{(n-j+1) \zeta^*_j}$. Note, $\{\zeta_j\}_{j=1}^k$ as defined represents the sample space of the IHG probability measure, and $\{\zeta^*_j\}_{j=1}^k$ represents the sample space of the Dirichlet random variable. As such, the CDF of the IHG random variable follows,
    \begin{align*}
        F_{IHG}(\zeta_1, \dots,\zeta_{k-1}) & = F_{IHG}(\floor{n\zeta^*_1},...,\floor{(n-k+2)\zeta^*_{k-1}})\\
        &= \sum_{i_1=1}^{\floor{n\zeta^*_1}}\dots \sum_{i_{k-1}=1}^{\floor{(n-k+2)\zeta^*_{k-1}}} \frac{(k-1)!}{n(n-1)\dots(n-k+2)} &&(\text{Theorem \ref{thm:4}, $n$ large enough})\\
        &= (k-1)! \frac{\floor{n\zeta^*_1}}{n}\frac{\floor{(n-1)\zeta^*_2}}{n-1}\dots \frac{\floor{(n-k+2)\zeta^*_{k-1}}}{n-k+2}\\
        &= B^{-1}(\bm{1}_k) \frac{\floor{n\zeta^*_1}}{n}\frac{\floor{(n-1)\zeta^*_2}}{n-1}\dots \frac{\floor{(n-k+2)\zeta^*_{k-1}}}{n-k+2}\\
        &\to B^{-1}(\bm{1}_k) \zeta^*_1\dots \zeta^*_{k-1} &&(n\to\infty)\\
        &= B^{-1}(\bm{1}_k) \int_0^{\zeta^*_1}\dots\int_0^{\zeta^*_{k-1}} \partial\zeta^*_1\dots \partial\zeta^*_{k-1}\\
        &= F_{Dirichlet}(\zeta^*_1, \dots, \zeta^*_{k-1}; \bm{1}_k)
    \end{align*}
\end{proof}

\section{Appendix C: Supplementary Results for Simulation Study}
\subsection{Simulation study: factorial subsets}
Following up from Section \ref{sec:sim}, we breakdown the factorial study into subsets along the time distribution (Figure \ref{fig:ms_ss_timedist}), the error distribution (Figure \ref{fig:ms_ss_errvar}), and the robustness distribution (Figure \ref{fig:ms_ss_robust}). Finally the performance of the three main models are broken down by number of segments in the synthetic data in Figure \ref{fig:ms_ss_nseg}.
\begin{figure}
    \centering
    \includegraphics[width=1.\linewidth]{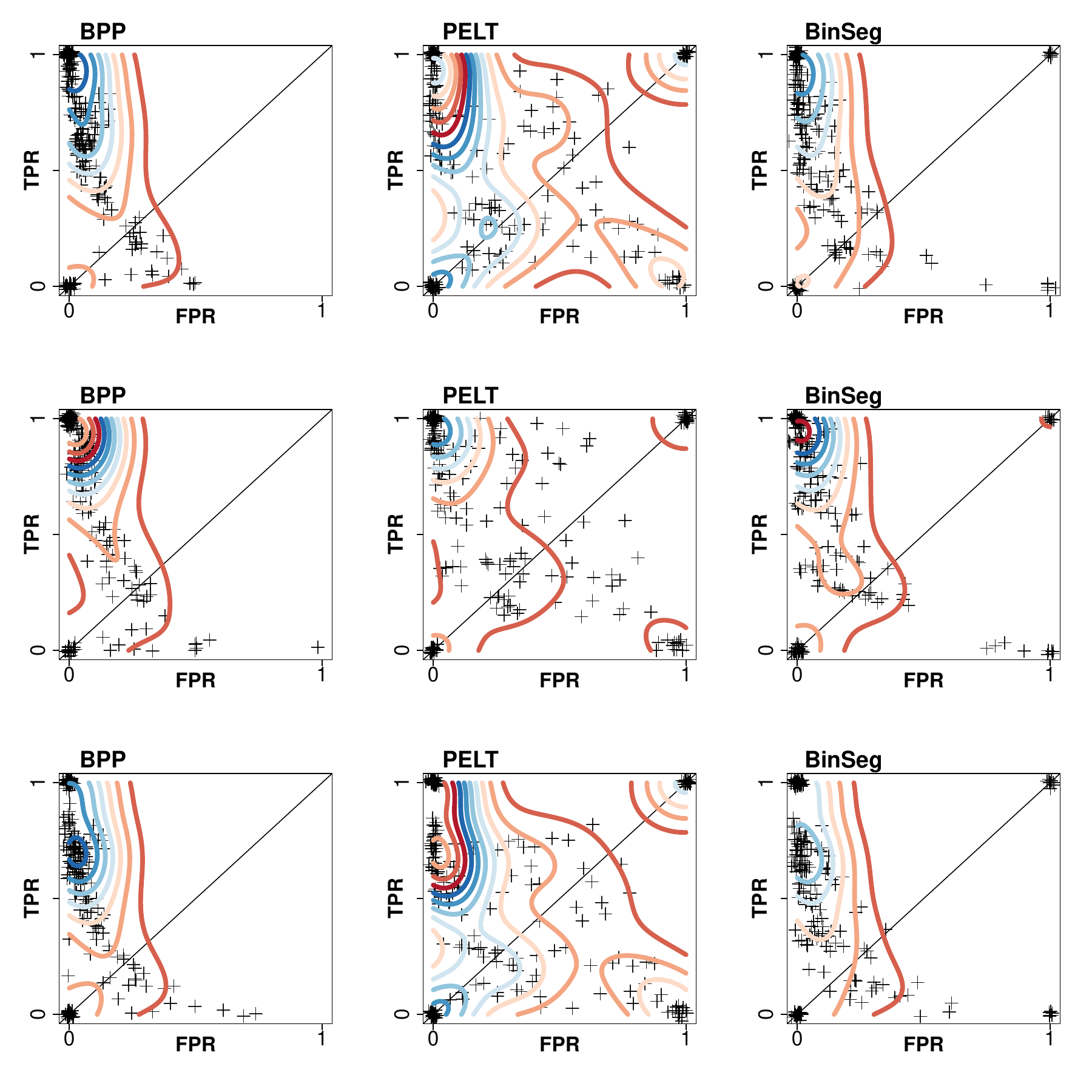}
    \caption{Row 1 is a subset of the data generated with time distribution uniformly spaced and change points uniformly distributed. Row 2 is a subset of the data generated with time distribution $t_i \overset{\text{i.i.d.}}{\sim} \textbf{Beta}(0.5,0.5)$ and change points simulated from $\textbf{BPP}$. Row 3 is a subset of the data generated with time distribution $t_i \overset{\text{i.i.d.}}{\sim} \textbf{Beta}(2,2)$ and change points simulated from $\textbf{BPP}$.}
    \label{fig:ms_ss_timedist}
\end{figure}

\begin{figure}
    \centering
    \includegraphics[width=1.\linewidth]{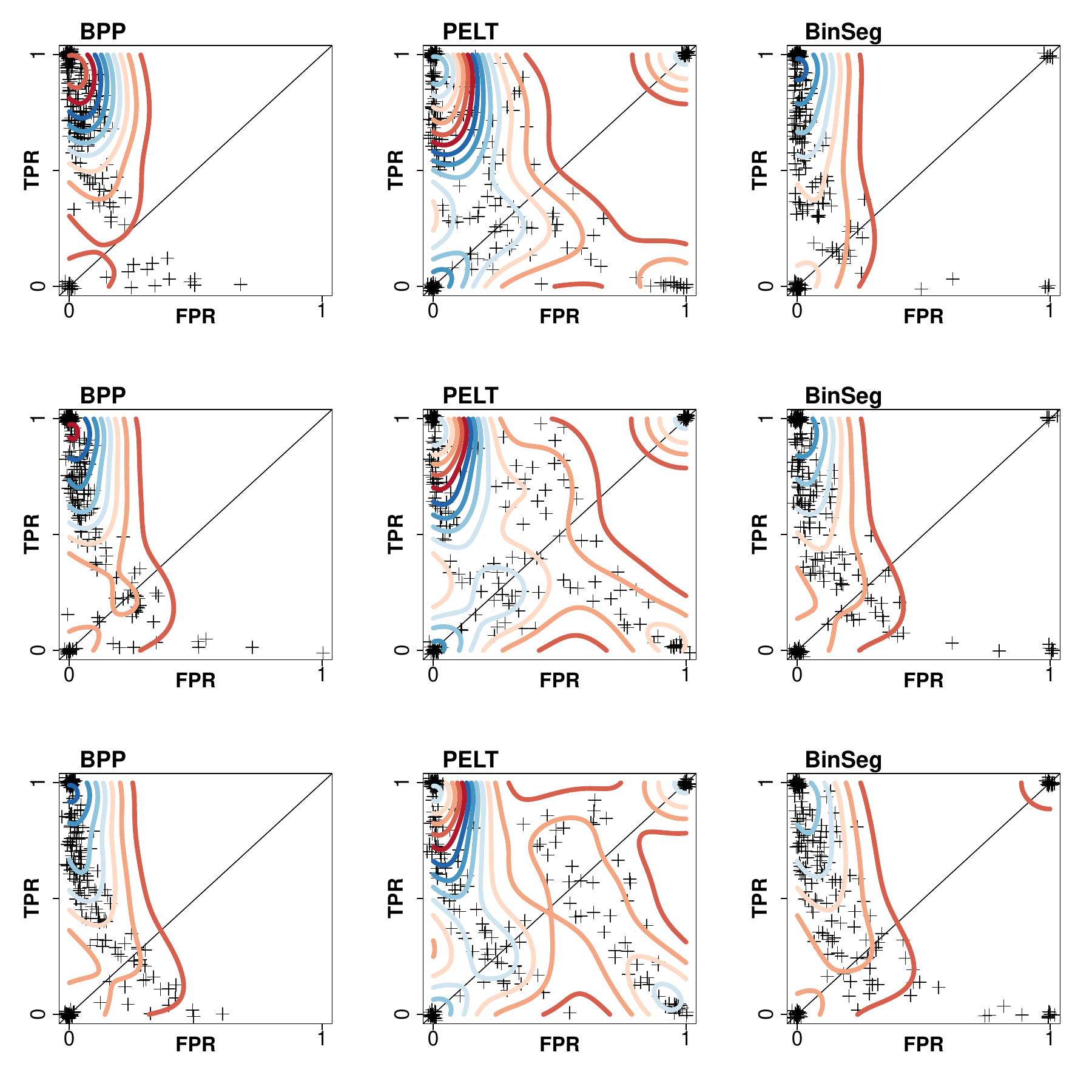}
    \caption{Row 1 is a subset of the data generated with error variance $0.1$. Row 2 is a subset of the data generated with error variance $0.2$. Row 3 is a subset of the data generated with error variance $0.3$.}
    \label{fig:ms_ss_errvar}
\end{figure}

\begin{figure}
    \centering
    \includegraphics[width=1.\linewidth]{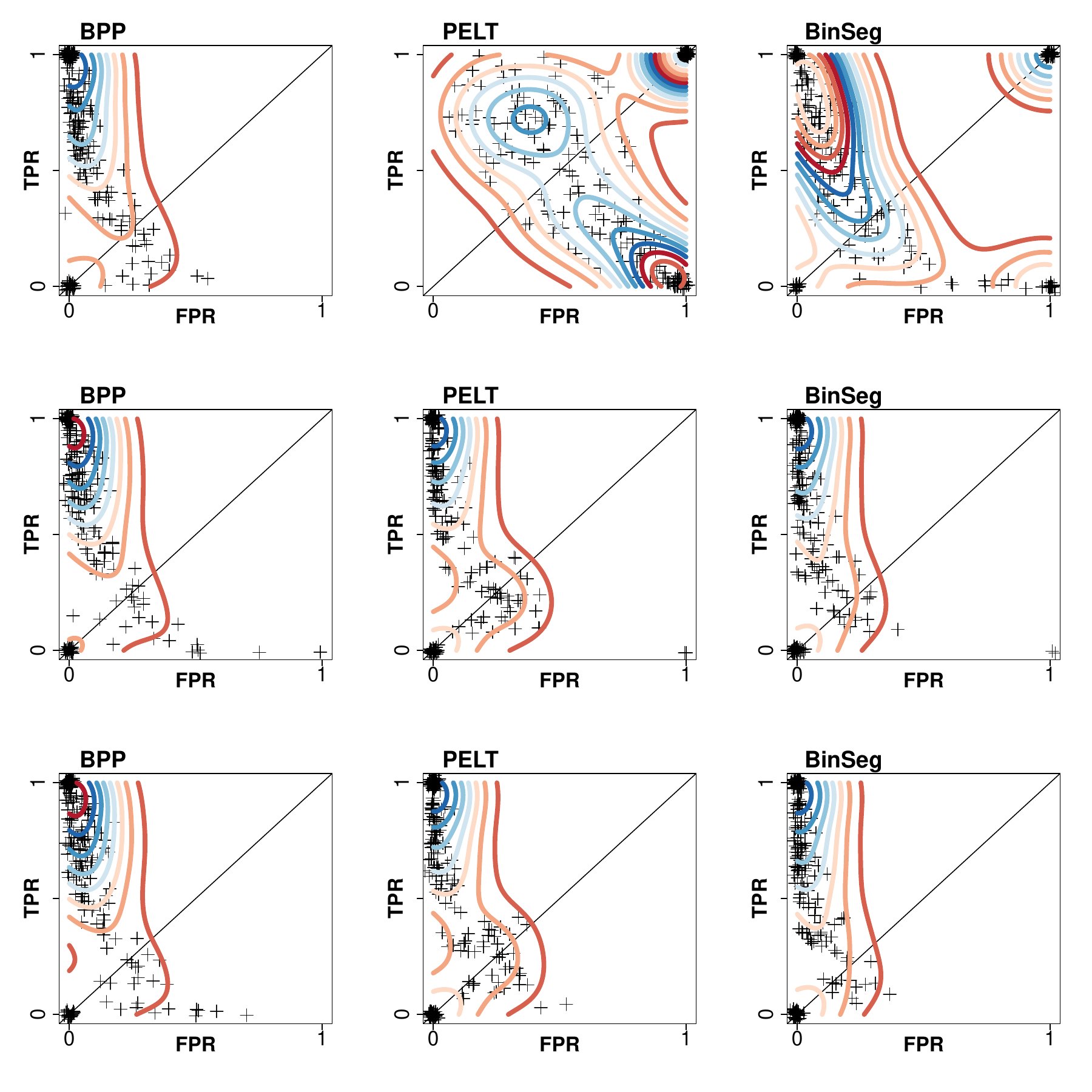}
    \caption{Row 1 is a subset of the data generated with robustness parameter $\nu = 3$. Row 2 is a subset of the data generated with robustness parameter $\nu = 10$. Row 3 is a subset of the data generated with robustness parameter $\nu = 100$.}
    \label{fig:ms_ss_robust}
\end{figure}

\begin{figure}
    \centering
    \includegraphics[width=1.\linewidth]{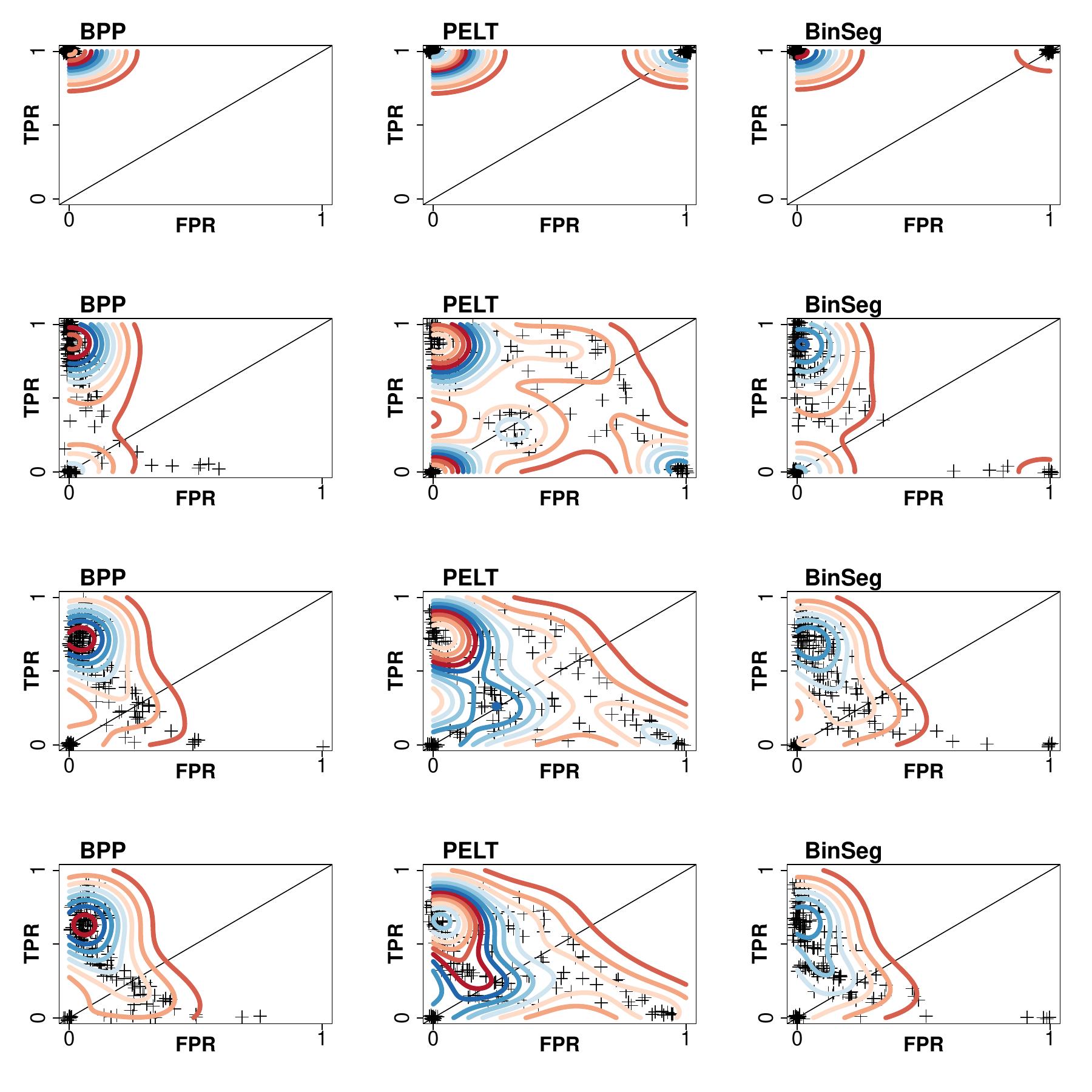}
    \caption{Study is broken down by number of changes. First row is 0 changes, to the fourth row of 3 changes.}
    \label{fig:ms_ss_nseg}
\end{figure}

\subsection{Simulation study with Gibbs sampler}
Figure \ref{fig:mainstudyGibbs} implements the Gibbs sampler from subsection \ref{sub:gibbs} on the synthetic data study with 10 replicates per setting as described in Section \ref{sec:sim}.
\begin{figure}
    \centering
    \includegraphics[width=0.5\linewidth]{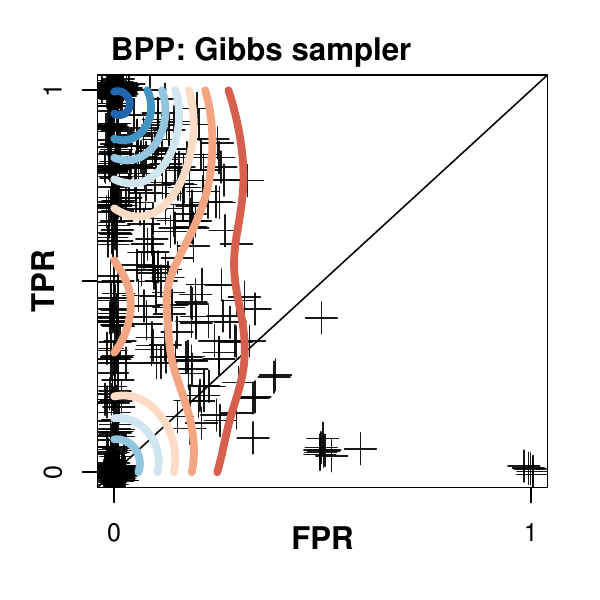}
    \caption{Gibbs sampler from subsection \ref{sub:gibbs} is run on 10 replications of the same synthetic data from Section \ref{sec:sim}.}
    \label{fig:mainstudyGibbs}
\end{figure}

\subsection{Simulation study with other models}
Following up from Section \ref{sec:sim}, we run the full study on three additional models in Figure \ref{fig:os_study}.  The first model is $\textbf{BPP}$ change point process model with Normal likelihood for the error distribution instead of t-distributed likelihood, the second it the noninformative discrete time model from Proposition \ref{prop:trans}, and the third model is $\textbf{BPP}$ but with a different prior on the number of segments following Equation \ref{eq:impliedpk}. We then breakdown the factorial study into subsets along the time distribution (Figure \ref{fig:os_ss_timedist}), the error distribution (Figure \ref{fig:os_ss_errvar}), and the robustness distribution (Figure \ref{fig:os_ss_robust}) for these three models. Finally the performance of the three main models are broken down by number of segments in the synthetic data in Figure \ref{fig:os_ss_nseg}.

\begin{figure}
    \centering
    \includegraphics[width=1.\linewidth]{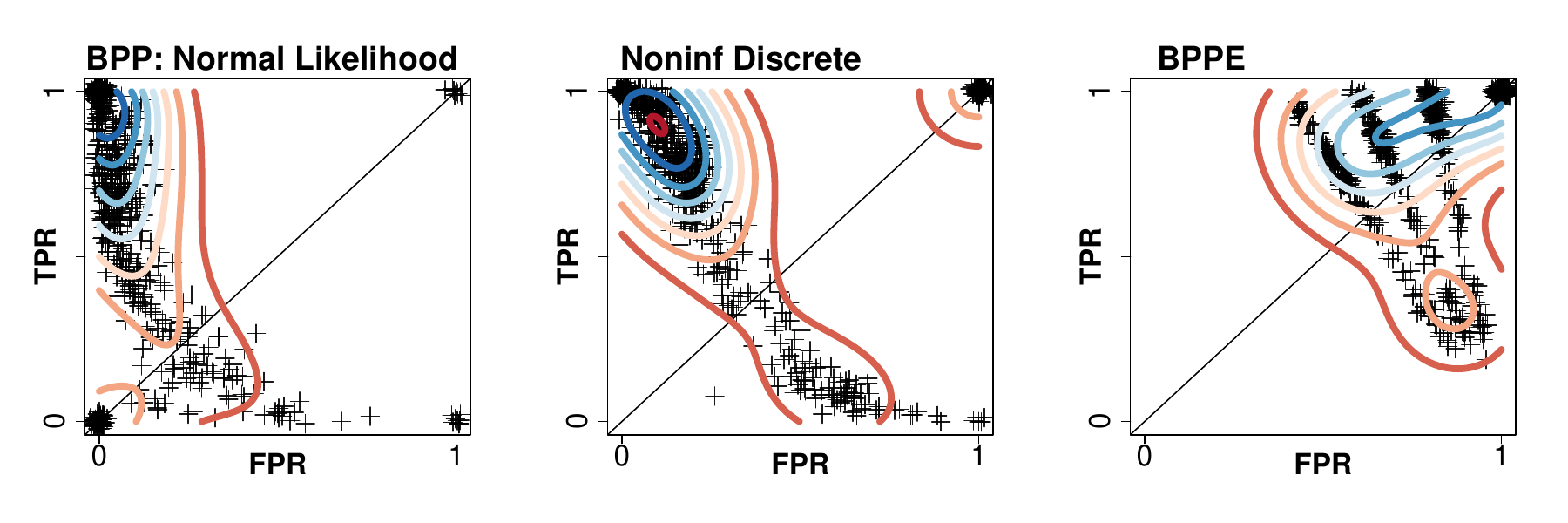}
    \caption{Comparing three additional models on the full factorial synthetic study. Left is the continuous time noninformative BPP model but with a normal error distribution. Middle is the noninformative discrete time model from Proposition \ref{prop:trans}. Right is the continuous time noninformative BPP model but with prior on number of segments that represents equally likely sequences across $k$.}
    \label{fig:os_study}
\end{figure}

\begin{figure}
    \centering
    \includegraphics[width=1.\linewidth]{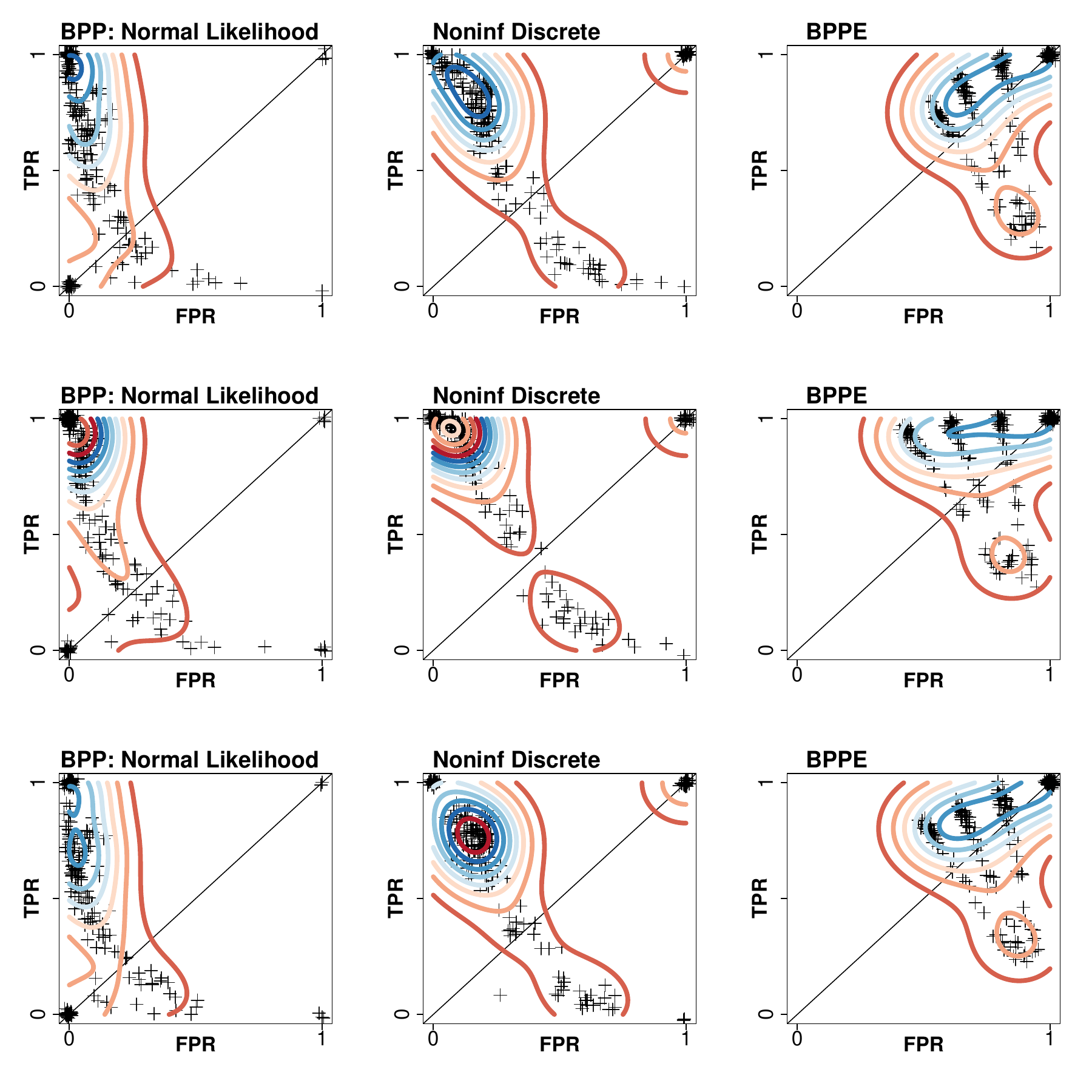}
    \caption{Row 1 is a subset of the data generated with time distribution uniformly spaced and change points uniformly distributed. Row 2 is a subset of the data generated with time distribution $t_i \overset{\text{i.i.d.}}{\sim} \textbf{Beta}(0.5,0.5)$ and change points simulated from $\textbf{BPP}$. Row 3 is a subset of the data generated with time distribution $t_i \overset{\text{i.i.d.}}{\sim} \textbf{Beta}(2,2)$ and change points simulated from $\textbf{BPP}$.}
    \label{fig:os_ss_timedist}
\end{figure}

\begin{figure}
    \centering
    \includegraphics[width=1.\linewidth]{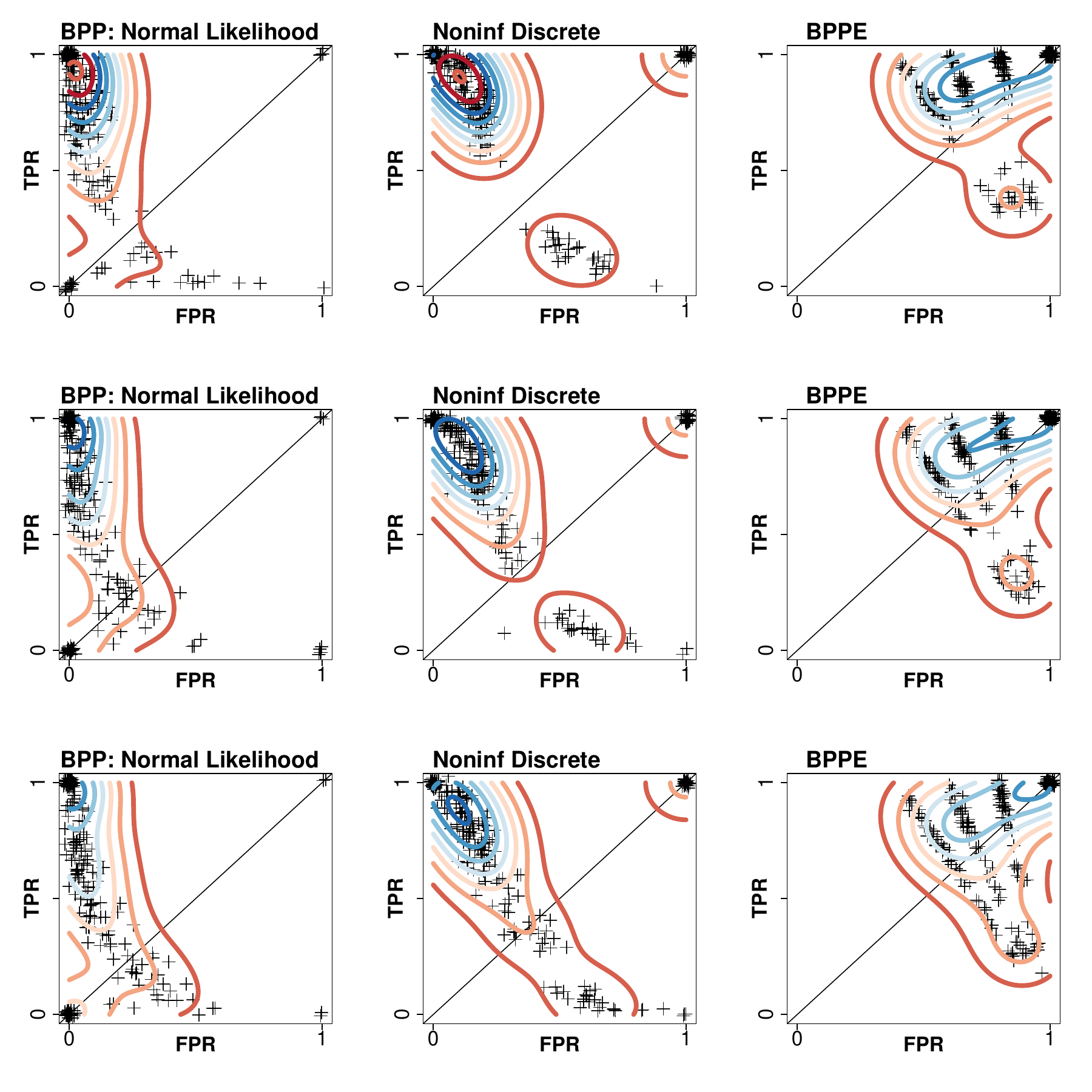}
    \caption{Row 1 is a subset of the data generated with error variance $0.1$. Row 2 is a subset of the data generated with error variance $0.2$. Row 3 is a subset of the data generated with error variance $0.3$.}
    \label{fig:os_ss_errvar}
\end{figure}

\begin{figure}
    \centering
    \includegraphics[width=1.\linewidth]{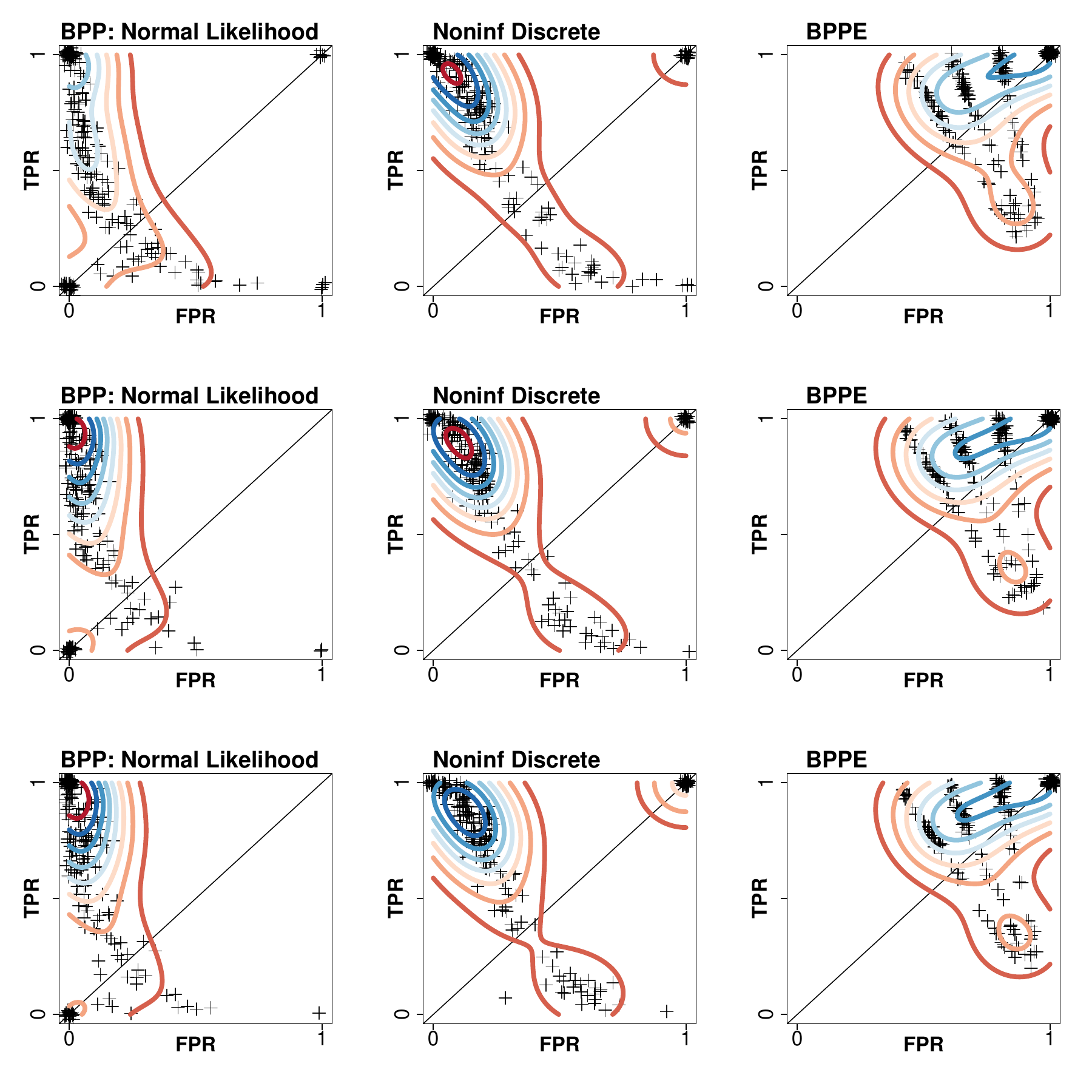}
    \caption{Row 1 is a subset of the data generated with robustness parameter $\nu = 3$. Row 2 is a subset of the data generated with robustness parameter $\nu = 10$. Row 3 is a subset of the data generated with robustness parameter $\nu = 100$.}
    \label{fig:os_ss_robust}
\end{figure}

\begin{figure}
    \centering
    \includegraphics[width=1.\linewidth]{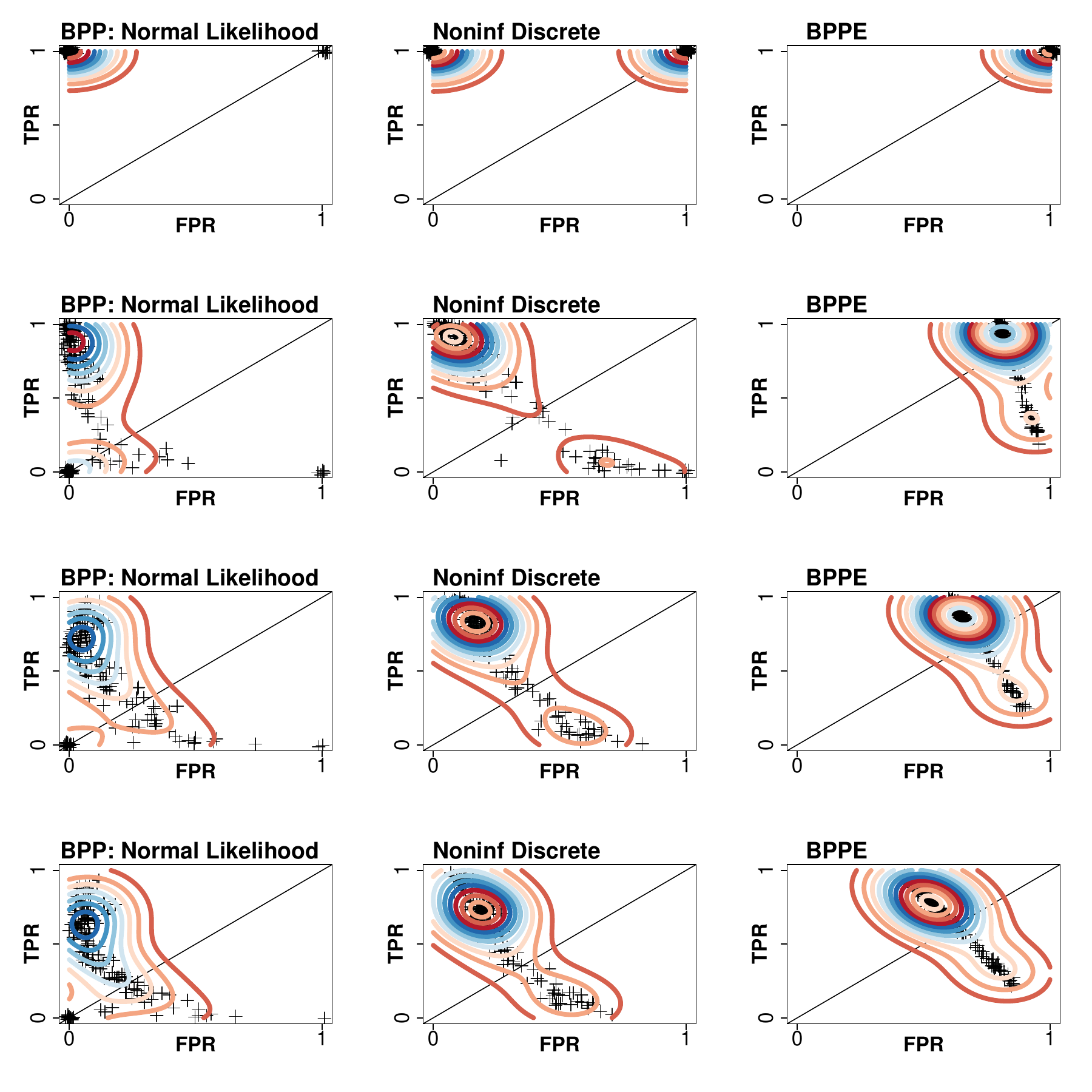}
    \caption{Study is broken down by number of changes. First row is 0 changes, to the fourth row of 3 changes.}
    \label{fig:os_ss_nseg}
\end{figure}

\section{Appendix D: Supplementary Results for Case Study}
\subsection{Prior on parameters}
The mean parameter $\bm{\theta} = (\alpha,\beta,\{\gamma_h,\delta_h\}_{h=1}^H)^T$ are a priori independent and 0 precision except for $\beta$. Since we do not want short periods of change to be captured by sharp slopes, we set the precision of $\beta$ to be 5 to help regularize and avoid spurious changes. Denote corresponding precision matrix as $\Lambda_{\theta}$.

Now, consider the prior distribution on the annual harmonic contrasts $\bm{\phi} = (\{\gamma_{h,l}\}_{h=1,l=1}^{H,J},\{\delta_{h,l}\}_{h=1,l=1}^{H,J})^T$ given their continuity constraints. We will construct this prior separately for $\gamma_{hl}$ and $\delta_{hl}$ for each year, and then put it together afterwards.  

Define $\bm{\gamma}_l = (\gamma_{1,l},\dots,\gamma_{H,l})^T$ as the vector of $\sin$ coefficients for the $l$th year. Assume these contrasts are Gaussian with mean zero, having an exponentially decaying diagonal variance of the seasonal anomalies with respect to the harmonic number.  Given the prior for $\bm{\gamma}_{l}$, we will derive the prior distribution for $\bm{\gamma}_{l,-H}$ conditioned on the continuity constraints on the $H$th harmonic. 

Let $\bm{\gamma}_l \sim N(\bm{0}, \Phi_C)$ where $\Phi_C = \psi \text{Diag}_{h = 1, \ldots, H} \{\exp{\lambda(1-h)}\}$. The $\psi$ parameter is the prior variance of the first harmonic, which then exponentially decays according to $\lambda$ as the harmonics increase. In all that follows, we assume $\lambda=1$. The joint distribution of $\bm{\gamma}_l$ and the continuity constraint $\xi_l = \sum_{h = 1}^H \gamma_{hl}$ is, with $s = \sum_{h'}\sum_h \Phi_{Chh'}$, 
\begin{align*}
\begin{bmatrix}\bm{\gamma} \\ \xi\end{bmatrix} = \begin{bmatrix} I_H \\ \bm{1}_H^\top\end{bmatrix}\bm{\gamma} &\sim N\left(\bm{0},\begin{bmatrix} \Phi_C & \Phi_C \bm{1}_H \\ (\Phi_C \bm{1}_H)^\top & s \end{bmatrix}\right)
\end{align*}

Using formulae for Gaussian conditional distributions, we arrive at, $\bm{\gamma}_l|\xi_l=0 \sim N(\bm{0}, \Phi_C - \Phi_C\bm{1}_H(\Phi_C\bm{1}_H)^\top/s)$. Only the first $h-1$ positions of this conditional multivariate Gaussian are used since the $h$th harmonic is constrained. The contrast covariance matrix is then the kronecker product over $2J$ copies of this covariance matrix for $2$ harmonics and $J$ years.
\[\Lambda_{\phi}^{-1} = I_{2J} \otimes \bigg(\Phi_C - \Phi_C\bm{1}_H(\Phi_C\bm{1}_H)^\top/s)\bigg)\]
The full parameter precision matrix is then the block diagonal operation of the precision matrix on $\bm{\theta}$ and $\bm{\phi}$ as,
\[\Phi^{-1} = \text{blkdiag}(\Lambda_{\theta},\Lambda_{\phi} )\]

\subsection{Applying other models to the case study}
\subsubsection{Case study results for model without interannually varying harmonics}
We also evaluate the three case study locations for the harmonic model without interannually varying harmonics from Equation \ref{eq:GausLik}. These results are in Figure \ref{fig:rondonia_6h}, Figure \ref{fig:cropRotation_6h} and Figure \ref{fig:permianBasin_6h}. The mean phenology function estimates are clearly different from our model in the original case study since interannual variation is not being captured.  The detected changes for deforestation and crop rotation are similar, however the model fails to capture the changes due to drought in the shrub and grassland example. 
\begin{figure}
\includegraphics[width=1.\linewidth]{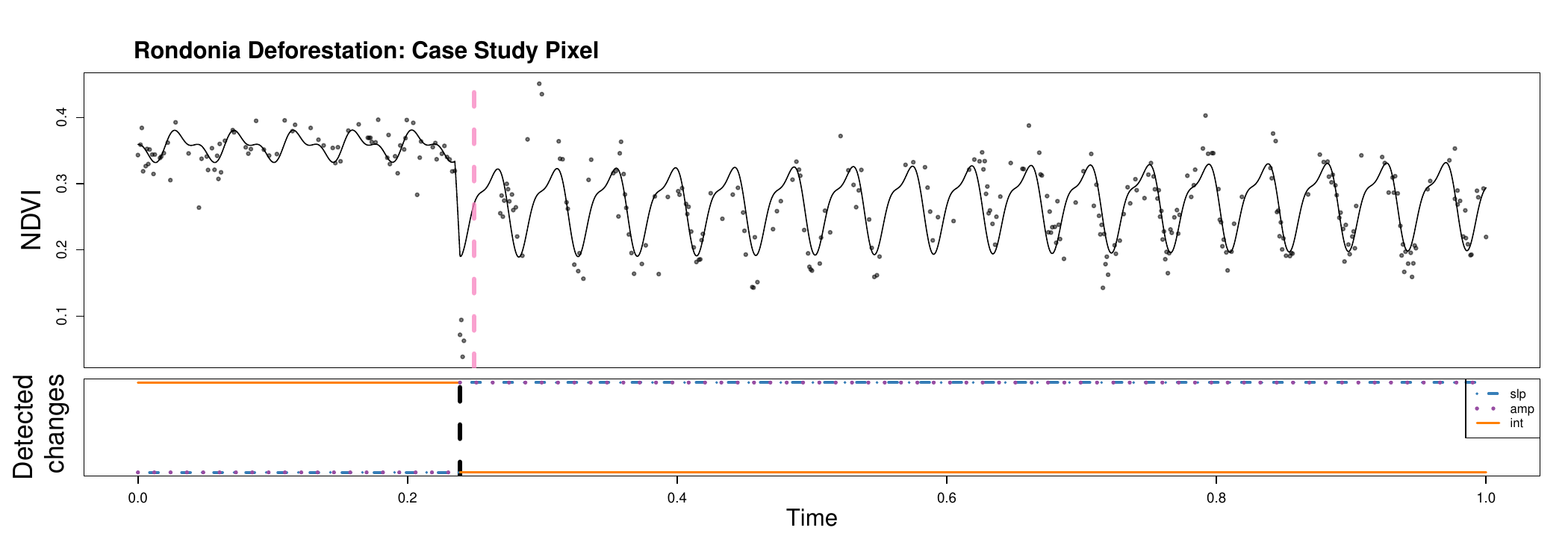}  
\caption{Evaluating the same case study location for deforestation in Rondonia, but without interannually varying harmonics as in Equation \ref{eq:GausLik}.}
\label{fig:rondonia_6h}
\end{figure}

\begin{figure}
\includegraphics[width=1.\linewidth]{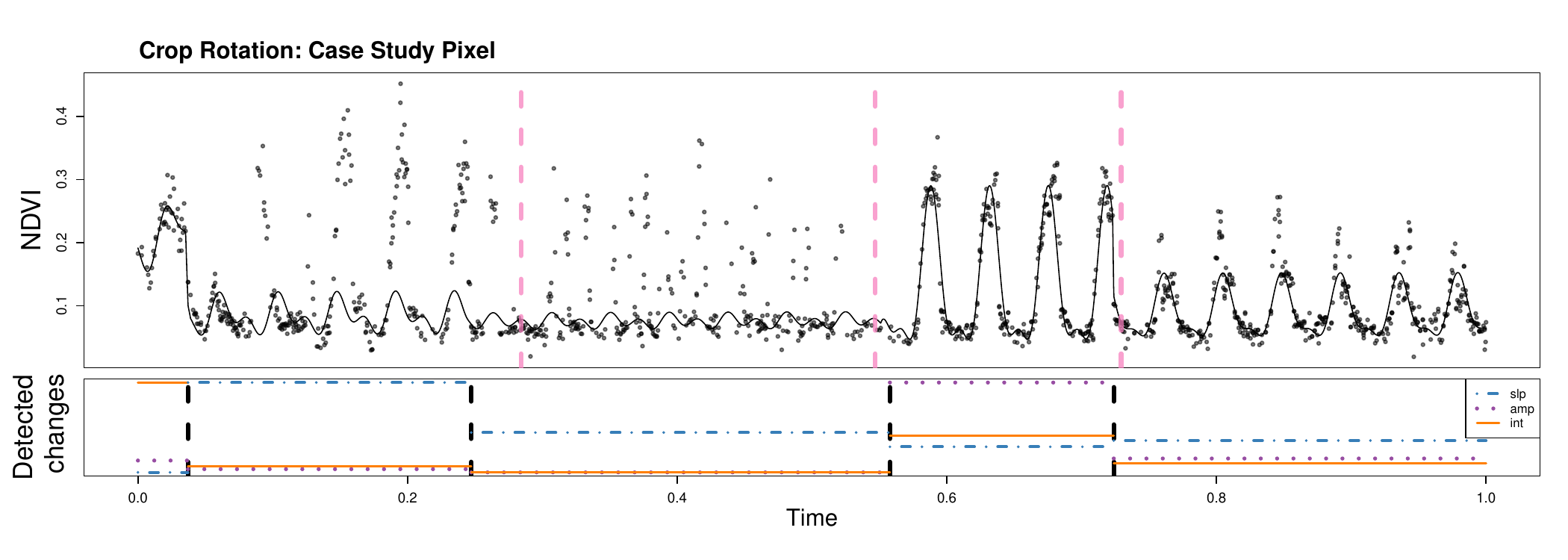} 
\caption{Evaluating the same case study location for crop rotation, but without interannually varying harmonics as in Equation \ref{eq:GausLik}. This model detects similar changes despite that it does not capture interannual variation.}
\label{fig:cropRotation_6h}
\end{figure}

\begin{figure}
\includegraphics[width=1.\linewidth]{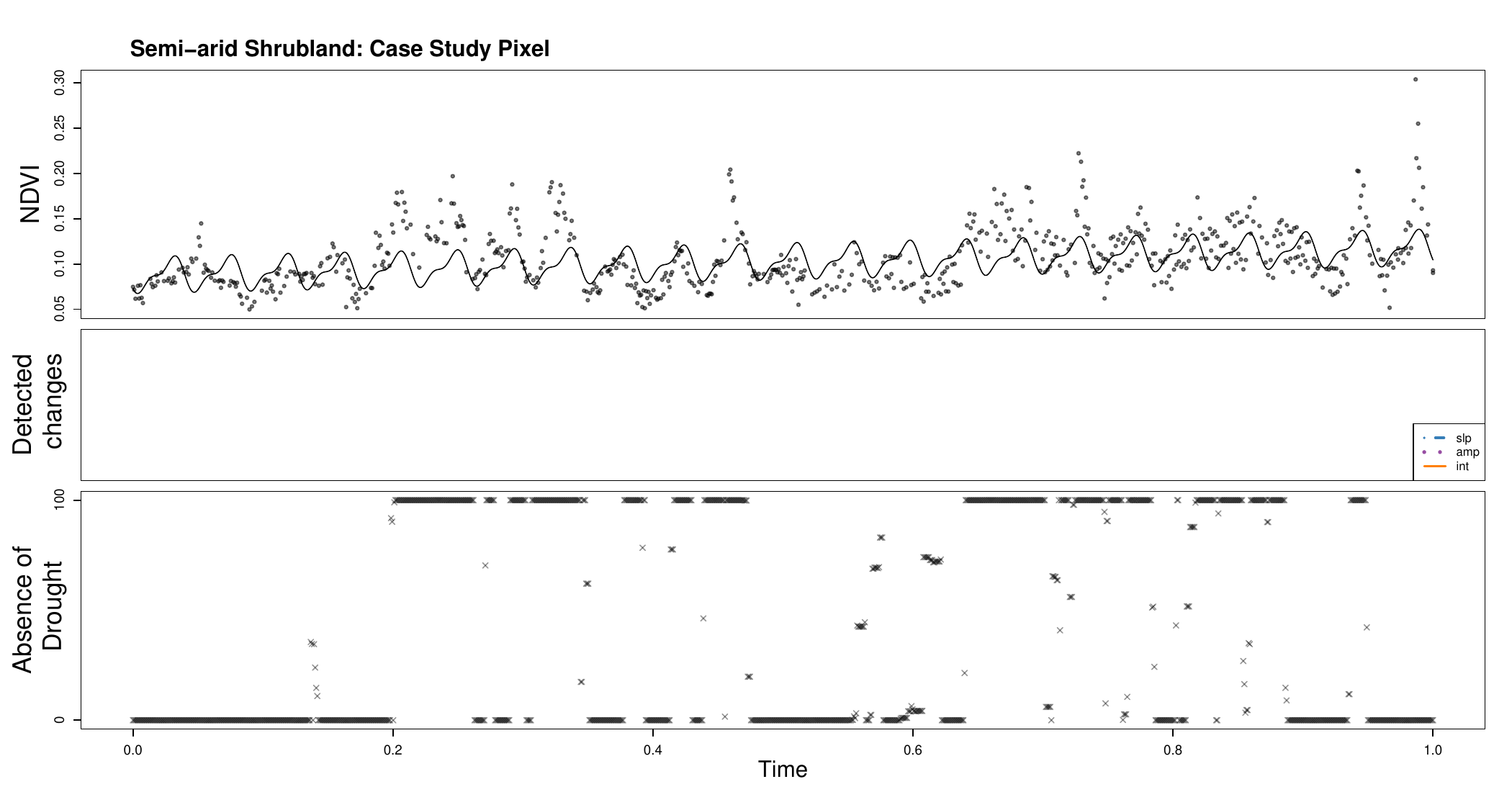}
\caption{Evaluating the same case study location for drought responses in shrub and grassland, but without interannually varying harmonics as in Equation \ref{eq:GausLik}. This model fails to detect changes due to drought as a result of removing interannual variation.}
\label{fig:permianBasin_6h}
\end{figure}

\subsubsection{Case study results for different prior on number of segments}
In subsection \ref{sub:pk}, we introduced two priors on the number of segments. The prior we use in the case study in Section \ref{sec:casestudy} is from Equation \ref{eq:pk}. In this subsection, we evaluate the case study pixels under the inverse of that prior,
\begin{equation}
\pi_0(k)\propto (2\pi)^{\frac{pk}{2}}|\Phi^{-1}|^{\frac{-k}{2}} \prod_{i=1}^n \big((1-t_{i})/(1-t_{i-1})\big)^{-k}     
\end{equation}
The results are in Figure \ref{fig:rondonia_impliedpk}, Figure \ref{fig:cropRotation_impliedpk}, and Figure \ref{fig:permianBasin_impliedpk}. The deforestation example demonstrates that the model under this prior incurs extra falsely detected changes compared to the prior in Equation \ref{eq:pk}.

\begin{figure}
\includegraphics[width=1.\linewidth]{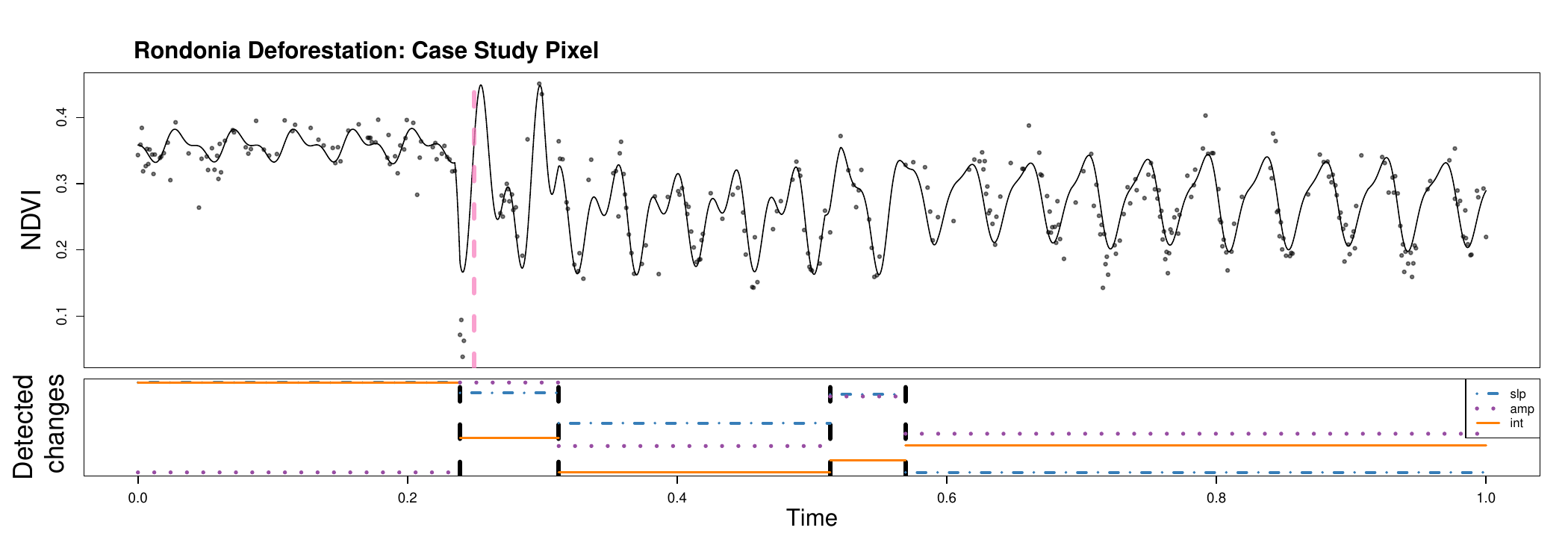}  
\caption{Evaluating the same case study location for deforestation in Rondonia, but with a different prior on the number of segments. Notice three more changes are added. These extra changes appear to be false positives as supported by high resolution imagery and reference to MapBiomas}
\label{fig:rondonia_impliedpk}
\end{figure}

\begin{figure}
\includegraphics[width=1.\linewidth]{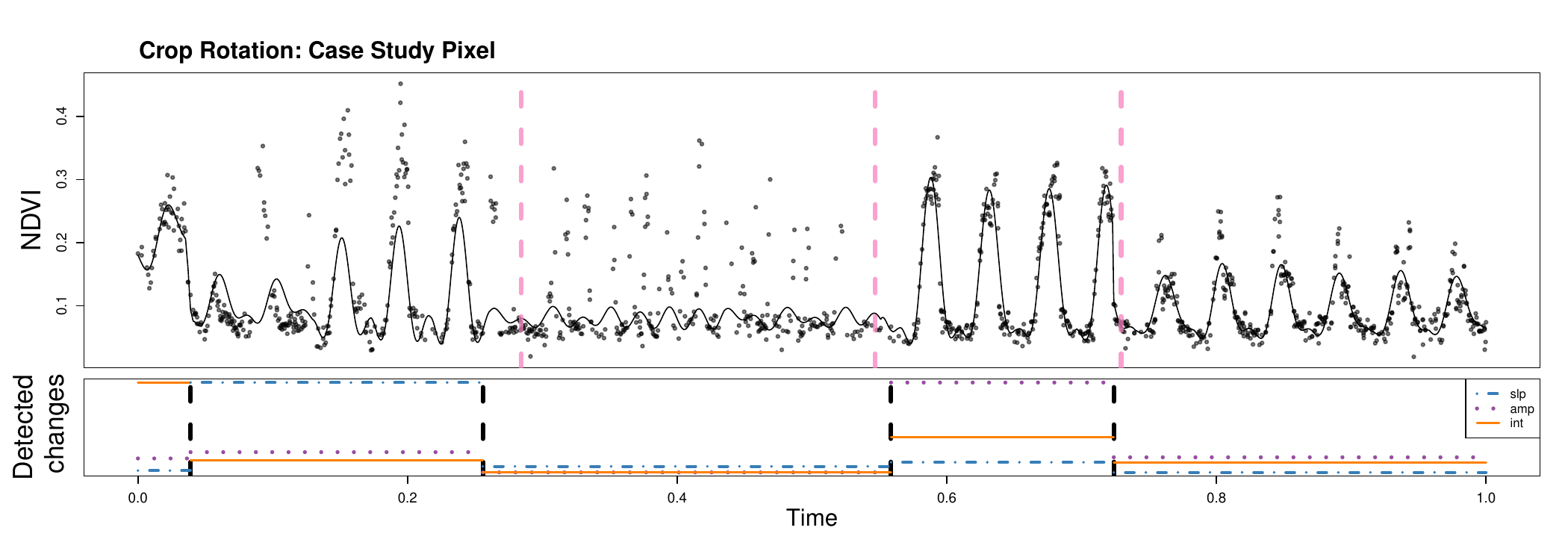} 
\caption{Evaluating the same case study location for crop rotation, but with the inverse prior on the number of segments. This model detects similar changes despite the different prior.}
\label{fig:cropRotation_impliedpk}
\end{figure}

\begin{figure}
\includegraphics[width=1.\linewidth]{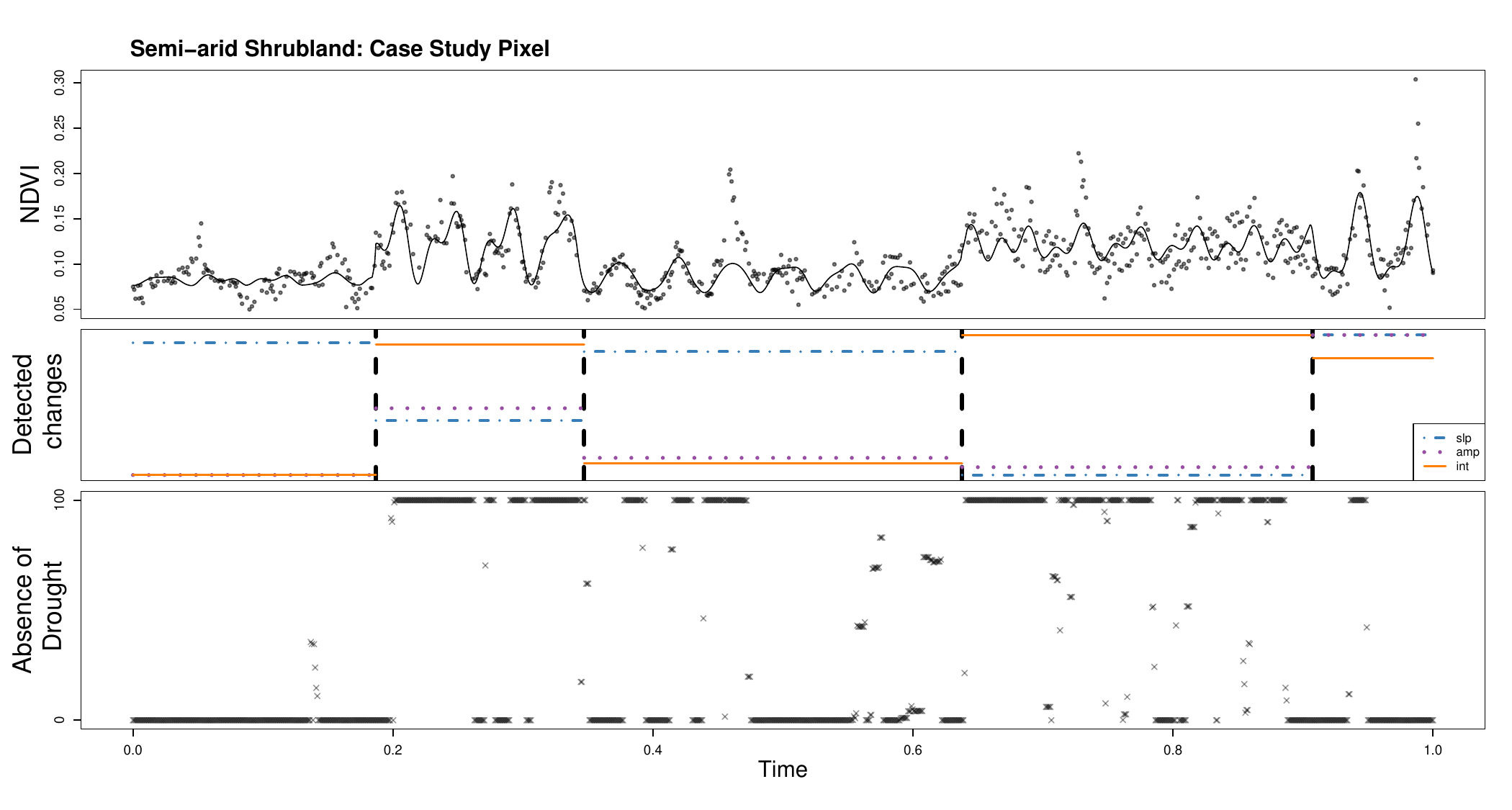}
\caption{Evaluating the same case study location for drought responses in shrub and grassland, but with the inverse prior on the number of segements. Change detected results do not change after switching the prior.}
\label{fig:permianBasin_impliedpk}
\end{figure}

\section{Appendix E: Supplementary Results for Methodology: EM and Simulation}
\subsection{Expectation Maximization}
Expectation maximization will be used to obtain posterior expectations of the robustness variables $\{q_i\}_{i=0}^n$ as well as the state variables $\{z_{t_i}\}_{i=0}^n$, and to maximize the marginal likelihood with respect to the mean and variance parameters for each segment $(\Theta,\sigma^2)$. As such, we evaluate the posterior expectations, $\mathbb{E}_{z_{t_i}|\bm{y},X,\Theta^{(s)}}[1\{z_{t_{i}}=j\}]$, $\mathbb{E}_{z_{t_i},z_{t_{i-1}}|\bm{y},X,\Theta^{(s)}}[1\{z_{t_{i}}=j,z_{t_{i-1}}=j\}]$.

Following \cite{little2019statistical}, the posterior distribution of $q_i|z_i=j,y_i,X,\Theta_j^{(s)} \sim Ga(\frac{\nu+1}{2}, (\frac{\nu}{2} + \frac{(y_i - x_i^T\theta_j^{(s)})^2}{2\sigma_j^{2(s)}}))$ from which the corresponding E-steps are readily available. Conditioning on $z_i = j$ and the likelihood mean function for the $j$th state at time $t_i$, $\mu_{j,t_i}^{(s)}$, and assessing the posterior for a single $q_i$,
\begin{align*}
    p(q_{t_i}|y_{t_i}, z_{t_i} = j;\nu) &\propto \frac{1}{q_{t_i}}^{-1/2} \exp{-\frac{q_{t_i}(y_{t_i} - \mu_{j,t_i}^{
    (s)})^2}{2\sigma^{2(t)}}} q_i^{\frac{\nu}{2}-1}\exp{-\frac{q_i\nu}{2}}\\
    &\propto  q_i^{\frac{\nu+1}{2}-1} \exp{-q_i(\frac{\nu}{2} + \frac{(y_i - \mu_{j,t_i}^{(s)})^2}{2\sigma^{2(s)}})}\\
\end{align*}
Which is a gamma distribution $Ga(\frac{\nu+1}{2}, (\frac{\nu}{2} + \frac{(y_{t_i} - \mu_{j,t_i}^{(s)})^2}{2\sigma^{2(s)}}))$.  The Q function follows,
\begin{align*}
    Q(\bm{\Theta}|\bm{\Theta}^{(s)}) \overset{(c)}{=} \mathbb{E}_{\bm{q},\bm{z}|\bm{y},X,\Theta^{(s)}}\bigg[\sum_{i=0}^n\sum_{j=1}^k 1\{z_{t_{i}}=j\}\bigg( -\log(\sigma) - \frac{q_{t_i}}{2\sigma^2}(y_{t_i} - x_{t_i}^T\bm{\theta}_j)^2\bigg) +\\ -pk\log(\sigma) -(\frac{1}{2\sigma^2}\sum_{j=1}^k \bm{\theta}_j^T\Phi^{-1}\bm{\theta}_j) - \log(\sigma^2)\\
    +\sum_{i=1}^n\sum_{j=1}^{k-1}\sum_{h=j}^{k} 1\{z_{t_i}=h,z_{t_{i-1}}=j\}\log\bigg(\pi(z_{t_{i}}=h|z_{t_{i-1}}=j)\bigg)\bigg]
\end{align*}
The M-steps for the mean parameters are weighted least squares $\hat{\bm{\theta}}_j = (X^TW_jX + \Phi^{-1})^{-1}X^TW_j\bm{y}$ where $W_j$ is a diagonal matrix with entries $\mathbb{E}[1\{z_i=j\}q_i|y,\Theta^{(s)}] = \mathbb{E}[q_i|1\{z_i=j\},y,\Theta^{(s)}]*\mathbb{E}[1\{z_i=j\}|y,\Theta^{(s)}]$. The first of those expectations is given above, and the marginal expectation of $z_i=j$ is provided by the forward-backward algorithm. The joint posterior expectations of $1\{z_i=j+1,z_{i-1}=j\}$ is also provided by the forward-backward algorithm. The M-step for the variance $\sigma^2$ can also be evaluated analytically, 
\[ \sigma^{2(s+1)} = \frac{\sum_{i=0}^n \sum_{j=1}^k \mathbb{E} \bigg[1\{z_{t_i}=j\}q_{t_i} \big| \bm{y},\Theta^{(s)},\sigma^{2(s)}\bigg] \bigg(y_{t_i} -x_{t_i}^T\bm{\theta}^{(s+1)}_j\bigg)^2 + \sum_{j=1}^k\bm{\theta}^{(s+1)}_j\Phi^{-1}\bm{\theta}^{(s+1)}_j}{\bigg(\sum_{i=0}^n \sum_{j=1}^k\mathbb{E}\big[1\{z_{t_i}=j\} \big| \bm{y},\Theta^{(s)},\sigma^{2(s)}\big]\bigg)+pk+2}
\] 
After the M-step is complete, the E-step is then repeated conditioned on the updated parameters.  The algorithm is repeated until convergence of the $Q$ function. 

The likelihood distribution of $y_{t_i}|z_{t_i}=j,\Theta$ after marginalizing out $q_{t_i}$ is t-distributed as follows,
\begin{align*}
    f_{y_{t_i}}(y_{t_i};\mu_{t_i},\sigma^2)
    &=(2\pi\sigma^2)^{-1/2} \frac{\frac{\nu}{2}^{\frac{\nu}{2}}}{\Gamma(\frac{\nu}{2})} \int q_i^{\frac{\nu+1}{2}-1} \exp{-q_i(\frac{\nu}{2} + \frac{(y_i - \mu_{j,t_i})^2}{2\sigma^{2}})} dq_i\\
    &=(2\pi\sigma^2)^{-1/2} \frac{\frac{\nu}{2}^{\frac{\nu}{2}}\Gamma(\frac{\nu+1}{2})}{\Gamma(\frac{\nu}{2})}(\frac{\nu}{2} + \frac{(y_i - \mu_{j,t_i})^2}{2\sigma^{2}})^{-\frac{\nu+1}{2}}\\
    &=\frac{1}{\sigma} \nu^{\frac{\nu}{2}}\frac{\frac{1}{2}^{\frac{\nu+1}{2}}\Gamma(\frac{\nu+1}{2})}{(\pi)^{1/2}\Gamma(\frac{\nu}{2})}(\frac{\nu}{2} + \frac{(y_i - \mu_{j,t_i})^2}{2\sigma^{2}})^{-\frac{\nu+1}{2}}\\
    &=\frac{1}{\sigma} \frac{\Gamma(\frac{\nu+1}{2})}{(\pi\nu)^{1/2}\Gamma(\frac{\nu}{2})}(1 + \frac{(y_i - \mu_{j,t_i})^2}{\sigma^{2}\nu})^{-\frac{\nu+1}{2}}\\
\end{align*}
Which is a location-scaled t-distribution with mean $\mu_{j,t_i}$ and scale $\sigma$.
\subsection{Gibbs Sampling}
\label{sub:gibbs}
Toward full Bayesian inference, analytical posteriors are not available for general models (see \cite{fearnhead2006exact} for a model obtaining exact posterior inference), however simulation for the full conditional distribution $[\bm{z}|\bm{y},\Theta,\sigma^2]$ can be derived and used within a broader Gibbs sampling methodology.

The posterior conditional distribution of the mean vectors $\bm{\theta}_j$ follow a Gaussian distribution since their prior is Gaussian. Let $W^{(j)}_{ii} = q_i 1\{z_i=j\}$ be diagonal,
\begin{align*}
    p(\bm{\theta}_j|\bm{y}, \bm{z},\bm{q},\sigma_j^2) 
    \propto \exp\{-\frac{1}{2}(\bm{\theta}_j - \bm{\mu}_j)^T\Lambda_j (\bm{\theta}_j - \bm{\mu}_j)\}    
\end{align*}
Where $\bm{\mu} = (X^TW^{(j)}X+\Phi^{-1})^{-1}X^T W^{(j)} \bm{y}$ and $\Lambda_j = (X^TW^{(j)}X+\Phi^{-1})/\sigma^2$ are the mean and precision matrix of the Gaussian posterior for $\bm{\theta}_j$.
The posterior conditional distribution of $\sigma^2$ is scaled-inverse-$\chi^2$ as follows,
\begin{align*}
    p(\sigma^2 | \bm{y}, \bm{z},\bm{q},\Theta)     &\propto (\sigma^2)^{-\frac{(\sum_{i=1}^n\sum_{j=1}^k 1\{z_{t_i}=j\}) + pk}{2}-1}\\ 
    &\quad\quad\exp\bigg(-\frac{\sum_{i=1}^n\sum_{j=1}^k q_{t_i} 1\{z_{t_i}=j\} (y_{t_i} - \bm{x}_{t_i}^T\bm{\theta}_j)^2  + \sum_{j=1}^k\bm{\theta}_j^T\Phi^{-1}\bm{\theta}_j}{2\sigma^2}\bigg)
\end{align*}
Which is a scaled-inverse-$\chi^2 (\nu_0, \tau^2_0)$ with parameters $\nu_0 = \sum_{i=1}^n \sum_{j=1}^k 1\{z_{t_i}=j\}+pk$ and $\tau^2_0 = \frac{\sum_{i=1}^n\sum_{j=1}^k q_{t_i} 1\{z_{t_i}=j\} (y_{t_i} - \bm{x}_{t_i}^T\bm{\theta}_j)^2  + \sum_{j=1}^k\bm{\theta}_j^T\Phi^{-1}\bm{\theta}_j}{\sum_{i=1}^n\sum_{j=1}^k 1\{z_{t_i}=j\} + pk}$, where $p$ is the dimension of $\bm{\theta}_j$ for all $j=1,\dots,k$.

\subsubsection{Conditional distribution of state variables}
The conditional distribution $p(\bm{z}|\bm{y},\bm{\theta},\bm{\sigma^2},\bm{q})$ can be derived using the contribution of \cite{chib1996calculating}, while carefully handling the robustness parameters $\bm{q}$. We cover high level details from \cite{chib1996calculating} here for our model.  Define $\bm{Z}_{t_i} = (z_{t_0},\dots,z_{t_i})^T$ and $\bm{Z}^{t_{i+1}} = (z_{t_{i+1}},\dots,z_{t_n})^T$, with similar vectors $\bm{Y}_{t_i},\bm{Y}^{t_{i+1}},\bm{Q}_{t_i},\bm{Q}^{t_{i+1}}$ for the observations $\bm{y}$ and robustness parameters $\bm{q}$. Start by factorizing the conditional distribution as follows,
\begin{align*}
p(\bm{z}|\bm{y},\bm{\theta},\bm{\sigma^2},\bm{q}) &= p(z_{t_n}|\bm{y},\bm{\theta},\bm{\sigma^2},\bm{q})p(z_{t_{n-1}}|\bm{Z}^{t_{n}},\bm{y},\bm{\theta},\bm{\sigma^2},\bm{q})\dots\\
&\quad\quad p(z_{t_i}|\bm{Z}^{t_{i+1}},\bm{y},\bm{\theta},\bm{\sigma^2},\bm{q})\dots p(z_{t_0}|\bm{Z}^{t_{1}},\bm{y},\bm{\theta},\bm{\sigma^2},\bm{q})    
\end{align*}

Except for the first $z_{t_n}$ term, these terms take the form $p(z_{t_i}|\bm{Z}^{t_{i+1}},\bm{y},\bm{\theta},\bm{\sigma^2},\bm{q})$. After using Bayes rule and noting conditional independencies from the Markov chain,
\begin{align*}
p(z_{t_i}|\bm{Z}^{t_{i+1}},\bm{y},\bm{\theta},\bm{\sigma^2},\bm{q}) &\propto p(z_{t_i},\bm{Z}^{t_{i+1}},\bm{Y}^{t_{i+1}},\bm{Q}^{t_{i+1}}|\bm{Y}_{t_{i}},\bm{Q}_{t_{i}},\bm{\theta},\bm{\sigma})\\
&\propto p(\bm{Z}^{t_{i+1}},\bm{Y}^{t_{i+1}},\bm{Q}^{t_{i+1}}|z_{t_i},\bm{\theta},\bm{\sigma})p(z_{t_i}|\bm{Y}_{t_{i}},\bm{Q}_{t_{i}},\bm{\theta},\bm{\sigma})\\
&\propto p(\bm{Y}^{t_{i+1}},\bm{Q}^{t_{i+1}}|\bm{Z}^{t_{i+1}},\bm{\theta},\bm{\sigma}) p(\bm{Z}^{t_{i+1}}|z_{t_i},\bm{\theta},\bm{\sigma})p(z_{t_i}|\bm{Y}_{t_{i}},\bm{Q}_{t_{i}},\bm{\theta},\bm{\sigma})\\
&\propto p(z_{t_{i+1}} |z_{t_i}) p(z_{t_i}|\bm{Y}_{t_i},\bm{Q}_{t_i},\bm{\theta},\bm{\sigma^2})    
\end{align*}
The first term is the continuous time transition probability from Theorem \ref{thm:conttime}. Regarding the second term, first note that $p(z_{t_0}=1|\bm{Y}_{t_0},\bm{Q}_{t_0},\bm{\theta},\bm{\sigma^2})=1$ since the prior is a point mass at 1, and thus we can proceed recursively. Assume $p(z_{t_{i-1}}|\bm{Y}_{t_{i-1}},\bm{Q}_{t_{i-1}},\bm{\theta},\bm{\sigma^2})$ is known. We have,
\begin{align*}
p(z_{t_i}|\bm{Y}_{t_i},\bm{Q}_{t_i},\bm{\theta},\bm{\sigma^2}) &\propto p(z_{t_{i}},y_{t_{i}},q_{t_{i}}|\bm{Y}_{t_{i-1}},\bm{Q}_{t_{i-1}})\\
&\propto p(z_{t_i}|\bm{Y}_{t_{i-1}},\bm{Q}_{t_{i-1}},\bm{\theta},\bm{\sigma^2}) p(y_{t_i}|q_{t_i},z_{t_i},\bm{\theta},\bm{\sigma^2})
\end{align*}
Since $p(q_i|z_i,\bm{\theta},\bm{\sigma^2}) = p(q_i)$ which is a constant with respect to $z_{t_i}$.  The first term above can be written as,
\[
p(z_{t_i}|\bm{Y}_{t_{i-1}},\bm{Q}_{t_{i-1}},\bm{\theta},\bm{\sigma^2}) = \sum_{j=1}^k p(z_{t_i}|z_{t_{i-1}}=j) p(z_{t_{i-1}}=j|\bm{Y}_{t_{i-1}},\bm{Q}_{t_{i-1}},\bm{\theta},\bm{\sigma^2})
\]
And the second term is the likelihood distribution for $y_i$.

\end{document}